\documentclass[aps,prl,superscriptaddress,twocolumn,preprintnumbers,floatfix]{revtex4-2}
\pdfoutput=1
\usepackage{mathtools,microtype}
\usepackage{amsthm}
\usepackage{amsfonts,amssymb }
\usepackage{graphicx} 
\usepackage{physics}
\usepackage{xcolor}
\usepackage{dsfont}
\usepackage{bm}
\usepackage{comment}
\usepackage[hidelinks]{hyperref}

\begin{document}
\preprint{MIT-CTP/6114}
\newcommand{\E}[0]{\mathop{{}\mathbb{E}}}
\newcommand{\Pro}[0]{\mathop{{}\mathbb{P}}}

\title{Emergent classicality and wavefunction branching \\ in an  isolated quantum
many-body system}
\author{Sa\'ul Pilatowsky-Cameo}
\email{saulpila@mit.edu}
\affiliation{Center for Theoretical Physics --- a Leinweber Institute, Massachusetts Institute of Technology, Cambridge, MA 02139, USA}

\author{Jordan Cotler}
\email{jcotler@fas.harvard.edu}
\affiliation{Department of Physics, Harvard University, Cambridge, MA 02138, USA}

\author{Daniel Ranard}
\email{dranard@caltech.edu}
\affiliation{Department of Physics, California Institute of Technology, Pasadena, CA 91125, USA}

\author{C. Jess Riedel}
\email{jessriedel@gmail.com}
\affiliation{NTT Research, Inc., Physics \& Informatics Laboratories, Sunnyvale, CA 94085, USA}

\begin{abstract}

Decoherence in quantum systems is conventionally modeled as the effect of interactions with an external environment.  However, such a prescription excludes isolated many-body systems, which are also expected to display classical behavior at macroscopic scales. In isolated systems, decoherence must emerge internally from microscopic degrees of freedom that are invisible to the macroscopic description. Here we explicitly show that classicality can  emerge in such a fashion. We consider a weakly disordered, $3$-local chaotic kicked top of $N$ qubits, where the collective spin sector serves as the macroscopic description, while the microscopic permutation sector acts as an internal bath, decohering the collective spin sector.  Starting from closed unitary dynamics, we derive and numerically confirm an effective Lindblad equation for the collective spin variables.  In the thermodynamic limit these reduced dynamics converge to a classical chaotic Fokker--Planck equation with  vanishingly small diffusion on the spherical phase space, producing a quantum-classical correspondence beyond the Ehrenfest time. The chaotic dynamics evolve the pure many-body wavefunction into continuously branching components associated with distinct classical trajectories. These branches acquire nearly orthogonal microscopic records in the permutation sector, preventing quantum interferences and ensuring the corresponding histories remain consistent.
\end{abstract}

\newtheorem{theorem}{Theorem}

    \newtheorem{corollary}{Corollary}
    \newtheorem{lemma}{Lemma}
    \newtheorem{prop}{Proposition}
    \newtheorem{conjecture}{Conjecture}
\theoremstyle{definition}
  \newtheorem{definition}{Definition}

\maketitle

\begin{figure*}
    \centering \includegraphics[width=0.9\textwidth]{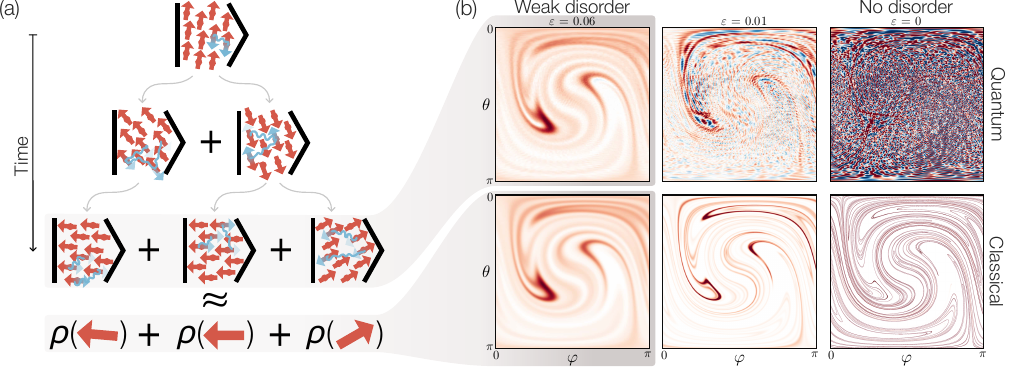}
    \caption{(a) \textit{Emergence of decoherence in an isolated many-body system of spins}. The system is initialized with all spins pointing in the same direction, with the exception of a few particles coupled in singlets (blue). The disorder in the Hamiltonian scrambles the singlets throughout, leading to microscopically distinct branches that cannot interfere in collective observables, so their superposition is indistinguishable from a classical mixture. (b)  \textit{Emergence of classicality}. Wigner function of the quantum evolution in the collective spin sector (top) and classical evolution according to the classical dissipative equation \eqref{eq:FokkerPlanck} (bottom). In the presence of a single realization of a small amount of disorder $\varepsilon$, the quantum and classical evolutions are identical, with no regions of negative probability (blue) ($k=4$, $N=200$, $\kappa=4.2$, $p=\pi/2$, $\theta_0=0.5$, $\varphi_0=1.1$, $t=8$). }
    \label{fig:1}
\end{figure*}

The theory of decoherence~\cite{Zurek2003, JoosZeh1985, Schlosshauer2005, Schlosshauer2019} explains why quantum effects do not manifest at macroscopic scales. A large quantum system inevitably entangles with its environment, which delocalizes phase information, suppresses interferences, and enables effectively classical behavior. Decoherence has been studied from many complementary viewpoints, including idealized bath models in Lindblad dynamics~\cite{Lindblad1976, GKS1976, BreuerPetruccione2002} with analyses of the threshold system-bath coupling necessary for quantum-classical correspondence \cite{toscano2005decoherence,wisniacki2009scaling,hernandez2025ehrenfests,hernandez2024classical,hernandez2025threshold}, a purified closed-system perspective where the system and environment  branch into superpositions of macroscopically distinct states~\cite{Everett1957,Zeh1970,BlumeKohout2006,Wallace2012,Saunders2010}, and the proliferation of redundant information in the environment via quantum Darwinism~\cite{Zurek2009, Ollivier2004, Brandao2015}.
Despite their conceptual differences, these approaches share the same basic premise: decoherence arises because the system couples to a distinct, external environment.

This premise is not fully satisfactory, since it excludes isolated systems, which exist not only as a matter of principle (e.g.,~a closed universe), but also in practice as enabled by modern experimental capabilities~\cite{Trotzky2012,Kaufman2016,Bernien2017}. A complete account of the emergence of classicality should explain how an isolated system can decohere absent an external environment.  
While one may treat an explicit system and bath as a jointly  ``isolated system'', nature does not always equip us with explicit system-bath decompositions; 
instead, one hopes to identify bath variables that emerge from the local interactions within the system.  A natural parallel exists in quantum thermalization~\cite{Deutsch1991,Srednicki1994,Rigol2008,Polkovnikov2011,Gogolin2016,Borgonovi2016}, where an isolated system can reach thermal equilibrium by acting as {\it its own bath}: even though the global state remains pure, subsystems effectively thermalize by becoming highly entangled with their complements. 
We ask if classicality may arise through a similarly intrinsic mechanism, where internal, perhaps hidden degrees of freedom play the role of a decohering environment.

In this work, we present a simple many-body model that exhibits decoherence from its internal degrees of freedom, giving rise to classical behavior without coupling to any external environment. Going beyond previous works in few-body systems~\cite{hillery2005quantum, brun2016decoherence}, we study a system of $N$ qubits evolving under a chaotic (kicked-top) Hamiltonian~\cite{Haake1987,FoxElston1994,Chaudhury2009,Haake2010,Anand2024,Seguralanda2025,Duque2026},  perturbed by weak disorder. The permutation degrees of freedom act as an intrinsic bath, decohering the collective spin sector, and 
 we analytically obtain, and numerically confirm, an effective Lindblad equation for the spin sector in the thermodynamic limit.  In contrast with previous work deriving reduced dynamics in isolated discrete systems~\cite{Ates2012,Ji2013,Nation2020,Mazza2021,Carnazza2022,Cemin2024,odonovan2025quantum}, here the classical variables have extensive support (many-body), are continuous, and possess a classical limit governed by a chaotic Hamiltonian.
This enables us to leverage recent techniques on the classical-quantum correspondence in open systems~\cite{hernandez2025ehrenfests, hernandez2024classical, hernandez2025threshold} to derive the emergence of classical behavior in the isolated model. We remark that, while there is an extensive body of work studying the quantum-classical correspondence in isolated quantum systems~\cite{Haake2010} where classicality is maintained up to the {\it Ehrenfest time}~\cite{Chirikov1988}, the  Ehrenfest time is prohibitively short (scaling only logarithmically with $\hbar_{\rm eff}^{-1}\sim N$). The emergent classicality we describe here is maintained beyond the Ehrenfest time.

Having established quasiclassical dynamics for the collective spin variable, we describe the branching of the overall wavefunction with the consistent histories formalism~\cite{Griffiths1984,Omnes1992,GellMannHartle1990,Halliwell1995,strasberg2023classicality,strasberg2024first,wang2025decoherence}. \emph{Histories} are sequences of projections, forming discrete-time trajectories of the collective variable, to which one would like to assign well-behaved classical probabilities. Histories are \emph{consistent} when their associated conditional states, called  branches~\cite{riedel2017classical,weingarten2022macroscopic,weingarten2025vacuum,taylor2025wavefunction,riedel2025wavefunction},  are orthogonal, and therefore admit probabilities obeying the Kolmogorov axioms.  We explicitly construct the branches in this disordered kicked-top model and numerically verify that the associated histories are indeed consistent. This is enabled by the strong distinguishability of the microscopic arrangement of spins in each branch, which is invisible to collective measurements (see Fig.~\ref{fig:1}). 

{\it Decoherence from microscopic degrees of freedom.}--- 
The Hilbert space of $N$ qubits has dimension $2^N$, but a macroscopic description typically involves only coarse-grained observables. In spin systems, natural examples are collective spin observables, such as the total magnetization in different directions.
 By Schur--Weyl duality~\cite{GoodmanWallach2009,BaconChuangHarrow2006}, the total Hilbert space splits into a direct sum of subspaces $\mathcal{H}_S=\mathcal{S}\otimes\mathcal{P}$, where $\mathcal{S}=\mathcal{S}(S)$ is the collective spin sector (an irreducible representation of $\mathrm{SU}(2)$ of dimension $d_{\mathcal{S}}=2S+1$) on which collective spin operators  $(\hat S_x,\hat S_y,\hat S_z)=\frac{1}{2}\sum_{i=1}^N (\hat X_i,\hat Y_i,\hat Z_i)$ \footnote{$\hat X_i$, $\hat Y_i$, and $\hat Z_i$ are the Pauli operators acting on spin $i$.}   act, and $\mathcal{P}=\mathcal{P}(S)$ is the $\mathrm{SU}(2)$-invariant permutation space of dimension $d_{\mathcal{P}}=\binom{N}{N/2-S}-\binom{N}{N/2-S-1}$. Crucially, outside the fully symmetric sector $S_{\max}=N/2$, the multiplicity space can become large. For example, if $S=N/2 -j$, then $d_{\mathcal P}\sim N^{j}$ while $d_{\mathcal S}\sim N$, so the overwhelming majority of the degrees of freedom are invisible to collective measurements. We will show that these internal degrees of freedom can serve as an effective environment for the collective sector, leading to decoherence. This system-bath decomposition is emergent in the sense that it is not a partition of locally interacting subsystems, but rather arises due to symmetry and hence is delocalized with respect to the microscopic interactions.

For simplicity, we consider an initial state contained entirely in one of the spin subspaces $\mathcal{H}_S$: $N-k$ qubits pointing in the same direction $\ket{\theta_0, \varphi_0}=\cos(\frac{\theta_0}{2})\ket{0}+e^{i\varphi_0}\sin(\frac{\theta_0}{2})\ket{1}$, while the remaining $k$ spins are paired in singlets (for even $k\geq 2$): \begin{equation}
\label{eq:initial_state}
    \ket{\psi(0)}= \ket{\theta_0, \varphi_0}^{\otimes (N-k)}\otimes \Big[\tfrac{1}{\sqrt{2}}(\ket{01}-\ket{10})\Big]^{\otimes k/2}.
\end{equation} If we choose $k$ to be small compared to $N$, then $S=(N-k)/2$ is close to the total spin of the top sector $S_{\max}=N/2$, and the  permutation sector has dimension $d_{\mathcal{P}}\sim N^{k/2}/(k/2)!$, which is polynomially larger than $d_{\mathcal{S}}\sim N$.

{\it Weakly disordered kicked top.}--- We consider a chaotic Hamiltonian describing the dynamics of the spins. For simplicity, we consider the paradigmatic kicked top  
\begin{equation}\label{eq:quantum-kicked-top}
\hat H_{\mathrm{KT}}(t)=\frac{\kappa}{2S+1} \hat S_x^2  + \sum_{n=1}^\infty p\,\delta(t- n)\hat S_z\,,
\end{equation}
but we expect our results to apply to a broader class of chaotic systems. 
Since $\hat H_{\mathrm{KT}}$ is permutation symmetric, it only acts on the spin sector $\mathcal{S}$. To observe decoherence, we will break the permutation symmetry.

A simple way to break this symmetry would be to introduce weak on-site disorder of the form $\varepsilon \sum_i c_i \hat Z_i$, with random coefficients $c_i$. This form of disorder couples $\mathcal{S}$ and $\mathcal{P}$, as it is neither permutation symmetric nor $\mathrm{SU}(2)$ invariant,  but it also couples different total spin subspaces, since it does not commute with the total spin.

For simplicity of the analysis and to greatly decrease the computational cost of simulation, we instead consider the following form of disorder that commutes with the total spin:
\begin{align} 
\label{eq:kickedtop}
    \hat H_{\textrm{dis}}&= \hat K_x \hat S_x + \hat K_y \hat S_y +\hat K_z \hat S_z\,,
    \end{align} 
from which we define the weakly disordered kicked top:
\begin{equation}
\label{equation:model}
    \hat H(t)=\hat H_{\mathrm{KT}}(t) + \varepsilon \hat H_{\mathrm{dis}}.
\end{equation} 
For each $\mu = x,y,z$, we take $\hat K_\mu$ to be an independent random all-to-all Heisenberg coupling, $\hat K_\mu = \sum_{i<j} c_{ij}^{(\mu)} \hat h_{ij}$, where $\hat h_{ij} \coloneqq \hat X_i \hat X_j+\hat Y_i \hat Y_j+\hat Z_i \hat Z_j$. The coefficients $c_{ij}^{(\mu)}$ are Gaussian with variance $\sim (4Nk)^{-1}$, with their mean component subtracted so that $\sum_{i<j} c_{ij}^{(\mu)} = 0$ for each $\mu$. This normalization keeps the spectral width of each $\hat K_\mu$ bounded~\cite{SM}.
 Because the Heisenberg couplings are $\mathrm{SU}(2)$ invariant, the operators $\hat{K}_\mu$ act only in $\mathcal{P}$, while the spin operators $\hat{S}_\mu$ act only in $\mathcal{S}$, so $\hat{H}_{\mathrm{dis}}$ naturally couples $\mathcal{S}$ and $\mathcal{P}$ within the subspace of total spin $S$. The ensemble of random disorder terms $\hat H_{\mathrm{dis}}$ is invariant under physical rotations (isotropic).

\textit{Effective Lindblad dynamics.}--- We begin by showing that $\mathcal{P}$  acts as an effective Markovian bath for the dynamics in $\mathcal{S}$. We first show numerically that the evolution of the reduced density matrix $\psi_{\mathcal{S}}(t)\coloneqq \tr_{\mathcal{P}} \dyad{\psi(t)}$ is well described by a Lindblad equation
\begin{equation}
\label{eq:Lindblad}
            \frac{\partial \rho}{\partial t} = -i [\hat H_{\mathrm{KT}}(t),\rho] +  \gamma \sum_{{\mathclap{\mu=x,y,z}}}\hat{\mathcal{D}}_\mu[\rho]\,,
\end{equation}
with dissipator $\hat{\mathcal{D}}_\mu[\rho]= 
-\frac{1}{2}[\hat S_\mu,[\hat S_\mu,\rho]]$ \footnote{Using that the $\hat S_\mu$ are Hermitian, the dissipator can be brought to the more familiar form $\hat{\mathcal{D}}_\mu[\rho]=\hat S_\mu \rho \hat S_\mu^\dagger -\frac{1}{2}( \hat S_\mu^\dagger \hat S_\mu \rho  + \rho \hat S_\mu^\dagger \hat S_\mu)$.}.
\begin{figure}[]
    \centering
    \includegraphics[width=1\linewidth]{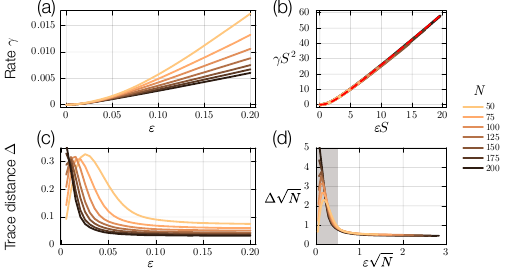}
    \caption{Effective Lindblad dynamics. (a) Optimized rate $\gamma$, obtained by minimizing the average trace distance $\Delta$ to the exact quantum dynamics. (b)   $S^2\gamma$ against $\varepsilon S$:  empirical rates $\gamma$ collapse to the predicted  $S^2\gamma=\varepsilon SJ_*(\varepsilon S)$ (dashed). (c) Time-averaged trace distance for optimal $\gamma$, for time $t\in [2,8]$. (d) Collapse of the trace-distance curves obtained by multiplying $\varepsilon$ and $\Delta$ by $\sqrt{N}$. Gray: Non-Lindblad regime (large $\Delta$). All plots are averaged over 20 disorder realizations; the fitting procedure is described in the SM~\cite{SM} (Fig.~S3) ($k=4$, $\kappa=4.2$, $p=\pi/2$, $\theta_0=0.5$, $\varphi_0=1.1$).} 
    \label{fig:2}
\end{figure}

In Fig.~\ref{fig:2} we show an empirical value of $\gamma$ obtained numerically, by minimizing the finite-time-averaged trace distance $\Delta=\frac{1}{2}\E_t\norm{\psi_\mathcal{S}(t) - \rho(t)}_1$ between the actual evolution $\psi_\mathcal{S}(t)$ and the Lindblad approximation $\rho(t)$. For disorder strength $\varepsilon \gtrsim 1/\sqrt{N}$, the reduced dynamics agree well with the Lindblad approximation: the error $\Delta$ decays as $\sim 1/\sqrt{N}$, and the fitted rate satisfies the empirical scaling $\gamma \propto \varepsilon/N$. Via a heuristic derivation of $\gamma$ and obtain~\cite{SM}
\begin{equation}
\label{eq:formofgamma}
    \gamma=\frac{\varepsilon}{S}J_*(\varepsilon S)\sim \frac{\varepsilon}{S}\,{O\!}\left( \log(\varepsilon S)^\alpha\right),
\end{equation}
where $J_*$ is the finite-time integral of a disorder-averaged effective bath correlation function at infinite temperature, evaluated up to a time proportional to $\varepsilon S$. Fig.~\ref{fig:2}(b) shows excellent agreement between the empirical rate $\gamma$ and our derived Eq.~\eqref{eq:formofgamma}.

For large $\tau \coloneqq \varepsilon S$, we numerically see that the correlator has a logarithmic-type divergence  $J_*(\tau)=O(\log(\tau)^\alpha)$  for a certain  $0<\alpha\leq1$. This logarithmic factor is characteristic of choosing the bath operators $\hat K_\mu$ to be random $2$-local Heisenberg couplings; in this construction, $\hat H_{\rm dis}$ is then $3$-local in the microscopic spins. The logarithm is absent for analogous disorder ensembles built from higher-locality operators, such as $4$-local disorder or GOE random matrices~\cite{SM}. Nevertheless, this logarithmic growth does not obstruct Markovian behavior, since $\tau$ is independent of the physical evolution time $t$.

{\it Emergent classicality.}--- The effective Lindblad dynamics lead to the emergence of classical equations of motion for the collective spin observables. Specifically, consider the classical diffusion equation obtained by treating the collective spin operators as if they were continuous commuting variables on the sphere,  $\hat{\bm S}=(\hat S_x,\hat S_y,\hat S_z) \to {\bm S}=( S_x, S_y, S_z)\in \mathbb{R}^3$, and replacing the commutator with a Poisson bracket, $-i[\,\cdot\,,\,\cdot\,] \;\longrightarrow\; \{\,\cdot\,,\,\cdot\,\}$ \footnote{ For spherical variables, the Poisson bracket is defined via the triple product $\{f,\,g\} = \mathbf{S}\cdot\!\left(\frac{\partial f}{\partial \mathbf{S}}\times\frac{\partial g}{\partial \mathbf{S}}\right)$.}, namely
\begin{equation}
\label{eq:FokkerPlanck}
    \frac{\partial f}{\partial t} = \{H_{\mathrm{KT}}(t),f\} + \gamma \sum_{\mathclap{\mu=x,y,z}}\,\mathcal{D}_\mu[f]\,.
\end{equation}
Here, $f(t,\bm S)\geq 0$ represents a classical distribution on the sphere, while the classical Hamiltonian  is given by~\footnote{The replacement $(2S+1)\to 2S$ provides more accurate classical correspondence. The underlying technical reason is the non-multiplicativity of the Weyl transform for spin variables: $\mathcal{W}(\hat S_x^2)/(2S+1)\approx \mathcal{W}(\hat S_x)^2/2S$~\cite{Varilly1989}.}
\begin{equation}
    H_{\mathrm{KT}}(t) = \frac{\kappa }{2S}S_x^2 \;+\; \smash{\sum_{n=1}^{\infty}}\, p\, S_z\delta(t-n)
\end{equation}
 and the classical dissipator by $
    {\mathcal{D}}_\mu[ f]= 
\frac{1}{2}\{ S_\mu,\{ S_\mu, f\}\}$.

In Fig.~\ref{fig:1}(b), we compare the classical dynamics of Eq.~\eqref{eq:FokkerPlanck} with the Wigner function of $\ket{\psi}_{\mathcal{S}}$~\cite{Varilly1989}. In the absence of disorder the quantum dynamics develop regions of negative probability, a signature of non-classicality~\cite{Kenfack2004}. However, introducing only a small amount of disorder [$\varepsilon=0.06$ in Fig.~\ref{fig:1}(b)]  in Eq.~\eqref{equation:model} leads to a positive Wigner function, with perfect quantum-classical correspondence. See the Supplemental Material (SM)~\cite{SM} for additional figures. 

This emergence of classicality can be asymptotically established rigorously, starting from the Lindblad dynamics. Under  sufficiently strong diffusion (but still vanishing in the thermodynamic limit), the dynamics remain in a mixture of coherent states in $\mathcal{S}$, associated with an asymptotically positive  phase-space distribution.

\begin{theorem}
\label{th:01}
If $\gamma\geq \gamma_*\coloneqq \frac{\abs{\kappa}}{2S+1}$, under the Lindblad equation~\eqref{eq:Lindblad}, starting in a spin coherent state, the dynamics remain a positive mixture of coherent states, 
\begin{equation}
    \rho(t)=\int \dd\Omega \,P(\Omega,t) \dyad{\Omega}, \qquad P(\Omega,t)\geq 0,
\end{equation}
where the integral is taken over the sphere.
Furthermore, the differential equation governing the evolution of $P$ matches Eq.~\eqref{eq:FokkerPlanck} up to a relative error $O(1/S)$ in the Hamiltonian terms and $O(\gamma_*/\gamma)$ in the dissipator.
\end{theorem}
\begin{proof} 
The function $P(\Omega,t)$ is simply the Glauber--Sudarshan quasiprobability distribution for the spin \cite{Arecchi1972}. Any state can be represented in this form, but $P$ in general might not be positive. We show that, under sufficient diffusion, positivity is guaranteed. 

The evolution of $\rho$ can be described by the evolution of $P$ via the substitutions~\cite{SM}
\begin{align}
    -i[\hat S_z,\rho]
    &\to \{S_z,P\},\qquad
    -[\hat S_\mu,[\hat S_\mu,\rho]]
    \to\{S_\mu,\{S_\mu,P\}\} \nonumber
    \\
\label{eq:dict-square-comm}
    -i[\hat S_x^2,\rho]
    &\to
        \{S_x^2,P\}
        +     \frac{1}{S}\sum_{\nu,\lambda}
    \epsilon_{x\nu\lambda}
        \{S_x,\{S_\nu,S_\lambda P\}\},
\end{align}
where $\epsilon_{x\nu\lambda}$ is the Levi-Civita symbol. Equation~\eqref{eq:Lindblad} is recast into a differential equation resembling Eq.~\eqref{eq:FokkerPlanck}, 
\begin{equation}
\label{eq:actualPevol}
    \frac{\partial P}{\partial t} =    \{ \widetilde{H}_{\mathrm{KT}}(t), P\} + \frac{1}{2}\sum_{\mu,\nu} \Lambda_{\mu\nu}\{ S_\mu ,\{ S_\nu,P\}\},
\end{equation}
but with a slightly modified classical Hamiltonian $\widetilde{H}_{\mathrm{KT}}(t)={H}_{\mathrm{KT}}(t) + \frac{3c}{4}  S_x^2$   and diffusion matrix $\Lambda_{\mu\nu} = \gamma \delta_{\mu\nu} + c\sum_\lambda ( \delta_{x\mu } \epsilon_{x\nu \lambda} + \delta_{x\nu} \epsilon_{x\mu \lambda} ) S_\lambda$, where $c=\frac{\kappa}{S(2S+1)}$. 
Equation~\eqref{eq:actualPevol} exactly describes the quantum evolution, where the distribution $P$ might become negative and singular due to the negative eigenvalues of the diffusion matrix $\Lambda_{\mu\nu}$.  However, if $\gamma \geq \abs{c}S=\gamma_*$, then $\Lambda_{\mu\nu}$ is positive semidefinite,  turning Eq.~\eqref{eq:actualPevol} into a classical diffusive equation which evolves $P$ into a positive and nonsingular distribution. Furthermore, the relative error in the coefficients for $S_x^2$ in Eqs.~\eqref{eq:FokkerPlanck} and~\eqref{eq:actualPevol} is $(3c/4)/(\kappa/2S)=O(1/S)$ and the relative error between $\Lambda_{\mu\nu}$ and the isotropic diffusion matrix is  $O(\norm{(\Lambda_{\mu\nu}-\gamma\delta_{\mu\nu})_{\mu\nu}}/\gamma)=O(\gamma_*/\gamma)$.
\end{proof}

Theorem \ref{th:01} guarantees classical evolution whenever $\gamma\gg O(S^{-1})$. Furthermore, in the SM~\cite{SM} we present an improved theorem (which also holds for general Hamiltonians) where the threshold for classicality is $\gamma\gg O(S^{-4/3})$. This improved threshold is enabled by allowing slightly squeezed states in the decomposition of $\rho$~\cite{hernandez2025ehrenfests,hernandez2024classical}, and is expected to be optimal~\cite{hernandez2025threshold}. Combined with Eq.~\eqref{eq:formofgamma}, it implies that, for sufficiently large $S$, our disordered kicked-top model displays classical behavior as long as $\varepsilon\geq  O(N^{-1/3+\delta})$ for any constant $\delta>0$.

{\it Quantum branching.}--- The emergence of classicality in closed systems is believed to be described via the theory of wavefunction branching~\cite{Wallace2012} and consistent histories~\cite{Griffiths1984,Omnes1992,GellMannHartle1990,Halliwell1995}, which we introduce next. We then numerically confirm that these phenomena indeed arise in our model. 

We may describe the state of the system as a superposition of different {\it branches},
\begin{align}
\label{eq:branchessuperposition}
\ket{\psi(t_M)}&=\int\dd \bm \Omega \,c_{\bm \Omega} \ket{\phi_{\bm \Omega}}.
\end{align}
Each branch is labeled by a {\it history},
${\bm \Omega}=(\Omega_0,\dots, \Omega_{M})$, a sequence of points on the sphere $\Omega_m=(\theta_m,\varphi_m)$, at times $t_0=0,t_1,\dots,t_M$, where the integral is taken with respect to the measure $\dd{\bm \Omega}=\dd\Omega_1\cdots\dd\Omega_M$, with $\dd\Omega_m=\tfrac{2S+1}{4\pi}\sin\theta_m\,\dd\theta_m\,\dd\varphi_m$. 
These branches are defined by inserting spin coherent-state projectors  $\hat{\Omega}_m = \dyad{\Omega_m}_{\mathcal{S}} \otimes \mathds{1}_{\mathcal{P}}$ at the corresponding times:
\begin{equation}
        c_{\bm \Omega}\ket{\phi_{\bm \Omega}} = \hat{\Omega}_M U_M \cdots \hat{\Omega}_{2} U_{2} \hat{\Omega}_1 U_1 \ket{\psi(0)} ,
\end{equation}
where the initial state is given by Eq.~\eqref{eq:initial_state}, with $\Omega_0 = (\theta_0,\varphi_0)$. The coefficient $c_{\bm \Omega}$ is chosen so that $\ket{\phi_{\bm \Omega}}$ is normalized, and $U_m = \mathcal{T}\exp\!\left(-i \int_{t_{m-1}}^{t_m} \dd t\, \hat H(t)\right)$ is the unitary generated by Eq.~\eqref{equation:model} between times $t_{m-1}$ and $t_m$. Equation~\eqref{eq:branchessuperposition} follows formally from the coherent-state resolution of the identity $\int\dd \Omega_m \hat{\Omega}_m=\mathds{1}$. In general, however, these branches may interfere, so the decomposition alone does not yet define classical alternatives.

\begin{figure}
    \centering
    \includegraphics[width=1\linewidth]{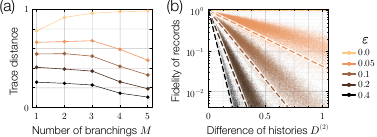}
    \caption{ (a) \textit{Quantum branching}. Trace distance between the reduced state of the global unitary dynamics and the mixture of branches [left and right sides of Eq.~\eqref{eq:equalityofobservables}, respectively], for different values of disorder $\varepsilon$ and number of branchings $M$. (b) \textit{Consistent histories}. Fidelity of records vs cumulative squared distance between histories $D^{(2)}$ after $M=5$ branchings. Each dot is a pair of independently sampled branches. Dashed lines mark the decay $e^{-\Gamma D^{(2)}}$. ($N=100$,  $k=4$, $\kappa=4.2$, $p=\pi/2$, $\theta_0=0.5$, $\varphi_0=1.1$, $\Delta t=1$) }
    \label{fig:3}
\end{figure}
The emergence of decoherence is a manifestation of the noninterference of branches in the collective observables. The reduced state in $\mathcal{S}$ becomes indistinguishable from a classical mixture over all branches:
\begin{equation}
\label{eq:equalityofobservables}
   \tr_{\mathcal{P}}({\psi(t_M)})\approx  \int \dd \bm \Omega \, |c_{\bm \Omega}|^2 \tr_{\mathcal{P}}(\dyad{\phi_{\bm \Omega}}),
\end{equation}
where $\tr_{\mathcal{P}}(\dyad{\phi_{\bm \Omega}})=\dyad{\Omega_M}$ is a coherent state. Indeed, numerics shown in Fig.~\ref{fig:3}(a) confirm that the trace distance between the two sides of Eq.~\eqref{eq:equalityofobservables} decreases with increasing $\varepsilon$. In fact, a stronger notion of noninterference is satisfied in our model, namely consistency of histories. 

Each branch has the form $\ket{\phi_{\bm \Omega}}=\ket{{\Omega_M}}_\mathcal{S}\otimes \ket{R_{\bm \Omega}}_\mathcal{P}$, and the state $\ket{R_{\bm \Omega}}$ is known as a {\it record}. In the presence of disorder, separated histories imprint distinguishable records in the permutation sector. Specifically, we estimate via a semiclassical Lindblad approximation \cite{SM} that the fidelity of records for two histories $\bm \Omega$ and $\bm \Theta$ should decay exponentially in $D^{(2)}_{\bm \Omega,\bm \Theta}\approx\frac{1}{2}\smash{\sum_{m=1}^M} \big[d(\Omega_{m-1},\Theta_{m-1})^2+d(\Omega_{m},\Theta_{m})^2\big]$ \footnote{At $m=1$, the term $d(\Omega_{m-1},\Theta_{m-1})$ vanishes since $\Omega_0=\Theta_0$ corresponds to the initial state.}. Namely,
\begin{equation}
\label{eq:fidelity_of_records}
     \abs{\braket{R_{\bm \Theta}}{R_{\bm \Omega}}}^2\sim \exp(-\Gamma D^{(2)}_{\bm \Omega,\bm \Theta}),
\end{equation}
 with a coefficient $\Gamma=\gamma S^2 \Delta t/(1+\gamma S\Delta t)$, where $\Delta t=t_{m}-t_{m-1}$. This is supported by the numerics  shown in Fig.~\ref{fig:3}(b). In the SM~\cite{SM} we present additional numerical verifications for random-matrix couplings. We remark that we can study these  overlaps because our system is closed, so we can canonically identify the record states of the bath. Such an identification would not be possible if one instead had postulated an open model governed by a Lindblad equation as the starting point.

The decay in the fidelity of distinct records prevents distinct branches from interfering in any collective observable $O_{\mathcal{S}}$. Interferences are controlled by the off-diagonal terms
\begin{equation}
\label{eq:off-diag-term}
    \bra{\phi_{\bm \Theta}}{O_\mathcal{S}} \otimes \mathds{1}_{\mathcal{P}}\ket{\phi_{\bm \Omega}}=\bra{\Theta_M}{O_\mathcal{S}}\ket{\Omega_M}\braket{R_{\bm \Omega}}{R_{\bm \Theta}}.
\end{equation}
 In the absence of disorder, $\gamma=0$ and hence $\braket{R_{\bm \Theta}}{R_{\bm \Omega}}=1$  regardless of the histories, thus allowing interference even when $\bm \Omega$ and $\bm \Theta $ are distinct. As $\gamma$ increases, the off-diagonal matrix element in Eq.~\eqref{eq:off-diag-term} is suppressed exponentially in $D_{\bm \Theta,\bm \Omega}^{(2)}$. This effect is analogous to the recording of which-path information in the double-slit experiment. If no degree of freedom records which slit the particle passed through, the two branches remain coherent and can interfere at the detector. If such a record is created, their coherence is suppressed and the interference disappears. In our model, the permutation sector creates the analogous record by  becoming correlated with the collective-spin history due to the disordered interaction.

{\it Discussion and outlook.}---
We have exhibited a concrete, numerically tractable isolated many-body model in which internal degrees of freedom act as an effective bath. For an emergent collective spin variable with chaotic classical dynamics, this leads to Lindblad dynamics, classical Fokker--Planck evolution, and consistent branching.  This model is deliberately idealized. Both the kicked-top dynamics and the disorder are $3$-local but all-to-all, whereas physical interactions are expected to be geometrically local. A natural next step is therefore to understand whether the same mechanism persists in geometrically local systems, where the quasiclassical variables are likely to be spatially coarse-grained densities rather than global collective spins, and whose classical equations of motion are not simply obtained by removing the hats from the quantum operators.

Additionally, our model exploits Schur--Weyl duality by having a disorder term that preserves total spin, constraining the dynamics to take place within a fixed spin sector, and ensuring a simple tensor-product decomposition $\mathcal{H}_S=\mathcal{S}\otimes\mathcal{P}$ between collective spin variables $\mathcal{S}$ and microscopic permutation variables $\mathcal{P}$.
This enabled numerical simulation of a system with an otherwise inaccessibly large total spin.
In contrast, a generic local perturbation would mix sectors, so that the quasiclassical variable (whether or not a collective spin) could take the form of a more general subalgebra $\mathcal{A}=\bigoplus_i \mathcal{S}_i \otimes\mathcal{P}_i$ \cite{zanardi2004quantum}.
Understanding this sector-wise form of emergent system-bath structure is an important step toward extending the present mechanism to more general and realistic settings, such as fermionic models, bosonic systems with unbounded local Hilbert spaces, and quantum field theories. 
A broader goal is to find a general framework describing the universal emergence of classical behavior in isolated quantum many-body systems, where wavefunction branches emerge naturally from the underlying quantum dynamics~\cite{riedel2017classical,weingarten2022macroscopic,weingarten2025vacuum,taylor2025wavefunction,riedel2025wavefunction}.
\\ \\
\textit{Acknowledgements.}--- We thank Curt von Keyserlingk for extensive insightful conversations and Soonwon Choi for a suggestion that informed our choice of the disorder term in Eq.~\eqref{eq:kickedtop}. We acknowledge the MIT SuperCloud and Lincoln Laboratory Supercomputing Center for providing computational resources utilized for our numerical analysis. We also acknowledge the use of GPT~5.5-6 for aid in generating code and analytical derivations, in particular the identification of Eqs.~(S12) and (S29) in the SM~\cite{SM}. This work is supported by the U.S. Department of Energy, Office of Science, under Award Number DE-SC0021013 titled `Emergent Phenomena in Quantum Dynamics: From Chaos to Spacetime'.  JC is supported by a fellowship from the Alfred P.~Sloan Foundation.

\bibliography{references}

\end{document}


\def\bibsection{\section*{\refname}} 
\newcommand{\E}[0]{\mathop{{}\mathbb{E}}}
\newcommand{\Pro}[0]{\mathop{{}\mathbb{P}}}

\title{Supplemental Material: Emergent classicality and wavefunction branching \\ in an isolated quantum
many-body system}

\author{Sa\'ul Pilatowsky-Cameo}
\email{saulpila@mit.edu}
\affiliation{Center for Theoretical Physics --- a Leinweber Institute, Massachusetts Institute of Technology, Cambridge, MA 02139, USA}

\author{Jordan Cotler}
\email{jcotler@fas.harvard.edu}
\affiliation{Department of Physics, Harvard University, Cambridge, MA 02138, USA}

\author{Daniel Ranard}
\email{dranard@caltech.edu}
\affiliation{Department of Physics, California Institute of Technology, Pasadena, CA 91125, USA}

\author{C. Jess Riedel}
\email{jessriedel@gmail.com}
\affiliation{NTT Research, Inc., Physics \& Informatics Laboratories, Sunnyvale, CA 94085, USA}

\newcommand{\supp}[0]{\mathrm{supp}}

\newtheorem{theorem}{Theorem}
    \newtheorem{claim}{Claim}

    \newtheorem{corollary}{Corollary}
    \newtheorem{lemma}{Lemma}
    \newtheorem{prop}{Proposition}
    \newtheorem{conjecture}{Conjecture}
\theoremstyle{definition}
  \newtheorem{definition}{Definition}
\def\theequation{S\arabic{equation}}
\renewcommand{\thecorollary}{S\arabic{corollary}}
\renewcommand{\thelemma}{S\arabic{lemma}}
\renewcommand{\theprop}{S\arabic{prop}}
\renewcommand{\thedefinition}{S\arabic{definition}}
\renewcommand{\theconjecture}{S\arabic{conjecture}}
\renewcommand{\thetheorem}{S\arabic{theorem}}
\renewcommand{\thefigure}{S\arabic{figure}}
\maketitle
\vspace{-1.2em}
{\small
In this Supplemental Material, we provide details supporting the derivations, rigorous statements, and numerical results of the main text. In Sec.~\ref{sec:priorwork}, we identify some key features of the emergence of classicality in our model, contrast it with related prior work, and present future directions. In Sec.~\ref{sec:numericaloptimization}, we describe the  numerical techniques used to simulate the unitary dynamics and present a time-resolved comparison between the reduced Wigner function and the corresponding classical dissipative evolution. In Sec.~\ref{sec:momsofRHC}, we compute normalization factors and moment estimates for the random coupling ensembles appearing in the disordered Hamiltonian. In Sec.~\ref{sec:analytic-rate}, we give a microscopic derivation of the effective Lindblad equation for the collective spin sector and of the corresponding diffusion rate. In Sec.~\ref{sec:proofofth1}, we present the detailed proof of Theorem~1, relating the spin Lindblad dynamics to a classical Fokker--Planck equation. In Sec.~\ref{sec:NTS}, we extend the classicality criterion using not-too-squeezed states, obtaining  an improved threshold over Theorem 1, which holds for general Hamiltonians. Finally, in Sec.~\ref{sec:fidelity}, we estimate the fidelity of records associated with different branches and compare this prediction with numerical simulations.
\vspace{-1.2em}
\tableofcontents}
\newpage

\section{Model features and comparison to prior works}
\label{sec:priorwork}

In this section, we introduce some terminology and then use it to characterize the main features of our work and compare our approach with related work in the literature.

\subsection{General terminology}

We take as given a notion of locality defined by a preferred tensor-product decomposition of a finite-dimensional Hilbert space into microscopic subsystems: $\mathcal{H} \equiv \bigotimes_j \mathcal{H}_j$.
We say an observable is $k$-\emph{local} when it is supported on at most $k$ subsystems, while a family of observables is \emph{many-body} if its support diverges in the thermodynamic limit.  Likewise, the dynamics are $k$-local when the Hamiltonian is a sum of $k$-local observables. 

A set of observables induces an algebra through operator multiplication and complex linear combinations. Generalizing from the identification between subsystems (in the sense of a tensor-product decomposition of Hilbert space) and the algebra of operators defined on them~\cite{zanardi2004quantum}, we can think of any algebra as a \emph{generalized subsystem}. Unlike the microscopic subsystems defining locality, a generalized subsystem need not correspond to a tensor-product factor of $\mathcal H$. Its algebra $\mathcal{A}$ may have a nontrivial center and, more generally, takes the Artin-Wedderburn form $\mathcal{A} = \bigoplus_i \mathcal{A}_i \otimes \mathcal{I}_i \subset \mathcal{B}(\mathcal{H})$.

The \emph{environment} of a generalized subsystem is the commutant (i.e., centralizer) of its algebra.  We say the environment is \emph{internal} when the degrees of freedom in the commutant are supported on the same microscopic subsystems as the generalized subsystem itself. If the environment is instead supported on the complementary microscopic subsystems, we say it is \emph{external}. The former occurs, for instance, when the generalized subsystem has support on all microscopic subsystems but its algebra is a proper subalgebra of the full operator algebra, while the latter is the familiar case in which the generalized subsystem is the complete operator algebra on a proper subset of microscopic subsystems.

We now turn to concepts that become applicable in the thermodynamic limit $N\to\infty$, which with appropriate scaling coincides with the classical limit $\hbar_{\mathrm{eff}}\to 0$.
The reduced density matrix $\rho(t)$ of a generalized subsystem at a given time $t$ is \emph{localized} in a set of observables $\hat{\mathbf{A}} = (\hat{A}_r)$ when it is a 
mixture, with some distribution $P(\mathbf{A},\mu;t)$,
of states 
$|\mathbf{A},\mu\rangle$ 
that are each spectrally localized\footnote{Although we do not commit to a precise quantification of approximate localization for use in a limit, for concreteness one could require that $\|\rho-\tilde{\rho}\|_{\mathrm{Tr}}$ is vanishing, where 
	$\tilde{\rho}=\int \! \dd \mathbf{A}\, \dd\mu\, P(\mathbf{A},\mu) |\mathbf{A},\mu\rangle\langle \mathbf{A},\mu|$ and where each state $|(A_r),\mu\rangle$ 
	is supported strictly on an interval of eigenvalues (for each $A_r$) whose width is vanishing.} 
around the eigenvalue $A_r$ for each of the observables $\hat{\mathbf{A}}=(\hat{A}_r)$.  
Marginalizing over the auxiliary\footnote{In our main model, we can take $\mu$ to be trivial, but in our more general theorem about spin decoherence its role is played by the approximate tangent covariance matrix $\sigma$ of the spin-squeezed states, which controls the squeezing strength and orientation.} parameter $\mu$ yields a classical distribution $P(\mathbf{A},t) = \int \!\dd\mu\, P(\mathbf{A},\mu;t)$ for the localized variable $\mathbf{A}$.  

Due to the spectral localization of the $|\mathbf{A},\mu\rangle$, the marginals of $P(\mathbf{A},t)$ converge in the thermodynamic limit to the corresponding single-time measurement distributions of the observables $\hat{\mathbf{A}}=(\hat{A}_r)$. This, however, is only a statement about single-time statistics. It does not guarantee that measurements at different times have classical joint statistics; the underlying quantum dynamics need not be Markovian even when the classical distribution $P(\mathbf{A},t)$ of localization obeys a Markovian equation. Defining quantum Markovianity is challenging \cite{li2018concepts, milz2019completely}, even within the restricted context of the quantum-classical transition \cite{strasberg2023classicality}. Nonetheless, the consistent histories framework can provide a stronger characterization of classicality than the behavior of $P(\mathbf{A},t)$ alone.

Given a set of exhaustive histories defined through time-ordered sequences of Heisenberg-picture projectors, $C_\alpha = Q_{\alpha_M}(t_M)\cdots\allowbreak  Q_{\alpha_1}(t_1)$, the histories are \emph{consistent} (precisely, \emph{medium decoherent}) when $\langle \psi |C_\alpha^\dagger C_\beta |\psi\rangle =0$ for $\alpha\neq\beta$.  This implies the probabilities $p(\alpha) = \langle \psi |C_\alpha^\dagger C_\alpha |\psi\rangle$ satisfy the Kolmogorov probability sum rules under coarse-graining.
In the thermodynamic limit, we say that the histories are \emph{trajectories} for a set of observables $\hat{\mathbf{A}}$ when (i) the time intervals between successive projectors vanish, and (ii) each projector is supported on a vanishing-width interval of eigenvalues of the observables $\hat{\mathbf{A}}$.  A consistent set of trajectory histories for $\hat{\mathbf{A}}$ is \emph{quasiclassical} when the probabilities satisfy the Markovian condition $p(\alpha_1,\ldots,\alpha_M)p(\alpha_m) = p(\alpha_1,\ldots,\alpha_m)p(\alpha_m,\ldots,\alpha_M)$.
This stronger notion of Markovianity guarantees, for instance, that multiple hypothetical measurements of the observables $\hat{\mathbf{A}}$ at different times by an external observer would have the expected classical correlations (with the single-time measurements as a special case converging to the distribution $P(\mathbf{A};t)$ associated with localization). 

When the Heisenberg-picture conditional state associated with each history at one time overlaps only a single conditional state associated with a history at an earlier time, the conditional states form a directed tree, with edges corresponding to nonzero overlap. We refer to the conditional states as \emph{branches} and say that the histories \emph{branch forward in time}.

\subsection{Model characterization and future directions}

With the preceding definitions, we can now compactly characterize our primary model at appropriate disorder strength. We heuristically derive a Lindblad equation for a generalized subsystem of a many-qubit system in the thermodynamic limit. Our model has these features:
\begin{enumerate}[label= \Alph*.]
	\item The dynamics are unitary and the Hamiltonian is $3$-local.
	\item There is a preferred variable (a set of observables, $\hat{\bm S}/S$) that generates the algebra of the generalized subsystem, whose environment is internal.
	\item In the thermodynamic limit,
	\begin{enumerate}[label= \roman*)]
		\item the preferred variable is many-body and continuous;
        \item the preferred variable is localized by the environment, and its classical distribution obeys a Kolmogorov forward equation, specifically a Fokker-Planck equation;
		\item the drift terms in the Fokker-Planck equation produce chaotic Hamiltonian flow; 
		\item the dissipative terms in the Fokker-Planck equation are negligible\footnote{See Refs.~\cite{zurek1998decoherence,hernandez2025ehrenfests}.} compared with the drift terms; and
		\item trajectory histories of the preferred variable are consistent, quasiclassical, branch forward in time, and marginalize to the same classical distribution as in (ii).
	\end{enumerate}

\end{enumerate}
Because we can numerically simulate the full unitary dynamics, rather than only simulating the reduced dynamics under an assumed Lindblad equation, we can \emph{confirm} the accuracy of the Lindblad description and directly check the consistency of the histories.

At the same time, our model has three simplifying features that leave important challenges for future work:
\begin{enumerate}
    \item The Hamiltonian is 3-local and thus contains genuine interactions, but it is not geometrically local.
	\item Our model constrains the dynamics within fixed total-spin subspaces, although these spaces are still exponentially large in the number of qubits when $k=\Theta(N)$.  Within each subspace the coupling takes the familiar finite-rank bilinear form, although the mapping to the $k$-local representation of the Hamiltonian is still nontrivial.
	\item The dynamics of the quasiclassical variable can be read off rather directly from the Hamiltonian $\hat{H}$. In the thermodynamic limit, the coupling $\epsilon \hat{H}_{\mathrm{dis}}$ becomes small relative to $\hat{H}_{KT}$, so that both $\hat{H}$ and $i[\hat{H},\hat{\mathbf{S}}]$ approach low-order polynomials in the quasiclassical variables $\hat{\mathbf{S}}=(\hat{S}_x,\hat{S}_y,\hat{S}_z)$. Thus the natural representation of the microscopic Hamiltonian already exposes the effective classical variables and their dynamics.
\end{enumerate}
A corresponding challenge for future work is to find a model that can be shown to have properties B and C above, but for which
\begin{enumerate}
	\item the Hamiltonian $\hat{H}$ is geometrically local, 
	\item the dynamics are not artificially constrained to a subspace of the Hilbert space, and/or
	\item neither the Hamiltonian $\hat{H}$, nor its commutator $[\hat{H},\hat{\mathbf{A}}]$ with the preferred observables $\hat{\mathbf{A}} = (\hat{A}_r)$,
    is a bounded-order polynomial in $\hat{\mathbf{A}}$, even in the thermodynamic limit.
\end{enumerate}
Although the conditions above fall short of a general definition of wavefunction branching \cite{riedel2025wavefunction}, we expect that identifying additional models exhibiting these phenomena, especially models without the simplifying features just described, will help clarify what such a definition should be.

\subsection{Relationship to prior work}

We now briefly comment on a selection of prior work and note how it differs from the results we present.

O’Donovan et al.\ carry out a traditional microscopic derivation of a Lindblad equation in the weak-coupling limit for a general system coupled to an external ETH-obeying bath through local operators \cite{odonovan2025quantum}. 
They numerically compare an ETH-obeying mixed-field Ising bath with an integrable version, finding that the former is well described by the derived Lindblad equation while the latter shows strong discrepancies at the finite sizes studied. 
The generality of the arguments is suggestive for identifying which general many-body variables could be expected to decohere.

Nation \& Porras~\cite{Nation2020}
also consider systems coupled to an external bath where the Hamiltonian eigenstates are taken to have the chaotic wavefunction form and the system-bath coupling is a random matrix.  To describe the emergence of classical statistical mechanics, 
they obtain dynamics corresponding to an Ornstein–Uhlenbeck process for a solvable random-matrix model, and they extend the underlying thermalization and Markovian analysis numerically to several local observables on the spin chain and one global observable (the mean spin).
They further consider sequences of measurements at multiple times and note that the outcomes satisfy a non-disturbance condition that they identify with consistency of histories.

Working on a spin chain, Mazza et al.~\cite{Mazza2021}
use machine-learning methods to find general linear generators that accurately produce the evolution of 1-local subsystems. 
Carnazza et al.~\cite{Carnazza2022}
and Cemin et al.~\cite{Cemin2024}
extend this to 2-spin subsystems and constrain the generators to be Lindbladian. These discrete variables are not many-body, and their environment is external rather than internal. The Lindblad dynamics are fit rather than derived, and they govern variables with no classical limit. 

Also on a spin chain, 
Ates et al.~\cite{Ates2012}
derive a classical Fokker–Planck equation describing non-chaotic relaxation for a continuous excitation density, a many-body variable. In related work, Ji et al.~\cite{Ji2013}
derive an analogous master equation for the discrete number of hard rods.
They do not investigate classical correlations between multiple times.

Wang and Strasberg \cite{wang2025decoherence}
study a closed many-body model of multi-time correlations: they consider  histories of a variable with divergent support in a geometrically local XXZ Heisenberg spin chain and show that they generally become consistent in the thermodynamic limit.
However, they restrict their attention to histories over a binary state space (whether the magnetization imbalance is positive or negative). They also only investigate the decoherence functional of histories for Haar-random initial states,  and thus do not provide a special low-entropy initial condition that would define the forward branching arrow of time.

\section{Numerical techniques for unitary evolution}
\label{sec:numericaloptimization}

In this section we describe the numerical methods used to simulate the
unitary dynamics under the disordered kicked-top model.  Our simulations are performed directly in the
fixed-total-spin subspace with
$$
    S=\frac{N-k}{2},
    \qquad
    \mathcal H_S=\mathcal S\otimes \mathcal P,
$$
where $\dim(\mathcal S)=2S+1$ and
$\dim(\mathcal P)=\binom{N}{k/2}-\binom{N}{k/2-1}$.  This avoids working in the
full $2^N$-dimensional Hilbert space.

We first construct an orthonormal basis for the multiplicity space
$\mathcal P$.  In the computational basis, the highest-weight states of total
spin $S=(N-k)/2$ lie in the sector with $k/2$ spins down.  We enumerate this
sector and form the rectangular matrix representing the collective raising
operator $\hat S^+$ from the $k/2$-down sector to the $(k/2-1)$-down sector.
The multiplicity space $\mathcal P$ is the nullspace of this matrix. We numerically compute an orthonormal basis for the nullspace and collect its vectors as the columns of an isometry $V$, so that $V^\dagger V=\mathds{1}_{\mathcal P}$. 

The collective sector $\mathcal S$ is represented by the usual spin-$S$
matrices.  In the $\hat S_z$ eigenbasis,
\[
    \hat S_z\ket{m}=m\ket{m},\qquad
    \hat S_\pm\ket{m}
    =
    \sqrt{S(S+1)-m(m\pm1)}\ket{m\pm1},
\]
with $\hat S_x=(\hat S_++\hat S_-)/2$ and
$\hat S_y=(\hat S_+-\hat S_-)/(2i)$.  A pure state is then stored as the
coefficient matrix
\begin{equation}
    \ket{\psi(t)}
    =
    \sum_{a=1}^{d_{\mathcal P}}\sum_{m=-S}^{S}
    C_{am}(t)\,\ket{a}_{\mathcal P}\ket{m}_{\mathcal S},
    \label{eq:coefficient-matrix}
\end{equation}
so the reduced density matrix on the collective sector is obtained by tracing out
$\mathcal P$,
$[\rho_{\mathcal S}(t)]_{mn}
    =
    \sum_{a=1}^{d_{\mathcal P}} C_{am}(t)C_{an}(t)^* .$

The random Heisenberg couplings (RHCs) $\hat K_\mu$ are constructed in the same multiplicity basis.  We draw independent Gaussian coefficients
$c_{ij}^{(\mu)}$, subtract their mean so that
$\sum_{i<j}c_{ij}^{(\mu)}=0$, and form
\[
    \hat K_\mu
    =
    V^\dagger
    \left(\sum_{i<j}c_{ij}^{(\mu)}
    \hat h_{ij}\right)
    V ,
    \qquad
    \hat h_{ij}=\hat X_i\hat X_j+\hat Y_i\hat Y_j+\hat Z_i\hat Z_j.
\]
The coefficients are normalized so that
$d_{\mathcal P}^{-1}\tr_{\mathcal P}(\hat K_\mu^2)=1$ on average, following
Lemma~\ref{lemm:std}.

Between kicks we evolve using a fixed-time Trotter step
$\delta t$ (we use $0.005$ or smaller in all simulations, and have verified that our results are independent of this number). We numerically apply
\begin{align}
    \ket{\psi(t+\delta t)}
    &\approx
    e^{-i\frac{\kappa\delta t}{2S+1}\hat S_x^2}
    e^{-i\varepsilon\delta t\hat K_x\hat S_x}
    e^{-i\varepsilon\delta t\hat K_y\hat S_y}
    e^{-i\varepsilon\delta t\hat K_z\hat S_z}
\label{eq:trotter-step}
\end{align}
by diagonalizing each \(\hat K_\mu\) and \(\hat S_\mu\) separately, since the eigenbasis of \(\hat K_\mu\otimes \hat S_\mu\) is just the tensor product of the eigenbases.

As a time-resolved complement to Fig.~1(b) of the main text, Fig.~\ref{fig:S1}
compares the Wigner function of the collective-spin state obtained from the
full unitary dynamics with the corresponding classical dissipative evolution.
The agreement illustrates the emergence of quantum--classical correspondence
for the random Heisenberg couplings.
\begin{figure}
    \centering
    \includegraphics[width=0.99\linewidth]{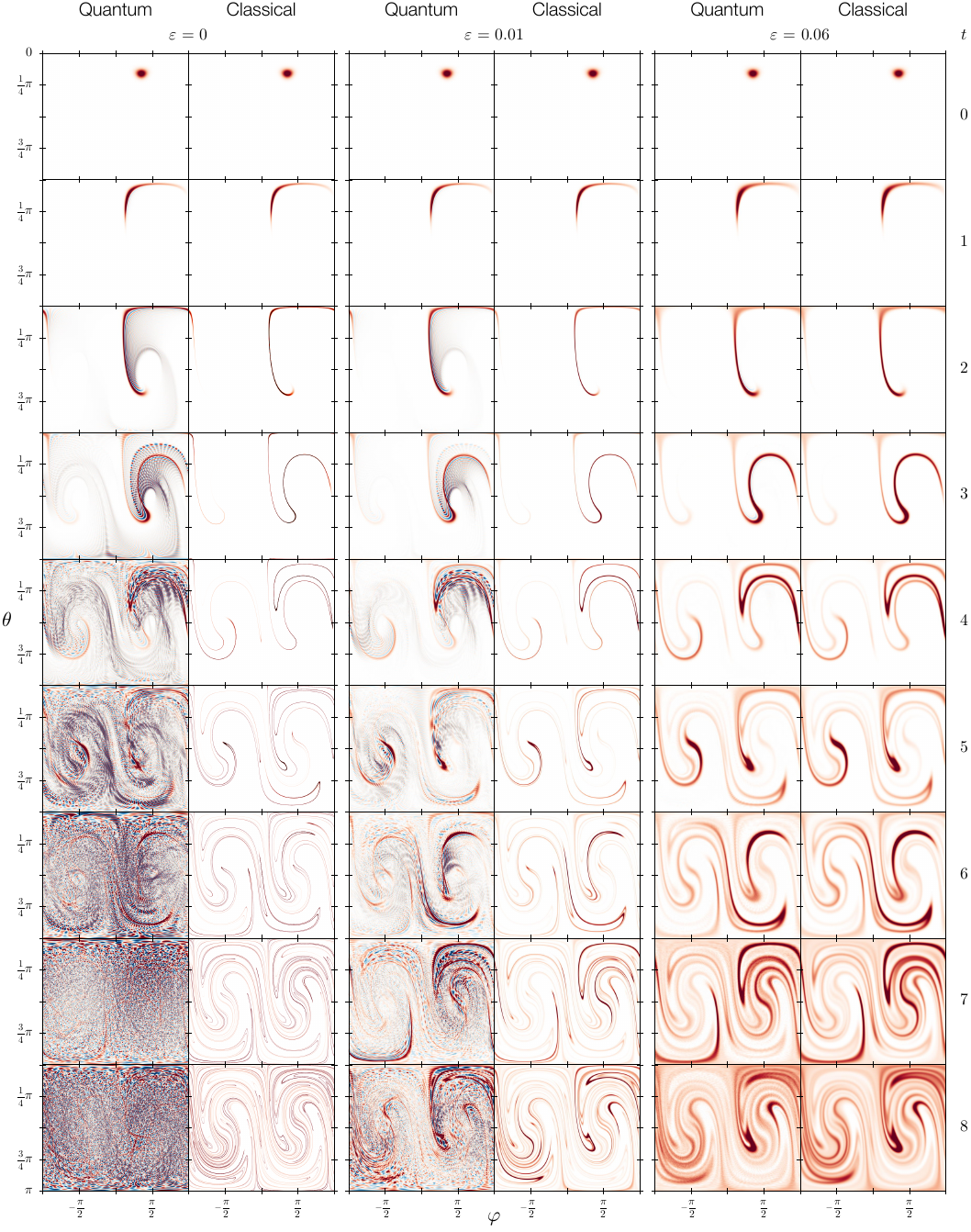}
    \vspace{-1em}
    \caption{Quantum vs classical dynamics in the collective sector (time from top to bottom). Odd columns: Wigner function of $\tr_{\mathcal{P}}(\psi(t))$, where $\ket{\psi(t)}$ undergoes unitary evolution under the disordered kicked-top Hamiltonian, with random Heisenberg couplings. Even columns: dissipative classical evolution, starting from a Gaussian distribution with variance equal to that of the spin coherent state (truncated Wigner approximation)~\cite{Villasenor2020}.  ($k=4$, $N=200$, $\kappa=4.2$, $p=\pi/2$, $\theta_0=0.5$, $\varphi_0=1.1$.)}
    \label{fig:S1}
\end{figure}
\section{Moments of random Heisenberg couplings}
\label{sec:momsofRHC}
In this section we compute ensemble averages for the random Heisenberg interaction
\begin{equation}
\label{eq:random_heisenberg_interaction}
    \hat K =\sum_{i<j}c_{ij} \hat h_{ij},\quad\quad \hat h_{ij}\coloneqq \hat X_i\hat X_j+\hat Y_i\hat Y_j+\hat Z_i\hat Z_j.
\end{equation}
The coefficients $c_{ij}$ are Gaussian, with their mean subtracted so that $\tr_\mathcal{P}(\hat K)\propto \sum_{i<j} c_{ij}=0$. We choose their standard deviation $\sigma$ so that $\E[\langle \hat K^2\rangle ]=1$, where $\langle \,\cdot\,\rangle\coloneqq \tr_{\mathcal{P}}(\,\cdot\,)/d_{\mathcal{P}}$. The following lemma gives this standard deviation explicitly.
\subsection{Standard deviation}
\begin{lemma}  \label{lemm:std} The condition $\E[\langle \hat K^2\rangle ]=1$ holds if and only if the standard deviation of the coefficients is
 \begin{align}
 \label{eq:std}
     \sigma^{-2}= \,8N(N-1) \eta (1 - \eta), \qquad
     \eta = \frac{k(2N - k + 2)}{4N(N-1)}.
\end{align}
\end{lemma}
\noindent In the limit $N\gg k$, the expression simplifies to
$\sigma^{-2}\approx 4Nk$. 
\begin{proof}
Let us compute the expected variance $\E[\langle{ \hat K^2\rangle }]$ explicitly and set it equal to $1$. We can write
\begin{equation*}
    \hat K = \sum_{i<j} c_{ij} (2 F_{ij} - \mathbb{I}),
\end{equation*}
where $F_{ij}$ is the swap operator between qubits $i$ and $j$. By the constraint $\sum_{i<j} c_{ij} = 0$, the identity term vanishes, yielding
\begin{equation*}
    \hat K = 2 \sum_{\alpha=1}^{N_p} c_\alpha F_\alpha,
\end{equation*}
where $\alpha$ indexes the $N_p = N(N-1)/2$ distinct pairs of sites. The coefficients $c_\alpha$ are Gaussian random variables subject to the sum constraint, which implies the covariance matrix:
\begin{equation*}
    \E[c_\alpha c_\beta] = \sigma^2\left(\delta_{\alpha \beta} - \frac{1}{N_p}\right).
\end{equation*}
We calculate the expected variance:
\begin{equation*}
    \E \!\left[ \langle \hat K^2\rangle  \right] = 4 \sum_{\alpha, \beta} \E[c_\alpha c_\beta] \langle F_\alpha F_\beta\rangle.
\end{equation*}
Substituting the covariance matrix splits the sum into diagonal and off-diagonal parts:
\begin{equation*}
    1=\E[\langle \hat K^2 \rangle] = 4 \sigma^2\left( \sum_\alpha \langle F_\alpha^2\rangle - \frac{1}{N_p} \sum_{\alpha, \beta} \langle F_\alpha F_\beta\rangle \right).
\end{equation*}
Since transpositions square to identity, the first term sums to $N_p$. The second term involves the square of the sum of all transpositions $\mathcal{C}_2 = \sum_\alpha F_\alpha$.
Since the operator $\mathcal{C}_2$ is permutation invariant, it acts as a scalar within the space $\mathcal{P}$. Using the identity $\mathcal{C}_2 = \mathbf{S}_{\text{tot}}^2 + \frac{N(N-4)}{4}$, this scalar must be $z=S(S+1) +N(N-4)/4$, and hence
\begin{equation*}
    \sigma^{-2} =4\left(N_p- \frac{z^2}{N_p}\right).
\end{equation*}
Using $S=(N-k)/2$ one can simplify this equation to Eq.~\eqref{eq:std}.
\end{proof}

\subsection{Variance of a commutator}
The next result gives the variance of a commutator between two independent random Heisenberg couplings.
\begin{lemma}
\label{lemm:commutav}
    Let $\hat K_0$ and $\hat K_1$ be independent random Heisenberg couplings. Then
    \begin{equation}
    \label{eq:commutator}
        \E\left [\langle[\hat K_0,\hat K_1]^\dagger[\hat K_0,\hat K_1]\rangle\right]=\sigma^4 \,32N(N-1)(N-2)(1-\phi_3),
    \end{equation}
where $\sigma$ is given by Lemma~\ref{lemm:std}, 
 and $\phi_3$ is the normalized character of the 3-cycle in the representation with total spin $S$, given by \cite[Eq. (5.2)]{Ingram1950}
\begin{align}
\label{eq:character}
    \phi_3 &= \frac{(N-3)!}{N!} \left( \frac{M_3}{2} - \frac{3}{2}N(N-1) \right),
   \\
    M_3 &= \sum_{j=1}^2 \left[ (\lambda_j - j)(\lambda_j - j + 1)(2\lambda_j - 2j + 1) + j(j-1)(2j-1) \right],&& \lambda_1=N-k/2, && \lambda_2=k/2. \nonumber
\end{align}

\end{lemma}
\begin{proof}
    Let us write $\mathcal{K} = [\hat K_0, \hat K_1]$, and $\hat K_0=2\sum_{\alpha} c_{\alpha}F_{\alpha}$ and $\hat K_1=2\sum_{\beta} c'_{\beta}F_{\beta}$ as sums over transpositions
$F_{ij}$ (with $\alpha,\beta$ indexing unordered pairs), with independent Gaussian coefficients satisfying
$\E[c'_{\alpha}c'_{\beta}]=\E[c_{\alpha}c_{\beta}]=\sigma^2(\delta_{\alpha\beta}-1/N_p)$ and $\E[c_{\alpha}c'_{\beta}]=0$.
Then
\begin{align}
    \mathcal{K}
    &= [\hat K_0,\hat K_1]
    = \left[2\sum_{\alpha} c_{\alpha}F_{\alpha},\,2\sum_{\beta} c'_{\beta}F_{\beta}\right] 
    = 4\sum_{\alpha,\beta} c_{\alpha}c'_{\beta}[F_{\alpha},F_{\beta}]. \label{eq:def-K-comm}
\end{align}
Only non-disjoint pairs contribute: for transpositions, $[F_{ij},F_{kl}]=0$ unless the pairs share exactly one site. There are $N(N-1)(N-2)$ such ordered choices: choose the repeated site in $N$ ways and the other two sites in $(N-1)(N-2)$ ways. The constraints $i<j$ and $k<l$ then fix the internal order within each pair.
For such shared-site pairs, we have
\begin{align*}
    [F_{ij},F_{jk}]
    &= F_{ij}F_{jk}-F_{jk}F_{ij} 
    = F_{(ijk)}-F_{(ikj)}, 
\end{align*}
so that
\begin{align}
    [F_{ij},F_{jk}]^\dagger [F_{ij},F_{jk}]
    = \bigl(F_{(ikj)}-F_{(ijk)}\bigr)\bigl(F_{(ijk)}-F_{(ikj)}\bigr) \nonumber
    = 2\mathbb{I}-F_{(ijk)}-F_{(ikj)}, 
\end{align}
where $F_{(ijk)}\coloneqq F_{ij}F_{jk}$ denotes the operator representing the 3-cycle $(ijk)$ in $\mathcal{P}$ (and $F_{(ikj)}\coloneqq F_{jk}F_{ij}=F_{(ijk)}^\dagger$).
Taking the trace over $\mathcal{P}$ and using $\tr_{\mathcal{P}}(F_{(ijk)})=\tr_{\mathcal{P}}(F_{(ikj)})=d_{\mathcal{P}}\phi_3$ yields 
\begin{align*}
    \langle[F_{ij},F_{jk}]^\dagger [F_{ij},F_{jk}]\rangle
    = 2(1-\phi_3). 
\end{align*}
We now perform an average over disorder. Since the disorder is Gaussian, we evaluate the fourth moments using Wick's theorem,
\begin{align*}
    \E\!\left[c_{\alpha}c_{\alpha'}c_{\gamma}c_{\gamma'}\right]
    &= \E[c_{\alpha}c_{\alpha'}]\E[c_{\gamma}c_{\gamma'}]
    + \E[c_{\alpha}c_{\gamma}]\E[c_{\alpha'}c_{\gamma'}]
    + \E[c_{\alpha}c_{\gamma'}]\E[c_{\alpha'}c_{\gamma}],
\end{align*}
and similarly for the $c'$'s. Mixed contractions vanish since the two disorder families are independent: $\E[c_{\alpha}c'_{\beta}]=0$.
Moreover $\E[c_{\alpha}c_{\alpha'}]=\sigma^2(\delta_{\alpha\alpha'}-1/N_p)$ (and similarly for $c'$).  Using Eq.~(\ref{eq:def-K-comm}),
\begin{align}
    \E\!\left[\langle \mathcal{K}^\dagger \mathcal{K} \rangle\right]
    &= 16 \sum_{\alpha,\beta}\sum_{\alpha',\beta'}
    \E\!\left[c_{\alpha}c_{\alpha'}\right]\E\!\left[c'_{\beta}c'_{\beta'}\right]\,
    \langle[F_{\alpha},F_{\beta}]^\dagger [F_{\alpha'},F_{\beta'}]\rangle.
    \label{eq:avg-double-sum}
\end{align}

Using $\sum_{\alpha}F_{\alpha}=\mathcal{C}_2\propto \mathbb{I}$ on $\mathcal{P}$, the $-1/N_p$ pieces in the covariances only contribute permutation-invariant terms (proportional to $\mathcal{C}_2$) which commute with all $F_{\alpha}$. Therefore these contributions vanish inside the commutators and make no contribution to $\E[\langle\mathcal{K}^\dagger \mathcal{K}\rangle]$. Consequently, we may replace $\E[c_{\alpha}c_{\alpha'}]\to\sigma^2 \delta_{\alpha\alpha'}$ and $\E[c'_{\beta}c'_{\beta'}]\to\sigma^2\delta_{\beta\beta'}$ in Eq.~(\ref{eq:avg-double-sum}), yielding

\begin{align}
    \E\!\left[\langle \mathcal{K}^\dagger \mathcal{K} \rangle \right]
    &= 16 \sigma^4 \sum_{\alpha,\beta}
    \langle[F_{\alpha},F_{\beta}]^\dagger [F_{\alpha},F_{\beta}]\rangle= \sigma^4\,32N(N-1)(N-2)(1-\phi_3). \qedhere
\end{align}
\end{proof}

\noindent In the limit $N\gg k$ the expression above simplifies greatly. In particular, using $1-\phi_3\approx 3k/(2N)$ and $\sigma^{-2}\approx 4Nk$, Eq.~\eqref{eq:commutator} becomes
\begin{equation}
    \E\left[\langle [\hat K_0,\hat K_1]^\dagger [\hat K_0,\hat K_1]\rangle\right]
    \approx \frac{3}{k}.
\end{equation}
This simple formula admits a higher-order generalization.  

\subsection{Higher-order moments of commutators}
For two operators
$A$ and $B$, define the iterated commutators
$[A,B]_0=B,$ and $[A,B]_{n+1}=[A,[A,B]_n].$
The leading fixed-$k$ moments of these commutators can be expressed in terms of
constants $C_n$, defined as follows.  Let $G(y)=\sum_{n\geq0}g_n y^n$ starting at $g_0=1$ be the formal power series satisfying
\begin{equation}
\label{eq:G-series-higher-comm}
    \left(1+yG(y)+2y^2G'(y)\right)
    \left(2-G(y)+yG(y)^2\right)=1.
\end{equation}
Define the formal power series:
\begin{equation}
\label{eq:P-series-higher-comm}
    H(y,z)=\sum_{n\geq0}\sum_{p=0}^{2n}g_{n}z^p y^n,\qquad
    P(y,z)=\frac{H(y,z)}{1+zyH(y,z)}=\sum_{n\geq0}\sum_{p=0}^{2n}P_{n,p}z^py^n.
\end{equation}
Then we set
\begin{equation}
\label{eq:Cn-higher-comm}
    C_n=
    \frac{1}{2}\left[
    \sum_{r=0}^n {2n\choose 2r}g_rg_{n-r}
    +
    \sum_{p=0}^{2n}(-1)^p{2n\choose p}P_{n,p}
    \right].
\end{equation}

\begin{lemma}
\label{lemm:higher-commutav}
Let $\hat K_0$ and $\hat K_1$ be independent random Heisenberg couplings, normalized as
in Lemma~\ref{lemm:std}.  For each fixed $n$ and fixed $k$,
\begin{equation}
\label{eq:higher-commutator}
    \E\left[
    \left\langle
    [\hat K_0,\hat K_1]_n^\dagger [\hat K_0,\hat K_1]_n
    \right\rangle
    \right]
    =
    \frac{C_n}{k^n}+O(N^{-1}),
\end{equation}
where $C_n$ is defined in Eq.~\eqref{eq:Cn-higher-comm}.
\end{lemma}

\begin{table}

\begin{tabular}{|c |l|}
\hline
$n$ & $C_n$ \\
\hline
0 & 1 \\
1 & 3 \\
2 & 28 \\
3 & 408 \\
4 & 7698 \\
5 & 173682 \\
6 & 4498622 \\
7 & 130698192 \\
8 & 4199406588 \\
9 & 147876395596 \\
10 & 5671974937140 \\
11 & 235911748619404 \\
12 & 10602128617075710 \\
13 & 513225453419669250 \\
14 & 26682263414687027790 \\
15 & 1485579557699327016216 \\
16 & 88328721495283696013796 \\
17 & 5592824183754640556807628 \\
18 & 376096726732217578613778896 \\
19 & 26789128136617248915600755328 \\
20 & 2016058168256938037502316523336 \\
\hline
\end{tabular}
\caption{Values of $C_n$ as defined by Eq.~\eqref{eq:Cn-higher-comm} up to $n=20$. }\label{tab:cns}
\end{table}
\noindent The first values are $C_0=1$ and $C_1=3$, so the case $n=1$ recovers
the large-$N$ form of Lemma~\ref{lemm:commutav}  (see Table~\ref{tab:cns} for values of $C_n$ up to $n=20$).
To prove Lemma~{\ref{lemm:higher-commutav}}, we first need the following:
 
\begin{lemma}
\label{lemm:centered-character-estimate}
Let \(\alpha_1,\ldots,\alpha_\ell\) be a sequence of pairs of qubits whose union totals \(v\leq N\) qubits and with \(\ell\geq2\) fixed as
\(N\to\infty\).
Define the traceless transpositions
$\widetilde F_\alpha\coloneqq F_\alpha-\langle F_\alpha \rangle \mathbb{I} .$
Then
\begin{equation}
\label{eq:centered-trace-asymptotic}
    \left\langle
    \widetilde F_{\alpha_1}\cdots \widetilde F_{\alpha_\ell}
    \right\rangle
    =
    \frac{k}{2N}
    \tr(
    \prod_{a=1}^\ell L_\alpha)
    +O(N^{-2}),
\end{equation}
where \(L_{\alpha}=(P_\alpha-\mathbb{I})\), with $P_\alpha$ the $v\times v$ permutation matrix associated with the transposition
\(\alpha\), acting on the space of $v$ labels.
\end{lemma}

\begin{proof}
Consider an arbitrary permutation \(\pi\in S_N\) with support size \(s=|\supp(\pi)|\leq \ell\), and let $F_\pi$ be the associated unitary acting on $\mathcal{P}$.
The normalized character associated with the permutation $\pi$ is \begin{equation*}
    \phi_\pi\coloneqq\langle F_\pi\rangle = 1 - \frac{k s}{2N} + O(N^{-2}).
\end{equation*}
Indeed, this follows from elementary representation theory. Since \(k\) is even, the Young diagram associated with the representation has two rows, $\left(N-\frac{k}{2},\frac{k}{2}\right)$. Denote $[N]=\{1,\dots,N\}$. Let $m_n\coloneqq \abs{\big\{S\subseteq [N] \mid |S|=n,\quad \pi(S)\subseteq S \big\}}$ count the number of subsets of $[N]$ of size $n$ that are fixed by $\pi$. By the Murnaghan–Nakayama rule

\begin{align*}
    \phi_\pi 
    =
    \frac{1}{d_\mathcal{P}}\left(m_{k/2} - m_{k/2-1}\right).
\end{align*}

A subset of size $n$ which is fixed by $\pi$ either avoids \(\supp(\pi)\), or contains
an entire nontrivial cycle of \(\pi\). There are $\binom{N-s}{n}$ options for the former, and since every nontrivial cycle has
length at least \(2\), the latter contribution is \(O(N^{n-2})\). Hence
\[
    m_{k/2}
    =
    {N-s\choose k/2}+O(N^{k/2-2}),\qquad  m_{k/2-1}
    =
    {N-s\choose k/2-1}+O(N^{k/2-3}).
\]
Since
\[
    d_{\mathcal P}
    =
    {N\choose k/2}-{N\choose k/2-1},
\]
expanding in powers of \(1/N\) gives
\begin{equation}
\label{eq:fixed-k-character-estimate}
    \phi_\pi
    =
    1-\frac{k|\supp(\pi)|}{2N}+O(N^{-2}).
\end{equation}
In particular, $\phi_2=1-\frac{k}{N}+O(N^{-2})$ for a transposition. Now expand the product of traceless transpositions:
\[
    \left\langle
    \widetilde F_{{\alpha_1}}\cdots \widetilde F_{{\alpha_\ell}}
    \right\rangle
    =
    \sum_{A\subseteq[\ell]}
    (-\phi_2)^{\ell-|A|}
    \phi\!\left(
    \prod_{a\in A}^{\rightarrow}\pi_{\alpha_a}
    \right),
\]
where the arrow means that the product is taken in the original order.
Using Eq.~\eqref{eq:fixed-k-character-estimate},
\[
    \phi\!\left(
    \prod_{a\in A}^{\rightarrow}\pi_{\alpha_a}
    \right)
    =
    1-\frac{k}{2N}
    \left|
    \supp\!\left(
    \prod_{a\in A}^{\rightarrow}\pi_{\alpha_a}
    \right)
    \right|
    +O(N^{-2}).
\]
Also,
\[
    (-\phi_2)^{\ell-|A|}
    =
    (-1)^{\ell-|A|}
    \left(
    1-\frac{k(\ell-|A|)}{N}
    \right)
    +O(N^{-2}).
\]
Upon multiplying, using $\sum_{A\subseteq[\ell]}(-1)^{\ell-|A|}=0,$
and
$\sum_{A\subseteq[\ell]}(-1)^{\ell-|A|}(\ell-|A|)=0,$ we get

\[
    \left\langle
    \widetilde F_{\alpha_1}\cdots \widetilde F_{\alpha_\ell}
    \right\rangle
    =
    -\frac{k}{2N}
    \sum_{A\subseteq[\ell]}
    (-1)^{\ell-|A|}
    \left|
    \supp\!\left(
    \prod_{a\in A}^{\rightarrow}\pi_{\alpha_a}
    \right)
    \right|
    +O(N^{-2}).
\]
Since $|\supp(\pi)|=\tr(I-P_\pi)$ and $P_{\prod_a {\pi_{\alpha_a}}}=\prod_{a} P_{\alpha_a}$ we get
\[
    \left\langle
    \widetilde F_{\alpha_1}\cdots \widetilde F_{\alpha_\ell}
    \right\rangle
    =
    \frac{k}{2N}
    \tr(
    \sum_{A\subseteq[\ell]}
    (-1)^{\ell-|A|}
    \prod_{a\in A}^{\rightarrow}P_{\alpha_a})
    +O(N^{-2})=
    \frac{k}{2N}
    \tr(
    \prod_{a=1}^\ell L_\alpha)
    +O(N^{-2}). \qedhere
\]
\end{proof}
Now we prove Lemma~\ref{lemm:higher-commutav}.

\noindent {\it Proof of  Lemma~\ref{lemm:higher-commutav}:}
The case \(n=0\) is the normalization
\(\E[\langle \hat K_1^2\rangle]=1=C_0\).  We henceforth assume \(n\geq1\). We expand the commutators using
\begin{equation*}
    [\hat K_0,\hat K_1]_n
    =
    \sum_{a=0}^n(-1)^a{n\choose a}\hat K_0^{n-a}\hat K_1\hat K_0^a .
\end{equation*}
By Hermiticity of $\hat K_0,\hat K_1$ and cyclicity of the trace, 
\begin{align*}
    \left\langle
    [\hat K_0,\hat K_1]_n^\dagger [\hat K_0,\hat K_1]_n
    \right\rangle
    \nonumber
    & =
    \sum_{a,b=0}^n
    (-1)^{a+b}{n\choose a}{n\choose b}
    \left\langle
    \hat K_0^{a+b}\hat K_1\hat K_0^{2n-a-b}\hat K_1
    \right\rangle
     \nonumber\\
    & =
    \sum_{p=0}^{2n}
    (-1)^p{2n\choose p}
    \left\langle
    \hat K_0^{2n-p}\hat K_1\hat K_0^p\hat K_1
    \right\rangle.
\end{align*}

We now perform the average over the
Gaussian disorder. Note that, in terms of the traceless transpositions $\widetilde F_\alpha\coloneqq F_\alpha-\langle F_\alpha \rangle \mathbb{I}$, we can write $\hat K_\mu
    =
    2\sum_\alpha \tilde c_\alpha^{(\mu)}\widetilde F_\alpha$, where $\alpha$ runs over all unordered pairs of qubits and \(\tilde c_\alpha^{(\mu)}\) are Gaussian random numbers with zero mean and variance \(\sigma^2\).
 Then, we obtain
\begin{align}
 &\E\left[
    \left\langle
    [\hat K_0,\hat K_1]_n^\dagger [\hat K_0,\hat K_1]_n
    \right\rangle
    \right] \nonumber
=
\sum_{p=0}^{2n}(-1)^p{2n\choose p}
\E\left[
    \left\langle
    \hat K_0^{2n-p}\hat K_1\hat K_0^p\hat K_1
    \right\rangle
\right] \nonumber \\
&\qquad=
2^{2n+2}
\sum_{p=0}^{2n}(-1)^p{2n\choose p}
\sum_{\alpha_1,\ldots,\alpha_{2n}}
\sum_{\beta,\beta'}
\E\!\bigg[\tilde c^{(1)}_\beta \tilde c^{(1)}_{\beta'}\prod_{j=1}^{2n}\tilde c^{(0)}_{\alpha_j}\bigg]
\left\langle
\widetilde F_{\alpha_1}\cdots \widetilde F_{\alpha_{2n-p}}
\widetilde F_\beta
\widetilde F_{\alpha_{2n-p+1}}\cdots \widetilde F_{\alpha_{2n}}
\widetilde F_{\beta'}
\right\rangle. \label{eq:expressionformoments}
\end{align}
 By Wick's theorem, the expectation reduces to contractions over pairings (here, {\it pairings} means pairings of the variables $\alpha_j$ or $\beta$, which themselves denote pairs of qubits), with covariances $\E[\tilde c_{\alpha}^{(\mu)}\tilde c_{\alpha'}^{(\mu')}]=\sigma^2\delta_{\mu\mu'}\delta_{\alpha\alpha'},$  \begin{equation*}
    \E\!\bigg[\tilde c^{(1)}_\beta \tilde c^{(1)}_{\beta'}\prod_{j=1}^{2n}\tilde c^{(0)}_{\alpha_j}\bigg]=\sigma^{2n+2}\delta_{\beta\beta'}\sum_{\Pi\in \mathrm{Pair}(2n)}
\prod_{\{i,j\}\in\Pi}\delta_{\alpha_i\alpha_j}, 
\end{equation*}
where $\mathrm{Pair}(2n)$ denotes the set of pairings of
the labels $\{1,\ldots,2n\}$ (e.g., $\mathrm{Pair}(4)=\bigl\{\{\{1,2\}\,\{3,4\}\}, \allowbreak \{\{1,3\},\{2,4\}\}, \allowbreak \{\{1,4\},\{2,3\}\}\bigr\}$). Thus, Eq.~\eqref{eq:expressionformoments} reduces to
 \begin{equation}
   \E\left[
    \left\langle
    [\hat K_0,\hat K_1]_n^\dagger [\hat K_0,\hat K_1]_n
    \right\rangle
    \right] =(4\sigma^2)^{n+1}
\sum_{p=0}^{2n}(-1)^p{2n\choose p}
\sum_{\Pi\in\mathrm{Pair}(2n)}
\sum_{\bm \alpha,\beta}
\left\langle
\widetilde F_{\Pi,p}(\bm \alpha;\beta)
\right\rangle, \label{eq:centered-wick-sum}\end{equation}
where
\begin{equation}
\label{eq:Fpip-word}
\widetilde F_{\Pi,p}({\bm \alpha};\beta)
=
\widetilde F_{\alpha_{\Pi(1)}}\cdots
\widetilde F_{\alpha_{\Pi(2n-p)}}\widetilde F_\beta
\widetilde F_{\alpha_{\Pi(2n-p+1)}}\cdots
\widetilde F_{\alpha_{\Pi(2n)}}\widetilde F_\beta,
\end{equation}
and  
$\Pi(j)$ is the pair containing the label $j\in\{1,\dots,2n\}$, 
and now each $\alpha_\pi$ is labeled by a pair $\pi \in \Pi$, forming the  $n$-tuple $\bm \alpha=(\alpha_\pi)_{\pi \in \Pi}$. The sum $\sum_{\bm \alpha}$ contains a sum over each $\alpha_\pi$, each running over the $N_p$ possible pairs of qubits.  

For fixed \(\Pi,\bm\alpha,\beta\), we say that a label \(\alpha_\pi\) is
{\it connected} to \(\beta\) if there exist labels
\(\alpha_{\pi_1},\ldots,\alpha_{\pi_s}\), with \(\pi_s=\pi\in \Pi\), such that
$ \beta\cap\alpha_{\pi_1}\neq\varnothing,$ and $   \alpha_{\pi_i}\cap\alpha_{\pi_{i+1}}\neq\varnothing
    \quad (1\leq i<s)$. That is, $\beta$ and $\alpha_{\pi_1}$ share a qubit, and also $\alpha_{\pi_i}$ and $\alpha_{\pi_{i+1}}$ for each $1\leq i<s$.
Let 
\begin{equation}
    \label{eq:disconnected-labels}
    \mathrm{dis}(\bm\alpha,\beta)=\{\pi \in \Pi \,|\, \alpha_\pi \text{ is not connected to }\beta\}.
\end{equation}
We split the sum into the connected component of $\beta$ and the rest, which contains some disconnected $\alpha_{\pi}$:
\begin{align}
\label{eq:U-Nk-definition}
Q_{p}
&\coloneqq
\sum_{\Pi\in\mathrm{Pair}(2n)}
\sum_{\substack{\bm\alpha,\beta\\ \mathrm{dis}(\bm\alpha,\beta)=\varnothing}}
\left\langle
\widetilde F_{\Pi,p}(\bm\alpha;\beta)
\right\rangle,
&&
D_p
\coloneqq
\sum_{\Pi\in\mathrm{Pair}(2n)}
\sum_{\substack{\bm\alpha,\beta\\ \mathrm{dis}(\bm\alpha,\beta)\neq\varnothing}}
\left\langle
\widetilde F_{\Pi,p}(\bm\alpha;\beta)
\right\rangle 
\end{align}
so
$$\E\left[
    \left\langle
    [\hat K_0,\hat K_1]_n^\dagger [\hat K_0,\hat K_1]_n
    \right\rangle
    \right]=(4\sigma^2)^{n+1}\sum_{p=0}^{2n}(-1)^p{2n\choose p} \left(Q_p + D_p\right).$$

\begin{claim}[Only the connected component contributes] We have
\label{claim:1}
    \begin{equation}
        \sum_{p=0}^{2n}(-1)^p{2n\choose p}D_p=0.
    \end{equation}
\end{claim}
\begin{proof}
    Define  the functional $\Sigma[\,\cdot\,]$ acting on a function $p\mapsto f(p)$ by
\[
\Sigma[f]
\coloneqq
\sum_{p=0}^{2n}(-1)^p{2n\choose p}f(p).
\]
Note that the functions $f_{a,b}(p)={p\choose a}{2n-p\choose b}$ with $a,b\geq 0$ and $a+b<2n$ are annihilated by this functional:
\begin{align}
\label{eq:Delta-kills-binomial}
\Sigma[f_{a,b}]
&=
\sum_{p=0}^{2n}(-1)^p{2n\choose p}
{p\choose a}{2n-p\choose b}
=
(-1)^a
\frac{(2n)!}{a!\,b!\,(2n-a-b)!}
\sum_{q=0}^{2n-a-b}(-1)^q{2n-a-b\choose q}
=0.
\end{align}

We now show that \(D_{p}\) is a linear combination of the functions $f_{a,b}(p)$.
Fix \(\bm\alpha,\beta\) and let $s$ be the size of $\Pi\setminus \mathrm{dis}(\bm\alpha,\beta)$. If
\(\mathrm{dis}(\bm\alpha,\beta)\neq\varnothing\), then \(s<n\). Moreover, if
\(\pi\in \mathrm{dis}(\bm\alpha,\beta)\) and \(\pi'\notin \mathrm{dis}(\bm\alpha,\beta)\), then
\[
    \alpha_\pi\cap\alpha_{\pi'}=\varnothing,
    \qquad
    \alpha_\pi\cap\beta=\varnothing,
\]
because otherwise \(\alpha_\pi\) would also be connected to \(\beta\). Hence
\begin{equation}
\label{eq:Fscommuteforpiindisnotdis}
        [\widetilde F_{\alpha_\pi},\widetilde F_{\alpha_{\pi'}}]=0,
    \qquad
    [\widetilde F_{\alpha_\pi},\widetilde F_\beta]=0
    \qquad
    (\pi\in \mathrm{dis}(\bm\alpha,\beta),\ \pi'\notin \mathrm{dis}(\bm\alpha,\beta)).
\end{equation}
Write $ i_1<\cdots<i_{2s}$
for the labels \(j\in\{1,\ldots,2n\}\) such that
\(\Pi(j)\notin \mathrm{dis}(\bm\alpha,\beta)\), and write
$j_1<\cdots<j_{2n-2s}$ 
for the remaining labels, so \(\Pi(j_q)\in \mathrm{dis}(\bm\alpha,\beta)\).  Suppose that
exactly \(a\) of \(i_1,\ldots,i_{2s}\) are in \(\{2n-p+1,\ldots,2n\}\), i.e.,
\(i_{2s-a}\leq 2n-p<i_{2s-a+1}\).

By Eq.~\eqref{eq:Fscommuteforpiindisnotdis} we may perform the following commutation:
\begin{align*}
\widetilde F_{\Pi,p}(\bm\alpha;\beta)
&=
\left(
\prod_{r=1}^{2s-a}
\widetilde F_{\alpha_{\Pi(i_r)}}
\right)
\underbrace{
\left(
\prod_{\substack{q:\,j_q\leq 2n-p}}
\widetilde F_{\alpha_{\Pi(j_q)}}
\right)
}_{(\star)}
\widetilde F_\beta
\left(
\prod_{r=2s-a+1}^{2s}
\widetilde F_{\alpha_{\Pi(i_r)}}
\right)
\underbrace{
\left(
\prod_{\substack{q:\,j_q>2n-p}}
\widetilde F_{\alpha_{\Pi(j_q)}}
\right)
}_{(\star\star)}
\widetilde F_\beta
\nonumber\\
&=
\left(
\prod_{r=1}^{2s-a}
\widetilde F_{\alpha_{\Pi(i_r)}}
\right)
\widetilde F_\beta
\left(
\prod_{r=2s-a+1}^{2s}
\widetilde F_{\alpha_{\Pi(i_r)}}
\right)
\widetilde F_\beta
\underbrace{
\left(
\prod_{q=1}^{2n-2s}
\widetilde F_{\alpha_{\Pi(j_q)}}
\right)
}_{(\star)\times(\star\star)} ,
\end{align*}
and hence
\begin{align*}
\bigg\langle
\widetilde F_{\Pi,p}(\bm\alpha;\beta)
\bigg\rangle
&=
\bigg\langle
\bigg(
\prod_{r=1}^{2s-a}
\widetilde F_{\alpha_{\Pi(i_r)}}
\bigg)
\widetilde F_\beta
\bigg(
\prod_{r=2s-a+1}^{2s}
\widetilde F_{\alpha_{\Pi(i_r)}}
\bigg)
\widetilde F_\beta
\bigg(
\prod_{q=1}^{2n-2s}
\widetilde F_{\alpha_{\Pi(j_q)}}
\bigg)
\bigg\rangle \eqqcolon W_{s,a},
\label{eq:trace-after-moving-J}
\end{align*}
where the dependence of the right-hand side on $p$ is now entirely through \(a\). Thus, we can collect
\begin{equation*}
    D_p=\sum_{\Pi\in\mathrm{Pair}(2n)}
\sum_{\substack{\bm\alpha,\beta\\ \mathrm{dis}(\bm\alpha,\beta)\neq\varnothing}}
\left\langle
\widetilde F_{\Pi,p}(\bm\alpha;\beta)
\right\rangle 
=
\sum_{s=0}^{n-1}\sum_{a=0}^{2s}
W_{s,a}
m_{s,a,p},
\end{equation*}
where $m_{s,a,p}$ counts the possible choices for $a$ and the upper limit \(s\leq n-1\) follows from \(\mathrm{dis}(\bm\alpha,\beta)\neq\varnothing\). Let us count the possible choices for $a$. Recall that \(s=|\Pi\setminus \mathrm{dis}(\bm \alpha,\beta)|\), and suppose that exactly \(a\) of the \(2s\) positions in the word $\widetilde F_{\Pi,p}(\bm \alpha,\beta)$ whose
labels are not in \(\mathrm{dis}(\bm \alpha,\beta)\) occur to the right of the first $\widetilde F_\beta$, i.e., among the last \(p\)
positions. Then the other \(2s-a\) such positions occur to the left of $\widetilde F_\beta$. The number of choices of these positions is then exactly
$m_{s,a,p}={p\choose a}{2n-p\choose 2s-a}$.

Applying the functional \(\Sigma\) we can then conclude 
\begin{align}
\sum_{p=0}^{2n}(-1)^p{2n\choose p}D_{p}
&=\Sigma\left[p\mapsto \sum_{s=0}^{n-1}\sum_{a=0}^{2s}
W_{s,a}
m_{s,a,p}\right ]=
\sum_{s=0}^{n-1}\sum_{a=0}^{2s}
W_{s,a}
\Sigma\left[p\mapsto
{p\choose a}{2n-p\choose 2s-a}
\right]
\nonumber=0,
\end{align}
where we used \(a+(2s-a)=2s<2n\) to apply Eq. \eqref{eq:Delta-kills-binomial} with $b=2s-a$.
\end{proof}
Because of Claim~\ref{claim:1}, we obtain
$$\E\left[
    \left\langle
    [\hat K_0,\hat K_1]_n^\dagger [\hat K_0,\hat K_1]_n
    \right\rangle
    \right]=(4\sigma^2)^{n+1}\sum_{p=0}^{2n}(-1)^p{2n\choose p} Q_{p}.$$
The proof is complete once we prove the following.
\begin{claim}
For fixed \(p,n\), and \(k\), with \(P_{n,p}\) and \(g_r\) defined by
Eqs.~\eqref{eq:G-series-higher-comm} and~\eqref{eq:P-series-higher-comm},
\begin{equation}
        (4\sigma^2)^{n+1} Q_{p}
        =
        \frac{1}{2k^n}\left(
        P_{n,p}
        +
        \begin{cases}
        g_{p/2}g_{n-p/2}, & p\ {\rm even},\\
        0, & p\ {\rm odd}
        \end{cases}
        \right)
        +O(N^{-1}) .
\end{equation}
\end{claim}

\begin{proof}
We begin by applying Lemma~\ref{lemm:centered-character-estimate} to \(Q_p\):
\begin{align*}
Q_p
&=
\sum_{\Pi\in\mathrm{Pair}(2n)}
\sum_{\substack{\bm\alpha,\beta\\ \mathrm{dis}(\bm\alpha,\beta)=\varnothing}}
\left[
\frac{k}{2N}
\tr\!\left(
L_{\alpha_{\Pi(1)}}\cdots
L_{\alpha_{\Pi(2n-p)}}L_\beta
L_{\alpha_{\Pi(2n-p+1)}}\cdots
L_{\alpha_{\Pi(2n)}}L_\beta
\right)
+O(N^{-2})
\right].
\end{align*}
Here the trace is taken over $\mathbb{R}^v$, where $v$ counts the total number of qubits appearing in \(\bm\alpha,\beta\), each labeling a canonical basis vector $e_i$, and \(L_{\{i,j\}}=(e_i-e_j)(e_i-e_j)^T=-(P_{\{i,j\}}-\mathbb{I})\) (the sign cancels since the number of terms is even).

 We now identify the terms that contribute at leading order in $N$. These terms are associated with {\it maximal trees}, in the following sense.  Let
\(G(\bm\alpha,\beta)\) be the graph whose vertices are the qubits appearing
in the pairs \(\beta\) and \(\alpha_\Pi\), and whose edges are those pairs. Let \(E\) be the number of distinct edges
and \(v\) the number of distinct qubits in this graph. Because we are summing only over configurations with
\(\mathrm{dis}(\bm\alpha,\beta)=\varnothing\), the graph
\(G(\bm\alpha,\beta)\) is connected. Hence $v\le E+1$.
Moreover \(E\le n+1\): there are at most \(n\) distinct \(\alpha_\pi\) edges, plus the distinguished edge \(\beta\). Therefore $v\le n+2.$
For a fixed graph using \(v\) distinct qubits, the number of embeddings into
\(\{1,\ldots,N\}\) is \(O(N^v)\). Thus, the leading order corresponds to graphs with the maximal $v=n+2$, and hence $E=n+1$: a tree with $n+1$ edges. Since $ (4\sigma^2)^{n+1}
    =
    \frac{1}{N^{n+1}k^{n+1}}(1+O(N^{-1})),$ we obtain
    \begin{align*}
         (4\sigma^2)^{n+1}Q_p
&=
\frac{1}{2N^{n+2}k^n}
\sum_{\Pi\in\mathrm{Pair}(2n)}
\sum_{\substack{\bm\alpha,\beta\\ G(\bm\alpha,\beta)\text{ is a tree}\\
|E(G)|=n+1}}
\tr\!\left(
L_{\alpha_{\Pi(1)}}\cdots
L_{\alpha_{\Pi(2n-p)}}L_\beta
L_{\alpha_{\Pi(2n-p+1)}}\cdots
L_{\alpha_{\Pi(2n)}}L_\beta
\right)
+O(N^{-1}) ,
    \end{align*}
where we summed the error term over the \(O(N^{n+2})\) leading
tree configurations.

We now transform the sums over pairings $\Pi$ and over pairs of qubits $\bm \alpha,\beta$ into sums over trees and {\it words}. This reduces the problem to counting such objects.
Let us label qubits such that \(\beta=\{0,1\}\), which we will take to be the {\it root} of the tree. Denote by \(\mathcal T_n\) the set of trees on the
vertex set \(\{0,1,\ldots,n+1\}\) which contain the distinguished edge
\(\beta\). For \(T\in\mathcal T_n\), write
\(E(T)\setminus\{\beta\}=\{e_1,\ldots,e_n\}\), and let \(\mathcal W(T)\) be
the set of words \(w=w_1\cdots w_{2n}\) in the alphabet
\(E(T)\setminus\{\beta\}\) such that each non-root edge appears exactly twice. The data of a pairing \(\Pi\in\mathrm{Pair}(2n)\), together with an assignment
\(\pi\mapsto\alpha_\pi\) of distinct non-root edges to the \(n\) pairs
\(\pi\in\Pi\), is equivalent to such a word: the letter in position \(j\) is
the edge assigned to the pair \(\Pi(j)\). Conversely, any word in which every
edge appears exactly twice determines the pairing by pairing the two
positions occupied by the same edge.

 The dependence on \(N\) now enters only through the choice of the
\(n+2\)  qubits supporting the maximal tree. We may first choose
this support, which gives \({N\choose n+2}\) choices. Once the support is
chosen, we simply have to select which pair of qubits forms the distinguished edge $\beta$, giving  $\frac{1}{2}(n+2)(n+1)$ more options.  Therefore
\begin{align*}
&\sum_{\Pi\in\mathrm{Pair}(2n)}
\sum_{\substack{\bm\alpha,\beta\\ G(\bm\alpha,\beta)\text{ is a tree}\\
|E(G)|=n+1}}
\tr\!\left(
L_{\alpha_{\Pi(1)}}\cdots
L_{\alpha_{\Pi(2n-p)}}L_\beta
L_{\alpha_{\Pi(2n-p+1)}}\cdots
L_{\alpha_{\Pi(2n)}}L_\beta
\right)
\nonumber\\
&\qquad =
{N\choose n+2}\frac{(n+1)(n+2)}{2}
\sum_{T\in\mathcal T_n}
\sum_{w\in\mathcal W(T)}
\tr\!\left(
L_{w_1}\cdots L_{w_{2n-p}}L_\beta
L_{w_{2n-p+1}}\cdots L_{w_{2n}}L_\beta
\right).
\end{align*}
Since
\(\binom{N}{n+2}=N^{n+2}+O(N^{n+1})\), we obtain
\begin{align*}
    (4\sigma^2)^{n+1}Q_p
    =
    \frac{A_{n,p}}{2k^n}+O(N^{-1}),
\end{align*}
where we defined
\begin{equation}
\label{eq:Anp-rhc}
A_{n,p}
\coloneqq
\frac{1}{2n!}
\sum_{T\in\mathcal T_n}
\sum_{w\in\mathcal W(T)}
\tr\!\left(
L_{w_1}\cdots L_{w_{2n-p}}L_\beta
L_{w_{2n-p+1}}\cdots L_{w_{2n}}L_\beta
\right),
\end{equation}
which is independent of $N$. Thus, we have removed all the $N$ dependence, and it simply remains to compute the coefficients \(A_{n,p}\). 

To compute the coefficients \(A_{n,p}\), we analyze the combinatorics in the space of trees. Write, for a word
\(u=u_1\cdots u_m\) in edges of \(T\setminus\{\beta\}\),
\[
    L_u \coloneqq L_{u_1}\cdots L_{u_m},
    \qquad L_{\varnothing}=\mathbb{I}.
\]
Since \(L_\beta=(e_0-e_1)(e_0-e_1)^T\), if
\(u=w_1\cdots w_{2n-p}\) and
\(v=w_{2n-p+1}\cdots w_{2n}\), then
\begin{align*}
    A_{n,p}
&=
\frac{1}{2 n!}
\sum_{\substack{T\in\mathcal{T}_n\\ \mathclap{w\in\mathcal W(T)}}}
\tr\!\left(
L_u(e_0-e_1)(e_0-e_1)^T
L_v(e_0-e_1)(e_0-e_1)^T
\right)
\\
&=
\frac{1}{2 n!}
\sum_{\substack{T\in\mathcal{T}_n\\ \mathclap{w\in\mathcal W(T)}}}
\left((e_0-e_1)^T L_u(e_0-e_1)\right)
\left((e_0-e_1)^T L_v(e_0-e_1)\right)
\\
&=
\frac{1}{2n!}
\sum_{\substack{T\in\mathcal{T}_n\\ \mathclap{w\in\mathcal W(T)}}}
\left(e_0^T L_u e_0+e_1^T L_u e_1\right)
\left(e_0^T L_v e_0+e_1^T L_v e_1\right)
\\
&=
\underbrace{\frac{1}{2 n!}
\sum_{\substack{T\in\mathcal{T}_n\\ \mathclap{w\in\mathcal W(T)}}}
\left[
(e_0^T L_u e_0)(e_0^T L_v e_0)
+
(e_1^T L_u e_1)(e_1^T L_v e_1)
\right]}_{\eqqcolon \, P_{n,p}}
+
\underbrace{\frac{1}{2n!}
\sum_{\substack{T\in\mathcal{T}_n\\ \mathclap{w\in\mathcal W(T)}}}
\left[
(e_0^T L_u e_0)(e_1^T L_v e_1)
+
(e_1^T L_u e_1)(e_0^T L_v e_0)
\right]}_{(\star)} .
\end{align*}
Let us focus on $(\star)$. The factors $e_a^T L_u e_a$ are zero unless the word $u$ is supported on the side of the tree attached to the vertex $a$ of $\beta$. Thus \(e_0^T L_u e_0\) vanishes unless \(u\) lies entirely on the
\(0\)-side of \(\beta\), while \(e_1^T L_v e_1\) vanishes unless \(v\) lies
entirely on the \(1\)-side. Hence the cross terms in $(\star)$ can be nonzero only when \(u\)
and \(v\) lie on opposite sides of \(\beta\), which forces \(p\) to be even. For even $p=2r$, we can split into two separate contributions where \(u\) uses exactly the \(n-r\) edges on the \(0\) side of \(\beta\), and \(v\) uses exactly the \(r\) edges on the \(1\) side, weighted by the combinatorial factor $\binom{n}{r}$. To write this in equations, let us denote by \(\mathcal T_m^{(a)}\) the set of rooted trees
with \(m+1\) vertices, root \(a\in\{0,1\}\), and with vertex set disjoint from the other
endpoint of \(\beta\).  For \(T\in\mathcal T_m^{(a)}\),
\(\mathcal W(T)\) denotes the set of words in the \(m\) edges of \(T\), with
each edge appearing exactly twice, and define
\begin{equation}
    \label{eq:gn-rhc}
     g_m
    \coloneqq
    \frac{1}{m!}
    \sum_{T\in\mathcal T_m^{(a)}}
    \sum_{w\in\mathcal W(T)}
    e_a^T L_{w} e_a,
    \qquad a\in\{0,1\}.
\end{equation}

Then
\begin{align*}
(\star)
&=
\frac{1}{2 n!}\,
2{n\choose r}
\bigg[
\sum_{T_0\in\mathcal T_{n-r}^{(0)}}
\sum_{u\in\mathcal W(T_0)}
e_0^T L_{u} e_0
\bigg]
\bigg[
\sum_{T_1\in\mathcal T_r^{(1)}}
\sum_{v\in\mathcal W(T_1)}
e_1^T L_{v} e_1
\bigg]
=
\frac{1}{2 n!}\,
2{n\choose r}
\bigl((n-r)!g_{n-r}\bigr)
\bigl(r!g_r\bigr)
=
g_rg_{n-r}.
\end{align*}
Therefore, we have shown
\[
    A_{n,p}
    =
    P_{n,p}
    +
    \begin{cases}
    g_{p/2}g_{n-p/2}, & p\ {\rm even},\\
    0, & p\ {\rm odd}.
    \end{cases}
\]
It remains to show that $P_{n,p}$ indeed satisfies Eq.~\eqref{eq:P-series-higher-comm} and that $G(y)=\sum_{n\geq0}g_n y^n$ satisfies Eq.~\eqref{eq:G-series-higher-comm}.

First observe that
$$P_{n,p}=\frac{1}{n!}
\sum_{\substack{T\in\mathcal{T}_n^{(0)}\\ w\in\mathcal W(T)}}
(e_0^T L_u e_0)(e_0^T L_v e_0)=\frac{1}{n!}
\sum_{\substack{T\in\mathcal{T}_n^{(1)}\\ w\in\mathcal W(T)}}
(e_1^T L_u e_1)(e_1^T L_v e_1)$$ 
for $u=w_1\cdots w_{2n-p}$ and $v=w_{2n-p+1}\cdots w_{2n},$ for any $p$. Note first that if \(n=0\) or \(p=0\), then trivially
\(g_n=P_{n,p}\). Thus assume \(n\neq 0\) and \(p\neq 0\). We calculate by inserting a 
resolution of the identity $\sum_x e_xe_x^T=\mathbb I$:
\begin{align*}
g_n
&=
\frac{1}{n!}
\sum_{T\in\mathcal T_n^{(0)}}
\sum_{w\in\mathcal W(T)}
e_0^T L_uL_v e_0
\\
&=
\frac{1}{n!}
\sum_{T\in\mathcal T_n^{(0)}}
\sum_{w\in\mathcal W(T)}
\sum_{x=0}^{n+1}
\left(e_0^T L_u e_x\right)
\left(e_x^T L_v e_0\right)
\\
&=
P_{n,p}
+
\frac{1}{n!}
\sum_{T\in\mathcal T_n^{(0)}}
\sum_{w\in\mathcal W(T)}
\sum_{x=1}^{n+1}
\left(e_0^T L_u e_x\right)
\left(e_x^T L_v e_0\right)
\\
&=
P_{n,p}
+
\frac{1}{n!}
\sum_{\substack{a+b=n-1\\ q+r=p-1}}
{n\choose a,b,1}
\bigg[
\sum_{T\in\mathcal T_a^{(0)}}
\sum_{\substack{w=uv\in\mathcal W(T)\\ |v|=q}}
\left(e_0^T L_u e_0\right)
\left(e_0^T L_v e_0\right)
\bigg]
\\
&\qquad\qquad\times
\bigg[
\sum_{T\in\mathcal T_b^{(1)}}
\sum_{\substack{w=uv\in\mathcal W(T)\\ |v|=r}}
\sum_{x=1}^{b+1}
\left(e_1^T L_u e_x\right)
\left(e_x^T L_v e_1\right)
\bigg]
\\
&=
P_{n,p}
+
\frac{1}{n!}
\sum_{\substack{a+b=n-1\\ q+r=p-1}}
{n\choose a,b,1}
\bigg[
\sum_{T\in\mathcal T_a^{(0)}}
\sum_{\substack{w=uv\in\mathcal W(T)\\ |v|=q}}
\left(e_0^T L_u e_0\right)
\left(e_0^T L_v e_0\right)
\bigg]
\bigg[
\sum_{T\in\mathcal T_b^{(1)}}
\sum_{w\in\mathcal W(T)}
e_1^T L_w e_1
\bigg]
\\
&=
P_{n,p}
+
\frac{1}{n!}
\sum_{\substack{a+b=n-1\\ q+r=p-1}}
\frac{n!}{a!\,b!}
\left(a!P_{a,q}\right)
\left(b!g_b\right)
\\
&=
P_{n,p}
+
\sum_{\substack{a+b=n-1\\ q+r=p-1}}
P_{a,q}g_b.
\end{align*}
The fourth equality is the only nontrivial step. It follows from the following decomposition. For the summand \((e_0^T L_u e_x)(e_x^T L_v e_0)\) to be nonzero, the first edge on the path from \(0\) to \(x\)
 must be crossed once by \(u\) and once by \(v\). Removing this edge separates
\(T\) into a component containing \(0\), with \(a\) edges, and a component
containing \(x\), with \(b\) edges, so \(a+b=n-1\). If \(q\) and \(r\) are the
numbers of letters of \(v\) lying in these two components, then
\(q+r=p-1\), since the removed edge also appears in \(v\). The first component
gives \(a!P_{a,q}\). For the second component, we relabel the endpoint of the
removed edge as \(1\), so the component becomes an element of
\(\mathcal T_b^{(1)}\). Under this relabeling, the sum over \(x\) becomes the
resolution of the identity on this component with $b+1$ vertices,
\(\sum_{x=1}^{b+1}e_xe_x^T=I\), and gives \(b!g_b\). Finally, the multinomial coefficient
\({n\choose a,b,1}\) counts the choice of the two components and the separating
edge.

Therefore, defining
\[
H(y,z)=\sum_{n\ge0}g_ny^n \sum_{p=0}^{2n}z^p
\qquad\text{and}\qquad
P(y,z)=\sum_{n\ge0}\sum_{p=0}^{2n}P_{n,p}z^py^n,
\]
we obtain the relation
\begin{align*}
H(y,z)
&=
\sum_{n\ge0}\sum_{p=0}^{2n}P_{n,p}z^py^n
+
\sum_{n\ge1}\sum_{p=1}^{2n}
\sum_{\substack{a+b=n-1\\ q+r=p-1}}
P_{a,q}g_{b}z^py^n
=
P(y,z)
+
\sum_{a,b\ge0}
\sum_{q=0}^{2a}
\sum_{r=0}^{2b}
P_{a,q}g_{b}z^{q+r+1}y^{a+b+1}
\\
&=
P(y,z)
+
zy
\left(
\sum_{a\ge0}\sum_{q=0}^{2a}P_{a,q}z^qy^a
\right)
\left(
\sum_{b\ge0}\sum_{r=0}^{2b}g_{b}z^ry^b
\right)
=
P(y,z)+zyP(y,z)H(y,z)\,,
\end{align*}
which implies Eq.~\eqref{eq:P-series-higher-comm}. Further, $
H(y,1)
=
P(y,1)+yP(y,1)H(y,1)$
and 
\begin{align*}
H(y,1)
=
\sum_{n\ge0}(2n+1)g_ny^n
=
\sum_{n\ge0}g_ny^n
+
2y\sum_{n\ge1}ng_ny^{n-1}
=
G(y)+2yG'(y).
\end{align*}
So
\begin{align*}
P(y,1)
&=
\frac{H(y,1)}{1+yH(y,1)}
=
\frac{G(y)+2yG'(y)}
{1+yG(y)+2y^2G'(y)}.
\end{align*}
From the above we get Eq.~\eqref{eq:G-series-higher-comm}:
\begin{align*}
\frac{1}{1+yG(y)+2y^2G'(y)}
&=
\frac{1}{1+y\bigl(G(y)+2yG'(y)\bigr)}
=
\frac{1}{1+yH(y,1)}
=
1-y\frac{H(y,1)}{1+yH(y,1)}
=
1-yP(y,1)\\
&=
1-\bigl(G(y)-1-yG(y)^2\bigr)
=
2-G(y)+yG(y)^2 .\qedhere
\end{align*}

\end{proof}

\subsection{Moments of random chiral couplings}
\label{sec:rcc-moments}
We can generalize the ensemble of 2-local RHCs to 3-local operators by replacing
transpositions with 3-cycles.  We define a random chiral coupling (RCC)
by
\begin{equation}
\label{eq:rcc-def}
    \hat K
    =
    \sum_{i<j<\ell}
    c_{ij\ell}\,
    i\left(F_{(ij\ell)}-F_{(i\ell j)}\right)
    =
    \frac{1}{2}\sum_{i<j<\ell}
    c_{ij\ell}\,
    \bm\sigma_i\cdot
    \left(\bm\sigma_j\times\bm\sigma_\ell\right),
\end{equation}
where $\bm\sigma_i=(\hat X_i,\hat Y_i,\hat Z_i)$ and
$F_{(ij\ell)}=F_{ij}F_{j\ell}$.  The second equality follows from
$F_{ij}=(\mathbb I+\bm\sigma_i\cdot\bm\sigma_j)/2$.  We take the coefficients to be independent centered Gaussians with
variance $\sigma_3^2$ and choose $\sigma_3$ so that
$\E[\langle\hat K^2\rangle]=1$.  Since $ \left\langle
    \left(F_{(ij\ell)}-F_{(i\ell j)}\right)^2
    \right\rangle
    =
    2(\phi_3-1)$, by Eq.~\eqref{eq:fixed-k-character-estimate}
the normalization is
\begin{equation}
\label{eq:rcc-std}
    \sigma_3^{-2}=2{N\choose 3}(1-\phi_3)
    =
    \frac{kN^2}{2}+O(N).
\end{equation}

The commutator moments of RCCs obey a simpler formula than the corresponding
moments of RHCs.
\begin{lemma}
\label{lemm:higher-commutav-rcc}
Let $\hat K_0$ and $\hat K_1$ be independent RCCs defined as in
Eq.~\eqref{eq:rcc-def}.  For each
fixed $n$ and fixed $k$,
\begin{align}
    \E\left[
    \left\langle
    [\hat K_0,\hat K_1]_n^\dagger [\hat K_0,\hat K_1]_n
    \right\rangle
    \right]
    =
    \frac{C_n^{(3)}}{k^n}+O(N^{-1}),
&&\text{
where}&&
    C_n^{(3)}
    =
    \frac{2^n}{(n+1)(n+2)}
    {2n\choose n}{2n+2\choose n+1}.\label{eq:Cn-rcc}
\end{align}
\end{lemma}
\begin{proof}
We follow the proof of Lemma~\ref{lemm:higher-commutav} step by step.  Let
\(\alpha=\{i,j,\ell\}\) denote a triple of sites, and write
\[
    \widetilde F_\alpha
    \coloneqq
    i\left(F_{(ij\ell)}-F_{(i\ell j)}\right),
    \qquad
    \hat K_\mu=\sum_{\alpha}c_\alpha^{(\mu)}
    \widetilde F_\alpha .
\]
The commutator expansion is exactly the same one used to obtain
Eq.~\eqref{eq:expressionformoments}.  Applying Wick's theorem to the Gaussian
coefficients gives the same expression as Eq.~\eqref{eq:centered-wick-sum},
\begin{equation}
\label{eq:rcc-wick-sum}
\E\left[
    \left\langle
    [\hat K_0,\hat K_1]_n^\dagger [\hat K_0,\hat K_1]_n
    \right\rangle
    \right]
=
\sigma_3^{2n+2}
\sum_{p=0}^{2n}(-1)^p{2n\choose p}
\sum_{\Pi\in\mathrm{Pair}(2n)}
\sum_{\bm \alpha,\beta}
\left\langle
\widetilde F_{\Pi,p}(\bm\alpha;\beta)
\right\rangle.
\end{equation}
Here \(\widetilde F_{\Pi,p}\) is defined by the same word as in
Eq.~\eqref{eq:Fpip-word}, with \(\beta\) and the \(\alpha_\pi\)'s now triples.

The next step parallels Lemma~\ref{lemm:centered-character-estimate}.
For a fixed word of triples, let \(L_\alpha\) be the finite permutation
matrix \(i(P_{(ij\ell)}-P_{(i\ell j)})\), acting on the labels that appear in
the word.  Since each \(\widetilde F_\alpha\) is a difference of two
permutations with coefficients summing to zero, the proof of
Lemma~\ref{lemm:centered-character-estimate} applies without change: only the
support size of the resulting permutation enters the \(O(N^{-1})\) term.
Therefore
\begin{equation}
\label{eq:rcc-character-estimate}
    \left\langle
    \widetilde F_{\alpha_1}
    \cdots
    \widetilde F_{\alpha_m}
    \right\rangle
    =
    \frac{k}{2N}
    \tr\left(
    L_{\alpha_1}
    \cdots
    L_{\alpha_m}
    \right)
    +O(N^{-2}).
\end{equation}
We use the definition of connectedness from Eq.~\eqref{eq:disconnected-labels},
with \(\alpha_\pi\) and \(\beta\) now triples.  Equivalently, connectedness is
computed in the 3-uniform hypergraph where two triples are adjacent if they
share at least one site.  The disconnected part of Eq.~\eqref{eq:rcc-wick-sum}
vanishes by the
same argument as Claim~\ref{claim:1}.  Indeed, disjoint triples commute, and
after commuting disconnected components to the end, the remaining dependence on
\(p\) is again killed by the finite-difference identity
Eq.~\eqref{eq:Delta-kills-binomial}.

It remains to evaluate the terms in which the triples form a connected
configuration.  After applying Wick's theorem there are \(n\) triples from
\(\hat K_0\), together with the distinguished triple \(\beta\) from
\(\hat K_1\).  A connected 3-uniform hypergraph with \(m\) hyperedges has at
most \(1+2m\) vertices, so here the number of vertices is at most \(2n+3\).
Combining Eq.~\eqref{eq:rcc-character-estimate} with Eq.~\eqref{eq:rcc-std}
gives
\[
    \sigma_3^{2n+2}\frac{k}{2N}
    =
    \frac{2^n}{k^nN^{2n+3}}+O(N^{-2n-4}).
\]
Thus a connected configuration with \(v\) vertices contributes at order
\(N^{v-(2n+3)}\).  Only the maximal case \(v=2n+3\) survives as
\(N\to\infty\).  These maximal connected configurations are {\it loose hypertrees}
made of triples: connected collections in which each added triple shares
exactly one site with the previous collection and contributes two new sites.

We can now define the coefficient \(A_{n,p}^{(3)}\), analogously to
\(A_{n,p}\) in Eq.~\eqref{eq:Anp-rhc}.  Fix
\(\beta=\{0,1,2\}\), which plays the role of the root, and let
\(\mathcal T_n^{(3)}\) be the set of such loose hypertrees \(T\) with
\(V(T)=\{0,1,\ldots,2n+2\}\) and \(\beta\in E(T)\).  For
\(T\in\mathcal T_n^{(3)}\), let \(\mathcal W(T)\) be the set of words in the
\(n\) hyperedges of \(T\) other than \(\beta\), with each such hyperedge
appearing exactly twice.
Each term that survives at leading order uses \(2n+3\) distinct sites.  Choosing
these sites and then choosing which three of them form \(\beta\) gives
\[
    {N\choose 2n+3}{2n+3\choose 3}
    =
    \frac{N^{2n+3}}{6(2n)!}+O(N^{2n+2}).
\]
Combining this with the factor \(2^n/(k^nN^{2n+3})\) above, one arrives at
\begin{equation}
\label{eq:Anp-rcc}
    A_{n,p}^{(3)}
    \coloneqq
    \frac{2^n}{6(2n)!}
    \sum_{T\in\mathcal T_n^{(3)}}
    \sum_{w\in\mathcal W(T)}
    \tr\!\left(
    L_{w_1}\cdots L_{w_{2n-p}}L_\beta
    L_{w_{2n-p+1}}\cdots L_{w_{2n}}L_\beta
    \right),
\end{equation}
so that
\[
    \E\left[
    \left\langle
    [\hat K_0,\hat K_1]_n^\dagger [\hat K_0,\hat K_1]_n
    \right\rangle
    \right]
    =
    \frac{1}{k^n}
    \sum_{p=0}^{2n}(-1)^p{2n\choose p}A_{n,p}^{(3)}
    +O(N^{-1}).
\]

We finally evaluate \(A_{n,p}^{(3)}\).  Write
\(u=w_1\cdots w_{2n-p}\) and \(v=w_{2n-p+1}\cdots w_{2n}\), and set
\(L_u=L_{u_1}\cdots L_{u_{|u|}}\), with the analogous
definition for \(L_v\).  The matrix \(L_\beta\) has zero
diagonal on the three vertices of \(\beta\) and satisfies
\((L_\beta)_{ab}(L_\beta)_{ba}=1\) for \(a\neq b\).  Therefore
the only terms that survive in the trace have the form
\[
    \tr(L_uL_\beta L_vL_\beta)
    =
    \sum_{a\neq b}
    (e_a^T L_u e_a)
    (e_b^T L_v e_b),
\]
with \(a,b\in\{0,1,2\}\).  Since \(L_\alpha\) has no identity part, a
diagonal element \(e_a^T L_u e_a\) based at \(a\) can be nonzero only
when the word \(u\) uses hyperedges entirely in the component attached to \(a\)
after removing \(\beta\).  Thus \(u\) must use one component attached to
\(\beta\), \(v\) must use a second one, and the third such component must be
empty.  In particular, \(p\) must be even.  For \(p=2r\), the two nonempty
components have
\(n-r\) and \(r\) hyperedges, respectively.

We define \(g_m^{(3)}\) analogously to \(g_n\) in Eq.~\eqref{eq:gn-rhc}, but
with the sum running over loose hypertrees made of triples, rooted at one
vertex, with \(m\) hyperedges.  To compute \(g_m^{(3)}\), consider a word that
starts at the root and returns to it.  If the word is nonempty, let
\(\alpha\) be the hyperedge used to leave the root for the first time.  Since
\(L_\alpha\) has no identity part, the next return to the root must use
\(\alpha\) again.  Between these two uses, the word lies entirely in one of the
two branches beyond \(\alpha\), and after the return the remaining letters form
an independent word rooted at the original vertex.  Therefore, for \(m\geq1\),
\(g_m^{(3)}=2\sum_{a+b=m-1}g_a^{(3)}g_b^{(3)}\), where the factor \(2\)
chooses the branch beyond \(\alpha\).  Thus the generating function
\(G^{(3)}(y)=\sum_{m\geq0}g_m^{(3)}y^m\) satisfies
\[
    G^{(3)}(y)=1+2yG^{(3)}(y)^2,
\]
and hence \(g_m^{(3)}=2^m\frac{1}{m+1}{2m\choose m}\).
After summing over the choices of the two nonempty components attached to
\(\beta\) and over the partitions of labels away from \(\beta\), the prefactor
in Eq.~\eqref{eq:Anp-rcc} accounts for the remaining combinatorial factors.
Thus
\[
    A_{n,p}^{(3)}
    =
    \begin{cases}
    g_r^{(3)}g_{n-r}^{(3)}, & p=2r\\
    0, & p\ {\rm odd}
    \end{cases}.
\]
Combining this with the alternating sum over \(p\) gives
\[
    C_n^{(3)}
    =
    \sum_{r=0}^n {2n\choose 2r}g_r^{(3)}g_{n-r}^{(3)}
    =
    2^n(2n)!
    \sum_{r=0}^n
    \frac{1}{r!(r+1)!(n-r)!(n-r+1)!}
    =
    \frac{2^n}{(n+1)(n+2)}
    {2n\choose n}{2n+2\choose n+1}.\qedhere
\]
\end{proof}

\section{Microscopic derivation of the Lindblad equation}
\label{sec:analytic-rate}
We now give a heuristic derivation of the effective Lindblad equation

\begin{equation}
\label{eq:Lindblad}
            \frac{\partial \rho}{\partial t} =\mathcal{L}_t[\rho]= -i [\hat H_{KT}(t),\rho] - \gamma \sum_{{\mathclap{\mu=x,y,z}}}\tfrac{1}{2}[\hat S_\mu,[\hat S_\mu,\rho]].
\end{equation}

\noindent To see why this equation should arise, recall the global Hamiltonian:
\begin{equation}
\label{eq:Hamiltoniananddisorder}
    \hat H(t)=\hat H_{\mathrm{KT}}(t) + \varepsilon H_{\textrm{dis}} \qquad \qquad     \hat H_{\textrm{dis}}= \hat K_x \hat S_x + \hat K_y \hat S_y +\hat K_z \hat S_z.
\end{equation}
The matrices $\hat K_x,\hat K_y,\hat K_z$ are independent samples from a bath ensemble. In the main text we focus on random Heisenberg couplings; here we also treat the Gaussian orthogonal ensemble (GOE). 
Recall that $H_{\mathrm{KT}}(t)$ and the collective spins $\hat S_\mu$ only act on $\mathcal{S}$, while each $\hat K_\mu$ only acts on $\mathcal{P}$.
To derive the Lindblad equation, we first rewrite the total Hamiltonian in the form $H_\mathcal{S}+H_\mathcal{P}+H_{\mathcal{S\mathcal{P}}}$. Although no bare bath Hamiltonian is present at first sight, an effective bath Hamiltonian emerges after expanding around a localized spin-coherent state.

\subsection{Effective bath Hamiltonian}

When the system state is localized near the north pole (positive $z$ direction) for large $N$, the state is an approximate eigenstate of $\hat{S}_z$ (with eigenvalue $S$), so the disorder Hamiltonian $\hat H_{\textrm{dis}}$ naturally breaks up into two parts, an effective bath term $S \hat K_z$ and an interaction $\hat K_x \hat S_x + \hat K_y \hat S_y$, which is analogous (as $N\to\infty$) to the standard diffusion interaction $\hat{x} B_x + \hat{p} B_p$ for a free particle in flat phase space. 

For a system initialized in a coherent state \(\ket{\Omega}\) centered at
\(\Omega=(\Omega_x,\Omega_y,\Omega_z)\) on the sphere, we can split the action of a
spin operator on this state as
\[
\hat S_\mu\ket{\Omega}
=
S\Omega_\mu\ket{\Omega}
+
(\hat S_\mu-S\Omega_\mu)\ket{\Omega},
\]
where the norm of the second term is \(O(\sqrt S)\), and is therefore small relative to the first one, which scales as $S$. This suggests decomposing the disorder into
$\varepsilon H_{\rm dis}=
\mathbb{I}_{\mathcal S}\otimes H_{\mathcal P}^{(\Omega)}
+
H_{\mathcal{SP}}^{(\Omega)}$ where
\[
H_{\mathcal P}^{(\Omega)}
=
\varepsilon S
\sum_{\mu=x,y,z}
\Omega_\mu \hat K_\mu
\]
acts only on the permutation sector for any $\Omega$, and where
\[
H_{\mathcal{SP}}^{(\Omega)}
=
\varepsilon
\sum_{\mu=x,y,z}
(\hat S_\mu-S\Omega_\mu)\otimes \hat K_\mu .
\]
Thus, when acting on a state
\(\ket{\psi}=\ket{\Omega}_{\mathcal S}\otimes\ket{\phi}_{\mathcal P}\), the
 disorder splits as
$\varepsilon H_{\rm dis}\ket{\psi}
=
\mathbb{I}_{\mathcal S}\otimes H_{\mathcal P}^{(\Omega)}\ket{\psi}
+
H_{\mathcal{SP}}^{(\Omega)}\ket{\psi}$.

We make a semiclassical approximation for the collective-spin dynamics:
an initially coherent state remains localized near a coherent state
\(\ket{\Omega(t)}\), whose center follows a classical
trajectory on the sphere. This is exactly true in the case of no system dynamics or only rotations of the sphere (we will study these cases in detail below), but we also expect it to be approximately true in the limit of large $S$. For any specified classical trajectory $\Omega(t)$, we naturally extend the above interaction and bath operators to be time-dependent:
\begin{align}
H_{\mathcal P}^{(\Omega)}(t)
:=H_{\mathcal P}^{(\Omega(t))}
=
\varepsilon S
\sum_\mu \Omega_\mu(t)\hat K_\mu ,
\qquad\qquad
H_{\mathcal{SP}}^{(\Omega)}(t)
:=H_{\mathcal{SP}}^{(\Omega(t))}
=\varepsilon
\sum_\mu
(\hat S_\mu-S\Omega_\mu(t))\otimes \hat K_\mu\,.
\end{align}
We move into the interaction picture with respect to
\(H_{\mathcal S}(t)+H_{\mathcal P}^{(\Omega)}(t)\), and define
\[
U^{(\Omega)}(t)
=
\mathcal T
\exp\left[
-i\int_0^t \dd t'\,
\left(
H_{\mathcal S}(t')\otimes\mathbb{I}_{\mathcal P}
+
\mathbb{I}_{\mathcal S}\otimes H_{\mathcal P}^{(\Omega)}(t')
\right)
\right]
=
U_{\mathcal S}(t)\otimes U_{\mathcal P}^{(\Omega)}(t).
\]
The interaction term in the interaction picture is, then,
\[
\widetilde H_{\mathcal{SP}}^{(\Omega)}(t)
=U^{(\Omega)\dagger}(t)  H_{\mathcal{SP}}^{(\Omega)}(t)
U^{(\Omega)}(t)=
\varepsilon
\sum_\mu
\Delta \widetilde S_\mu(t)
\otimes
\widetilde K_\mu(t),
\]
where
\[
\Delta \widetilde S_\mu(t)
=
U_{\mathcal S}^\dagger(t)
(\hat S_\mu-S\Omega_\mu(t))
U_{\mathcal S}(t),
\qquad
\widetilde K_\mu(t)
=
U_{\mathcal P}^{(\Omega)\dagger}(t)
\hat K_\mu
U_{\mathcal P}^{(\Omega)}(t).
\]

\subsection{Dyson series}
Having moved to the interaction picture, the full unitary dynamics are simply given by the Schr\"odinger equation $i\partial_t\widetilde \rho_{\mathcal{SP}}(t)=[\widetilde H_{\mathcal{SP}}^{(\Omega)}(t),\widetilde \rho_{\mathcal{SP}}(t)]$, starting at $\widetilde \rho_{\mathcal{SP}}(0)= \rho_{\mathcal{SP}}(0)$.  We integrate these dynamics in a Dyson series up to a certain {\it observation time} $T$ for which we ask about the empirical rate $\gamma$, which, as we will see, may in principle depend on $T$.

\begin{align}
\widetilde\rho_{\mathcal S\mathcal P}(T)
&=
\widetilde\rho_{\mathcal S\mathcal P}(0)
-i\int_0^T \dd t_1\,
[
\widetilde H_{\mathcal{SP}}^{(\Omega)}(t_1),
\widetilde\rho_{\mathcal S\mathcal P}(0)
]
-\int_0^T \dd t_1
\int_0^{t_1}\dd t_2\,
[
\widetilde H_{\mathcal{SP}}^{(\Omega)}(t_1),
[
\widetilde H_{\mathcal{SP}}^{(\Omega)}(t_2),
\widetilde\rho_{\mathcal S\mathcal P}(0)
]
]
+\cdots.
\end{align}
We compare this to the leading change from the dissipative part of the Lindblad equation~\eqref{eq:Lindblad} in the weak-decoherence regime, \(\gamma T\ll 1\):
\begin{equation}
\label{eq:lidbladexpansion}
    \widetilde\rho_{\mathcal S}(T) - \rho_{\mathcal S}(0)
    =
    -\frac{\gamma }{2}
    \sum_{\mu}
    \int_0^T \dd s\,
    [
    \widetilde S_\mu(s),
    [
    \widetilde S_\mu(s),
    \rho_{\mathcal S}(0)
    ]
    ]
    +O((\gamma T)^2),
\end{equation}
where
\begin{align}
    \widetilde S_\mu(s)
    =
    U_{\mathcal S}^\dagger(s)\hat S_\mu U_{\mathcal S}(s)\,,\qquad 
    \widetilde\rho_{\mathcal S}(t)
    =
    U_{\mathcal S}^\dagger(t)\rho_{\mathcal S}(t)U_{\mathcal S}(t). 
\end{align}
The rate \(\gamma\) can be extracted by matching the
second-order microscopic change to the coefficient multiplying the double
commutator in this weak-decoherence expansion. The higher-order terms in the Dyson series generate the $O((\gamma T)^2)$ terms in Eq.~\eqref{eq:lidbladexpansion}, so for the purposes of extracting \(\gamma\), we may truncate them. We then further approximate the bath as initially being in an infinite-temperature state
\begin{align}\label{eq:infinite-temp}
\widetilde\rho_{\mathcal S\mathcal P}(0)
=
\rho_{\mathcal S}(0)\otimes
\frac{\mathbb I_{\mathcal P}}{d_{\mathcal P}},
\end{align}
 but do not invoke the Born approximation in the usual sense. 

By Eq.~\eqref{eq:infinite-temp}, the first-order term vanishes after tracing over \(\mathcal P\), since
\(\langle \hat K_\mu\rangle=0\). Thus, tracing out the bath yields
\begin{align}
\widetilde\rho_{\mathcal S}(T) -  \rho_{\mathcal{S}}(0)
&\approx
-\operatorname{tr}_{\mathcal P}
\int_0^T \dd t_1
\int_0^{t_1}\dd t_2\,
\left[
\widetilde H_{\mathcal{SP}}^{(\Omega)}(t_1),
\left[
\widetilde H_{\mathcal{SP}}^{(\Omega)}(t_2),
\rho_{\mathcal S}(0)\otimes
\frac{\mathbb I_{\mathcal P}}{d_{\mathcal P}}
\right]
\right]
\nonumber\\
&=
-\varepsilon^2
\sum_{\mu,\nu}
\int_0^T \dd t_1
\int_0^{t_1}\dd t_2\,
\operatorname{tr}_{\mathcal P}
\left(
\left[
\Delta\widetilde S_\mu(t_1)\otimes \widetilde K_\mu(t_1),
\left[
\Delta\widetilde S_\nu(t_2)\otimes \widetilde K_\nu(t_2),
\rho_{\mathcal S}(0)\otimes
\frac{\mathbb I_{\mathcal P}}{d_{\mathcal P}}
\right]
\right]
\right)
\\&=
-\varepsilon^2
\sum_{\mu,\nu}
\int_0^T dt_1
\int_0^{t_1}dt_2\,
\langle \widetilde K_\mu(t_1)\widetilde K_\nu(t_2)\rangle
\left[
\Delta \widetilde{S}_\mu(t_1),
\left[
\Delta \widetilde{S}_\nu(t_2),
\rho_{\mathcal S}(0)\right]\right]\\&=
-\varepsilon^2
\sum_{\mu,\nu}
\int_0^T \dd t_1
\int_0^{t_1}\dd t_2\,
B_{\mu\nu}(t_1,t_2)
\left[
\widetilde S_\mu(t_1),
\left[
\widetilde S_\nu(t_2),
\rho_{\mathcal S}(0)
\right]
\right],\label{Eq:secondorderdiff}
\end{align}
where we introduced the bath correlation function
\[
B_{\mu\nu}(t_1,t_2)
=
\left\langle 
\widetilde K_\mu(t_1)\widetilde K_\nu(t_2) 
\right\rangle =\left\langle
\hat K_\mu U_{\mathcal{P}}(t_1,t_2)
\,\hat K_\nu\,
U_{\mathcal{P}}^\dagger(t_1,t_2)\right\rangle,
\]
with $U_{\mathcal{P}}(t_1,t_2)=\mathcal T
e^{-i\varepsilon S\int_{t_2}^{t_1} d{s}\, \hat K_{\Omega(s)}}$ and $\hat K_\Omega=\sum_\sigma \Omega_\sigma \hat K_\sigma$.

We first study the case where there are no system dynamics, $H_{\mathcal{S}}=0$ (i.e., $\kappa=0=p$ in the kicked-top Hamiltonian).
\subsection{Diffusion rate without system dynamics}

In the case of no system dynamics, the trajectory $\Omega(t)=\Omega$ is constant, so we can write
\begin{equation}
    B_{\mu\nu}(\tau_1,\tau_2)= \langle \hat K_\mu e^{-i  (\tau_1-\tau_2) \hat K_\Omega} \hat K_\nu e^{i  (\tau_1-\tau_2) \hat K_\Omega} \rangle,
\end{equation}
where we introduced the scaled time $\tau\coloneqq \varepsilon S t$.
Further, in this case $U_\mathcal{S}(t)=\mathbb{I}$, so $\widetilde S_\mu(t)=\hat S_\mu$ is also constant and
\begin{equation}
   \widetilde\rho_{\mathcal S}(T) -  \rho_{\mathcal{S}}(0)
=-\frac{T}{2}\sum_{\mu\nu}\gamma_{\mu \nu}(T)
\left[
 \hat S_\mu,
\left[
 \hat S_\nu,
\rho_{\mathcal S}(0)
\right]
\right],
\end{equation}
where we introduced the directional rates
 \begin{align}
   \gamma_{\mu\nu}(T)
   &=
   \frac{2}{S^2T}
   \int_0^{\varepsilon S T}\dd \tau_1
   \int_0^{\tau_1}\dd \tau_2\,
   B_{\mu\nu}(\tau_1,\tau_2)
   =
   \frac{2\varepsilon}{S}
   \int_0^{\varepsilon S T}\dd \tau\,
   \left(
   1-\frac{ \tau}{\varepsilon S T}
   \right)
   B_{\mu\nu}(\tau).
\end{align}
In the last equality we substituted $\tau=\tau_1-\tau_2$ and performed one of the integrals. 

We now average over the disorder in $\hat K_\mu$ to obtain an average prediction over the ensemble, remaining agnostic for the moment as to whether the $\hat K_\mu$ are RHCs or GOE operators.
We compute
\begin{equation}
    \E[\gamma_{\mu\nu}(T)]=   
   \frac{2\varepsilon}{S}
   \int_0^{\varepsilon S T}\dd \tau\,
   \left(
   1-\frac{ \tau}{\varepsilon S T}
   \right)
    \E[B_{\mu\nu}(\tau)].
\end{equation}

Without loss of generality, take $\Omega=x$. The correlators with $\mu\neq\nu$ vanish, because they always contain a sole average $\E_{\hat K_\mu}[\hat K_\mu]=0$. For instance,
\begin{equation}
    \E[B_{xy}(\tau)]=  \E_{\hat K_x}[\langle \hat K_x e^{-i  \tau \hat K_x} \cancel{\E_{\hat K_y}[\hat K_y]} e^{i  \tau \hat K_x} ]\rangle=0.
\end{equation}
Further, for the direction parallel to $\Omega$, namely $\mu=\nu=x$, the commutator vanishes $\bigl[
 \hat S_x,
\bigl[
 \hat S_x,
\rho_{\mathcal S}(0)
\bigr]
\bigr]=0$, so the only relevant correlators are those in the transverse directions to $\Omega$, $\E[B_{yy}]$ and $\E[B_{zz}]$. These two have to be equal, since the ensemble of matrices $\hat K_\mu$ is isotropic, meaning that for any rotation matrix $R_{\mu\nu}$ the matrix $
\hat K'_\mu=\sum_{\nu}R_{\mu\nu}\hat K_\nu$ is distributed identically to $\hat K_\mu$ (this follows directly from expanding the matrices as sums of transpositions with Gaussian coefficients for the random Heisenberg couplings, and is also true if $\hat K_\mu$ is sampled from random-matrix ensembles, such as GOE). 
An additional average over rotations that leave $\Omega=x$ invariant would not change the result, so $\E[B_{yy}(\tau)]=\E[B_{zz}(\tau)]$. The single relevant quantity is

\begin{equation}
      \bar{\gamma}(T)=\E[\gamma_{zz}(T)]=\E[\gamma_{yy}(T)]
   =\frac{\varepsilon}{S}
   \bar J(\varepsilon S T),
\end{equation}
where
\begin{align}
\label{eq:J-def}
 \bar J(\tau)=\int_{-\tau}^{\tau}\dd \tau'\,
   \left(
   1-\frac{|\tau'|}{\tau}
   \right)
   \bar B( \tau')\,,\qquad
        \bar B( \tau)= \E[\langle \hat K_0 e^{-i   \tau \hat K_1} \hat K_0 e^{i \tau \hat K_1} \rangle],
\end{align}
with independent samples $\hat K_0,\hat K_1$.
Next, we compute the integral $\bar J(\tau)$ when $\hat K_0$ and $\hat K_1$ are RHCs and also when they are sampled from the GOE, in the limit $N\to \infty$.

\subsubsection{Correlator for random Heisenberg couplings}
\begin{figure}
    \centering
    \includegraphics[width=0.6\linewidth]{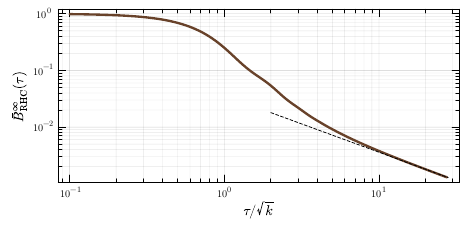}
    \vspace{-1.5em}
    \caption{Ensemble-averaged bath correlator for random Heisenberg couplings $\bar B_{\mathrm{RHC}}^\infty( \tau)$, constructed by computing $C_n$ up to $n=2500$ and truncating the rest of the Taylor series. The dashed line is a power-law decay $\sim\sqrt{k}/\tau$, i.e.\ $\sim\tilde\tau^{-1}$ with $\tilde\tau\coloneqq\tau/\sqrt{k}$.}
    \label{fig:S2}
\end{figure}
Expanding in nested commutators and using Lemma~\ref{lemm:higher-commutav} gives:
\begin{align}
    \bar B_{\mathrm{RHC}}^{\infty}( \tau)=\lim_{N\to\infty} \E[\langle \hat K_0 e^{-i   \tau \hat K_1} \hat K_0 e^{i \tau \hat K_1}\rangle]=\sum_{n=0}^{\infty} \frac{(-1)^n \tau^{2n}}{(2n)!} \lim_{N\to\infty} \E[\langle [\hat K_0, \hat K_1]_n [\hat K_0, \hat K_1]_n^\dagger \rangle] =\sum_{n=0}^{\infty}\frac{ (-1)^n \tau^{2n}}{(2n)!} \frac{C_n}{k  ^n}, 
\end{align}
where $C_n$ are defined in terms of a recursive expression \eqref{eq:Cn-higher-comm}. 
Performing the integral yields
\begin{equation}
    \bar J_{\mathrm{RHC}}^{\infty}(\tau)\coloneq\int_{-\tau}^{\tau} \dd \tau'\left(
   1-\frac{|\tau'|}{\tau}
   \right)
   \bar B^{\infty}_{\mathrm{RHC}}( \tau')=\sum_{n=0}^\infty \frac{(-1)^n C_n }{(2n)! k^n (1+3n+2n^2)}\,\tau^{1+2n}.
\end{equation}

We show a plot of $\bar B_{\mathrm{RHC}}^{\infty}( \tau)$ up to the value $\tilde\tau\coloneqq \tau/\sqrt{k}\sim 30$ in Fig.~\ref{fig:S2}. Up to the accessible times $\tilde\tau\sim 30$, $\bar B_{\mathrm{RHC}}^{\infty}$ appears to have a power-law tail $\sim~\tilde \tau^{-1}$. However,  a preliminary asymptotic analysis for much later times (not shown) suggests instead $\sim~(\tilde\tau\sqrt{\log\tilde \tau})^{-1}$. In any case, such a tail implies a logarithmic-type divergence for $\bar J_{\mathrm{RHC}}^{\infty}(\tau)$. As we will see below, this slow divergence will not prevent Markovianity of the bath once the system Hamiltonian is turned on.

\subsubsection{Correlator for GOE matrices}
In the case of GOE matrices, we can compute the correlator using free probability~\cite{Voiculescu1991}. Averages of products involving
different GOE matrices simplify much like products of independent random
variables, provided one keeps the operator ordering. Concretely, it allows us to decompose:
\begin{equation}
\bar{B}^\infty_{\mathrm{GOE}}(\tau) := \lim_{N\to \infty}\E[\langle \hat K_0 e^{-i\tau \hat K_1} \hat K_0 e^{i\tau \hat K_1} \rangle]
=\lim_{N\to \infty}\E[\langle \hat K_0^2\rangle] 
\E[\langle e^{-i\tau \hat K_1} \rangle]
\E[ \langle e^{i\tau \hat K_1} \rangle]
=\left(\frac{J_1(2\tau)}{\tau}\right)^2,
\label{eq:B-explicit}
\end{equation}
where $J_j(\tau)$ is the Bessel function of the first kind of order $j$, and we used $\lim_{N\to \infty}\E[\langle e^{i\tau \hat K}\rangle]={J_1(2\tau)}/{\tau} $. Integrating yields
\begin{equation}
\label{EQ:J:GOE}
    \bar{J}^\infty_{\mathrm{GOE}}(\tau) :=\int_{-\tau}^{\tau}\dd \tau'\,
   \left(
   1-\frac{|\tau'|}{\tau}
   \right)
   \bar{B}^\infty_{\mathrm{GOE}}( \tau')=\frac{1}{3\tau}[(3+16\tau^2)J_0(2\tau)^2-8\tau J_0(2\tau)J_1(2\tau)+(1+16\tau^2)J_1(2\tau)^2-3].
\end{equation}

\subsubsection{Correlator for random chiral couplings}
The 2-local RHC ensemble can be generalized to the 3-local random chiral
couplings (RCCs) introduced in Sec.~\ref{sec:rcc-moments}.  We do not use RCCs
in our numerical analysis, but their correlator is a useful comparison point:
interestingly, it takes a much simpler form than the RHC correlator, closer to
what we found for GOE matrices.

Expanding in nested commutators as in the RHC case and using
Lemma~\ref{lemm:higher-commutav-rcc} gives
\begin{align}
\label{eq:B-rcc-series}
    \bar B_{\mathrm{RCC}}^\infty(\tau)
    &\coloneq
    \lim_{N\to\infty}
    \E[\langle
    \hat K_0 e^{-i\tau \hat K_1}
    \hat K_0 e^{i\tau \hat K_1}
    \rangle]
    =
    \sum_{n=0}^{\infty}
    \frac{(-1)^n\tau^{2n}}{(2n)!}
    \frac{C_n^{(3)}}{k^n},
\end{align}
where $C_n^{(3)}$ is given by Eq.~\eqref{eq:Cn-rcc}.  In this
case the series can be summed explicitly.  To see this, use the Taylor series
\[
    \frac{J_1(2x)}{x}
    =
    \sum_{m=0}^{\infty}
    \frac{(-1)^m x^{2m}}{m!(m+1)!}.
\]
Comparing coefficients with Eq.~\eqref{eq:Cn-rcc}, one sees that the square of
this Taylor series reproduces Eq.~\eqref{eq:B-rcc-series} after setting
$x=\sqrt{2/k}\,\tau$.  Thus
\begin{equation}
\label{eq:B-rcc-goe-rescaled}
    \bar B_{\mathrm{RCC}}^\infty(\tau)
    =
    \left[
    \frac{
    J_1\!\left(2\sqrt{2/k}\,\tau\right)
    }{
    \sqrt{2/k}\,\tau
    }
    \right]^2
    =
    \bar B_{\mathrm{GOE}}^\infty\!\left(\sqrt{\frac{2}{k}}\,\tau\right),
\end{equation}
where in the last equality we used Eq.~\eqref{eq:B-explicit}.  Therefore the
integrated correlator in this case is the same as for GOE, but with a 
rescaling:
\begin{equation}
\label{eq:J-rcc-goe-rescaled}
    \bar J_{\mathrm{RCC}}^\infty(\tau)
    =
    \sqrt{\frac{k}{2}}\,
    \bar J_{\mathrm{GOE}}^\infty\!\left(\sqrt{\frac{2}{k}}\,\tau\right).
\end{equation}
Thus, unlike the RHC correlator in Fig.~\ref{fig:S2}, the RCC correlator has
the same integrable tail as the GOE correlator.  This suggests that the logarithmic-type divergence found for RHCs
is special to the 2-local couplings.

\begin{figure}
    \centering
    \includegraphics[width=1\linewidth]{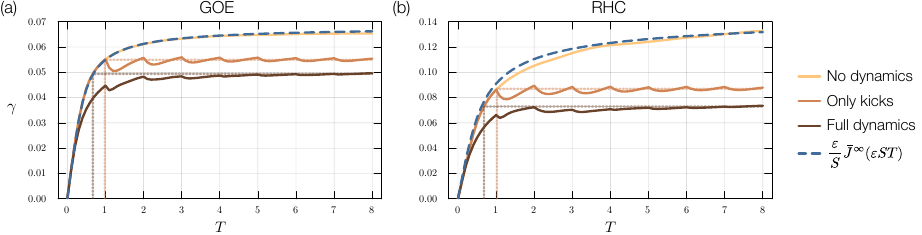}
   \caption{
Time-dependent rates extracted from exact disorder-sampled dynamics.
For each bath ensemble we sample \(M\) independent realizations of
\((\hat K_x,\hat K_y,\hat K_z)\), evolve the corresponding global states
\(\ket*{\psi^{(i)}(t)}\), and form
\(\rho_{\mathcal S}^{(i)}(t)=\tr_{\mathcal P}{\psi^{(i)}(t)}\).
We then evolve a single reduced density matrix \(\rho_\gamma(t)\) with the
same system dynamics and a piecewise-constant instantaneous Lindblad rate
\(\gamma_{\rm inst}(t)\). At each time step, \(\gamma_{\rm inst}(t)\) is chosen
to minimize the trace distance to all samples
\(M^{-1}\sum_{i=1}^M \tfrac{1}{2}\norm*{\rho_{\mathcal S}^{(i)}(t)-\rho_\gamma(t)}_1\). The plotted quantity is
the cumulative rate \(\gamma(T)=T^{-1}\int_0^T\,dt\gamma_{\rm inst}(t)\).
Panel (a) shows GOE bath operators ($N=150,M=10$), with a resulting average trace distance $\sim 0.02$, while panel (b) shows random Heisenberg
couplings (RHC) with average trace distance $\sim 0.07$  ($N=200,M=20$). In each panel we compare no system dynamics $(\kappa=0, p=0)$, kicks
only $(\kappa=0, p=\pi/2)$, and full kicked-top dynamics $(\kappa=4.2, p=\pi/2)$. The dashed blue curve \(J\) is the
prediction obtained from the bath correlator in the absence of system dynamics.
The dotted vertical lines indicate the effective cutoff times selected by the
only-kicks and full-dynamics curves ($k=4$, $\theta_0=0.5$, $\varphi_0=1.1$). 
}
    \label{fig:S3}
\end{figure}

\subsection{Diffusion rate with system dynamics}
Fig.~\ref{fig:S3} shows an empirical rate $\gamma(T)$ obtained by fitting a time-dependent instantaneous rate to the trajectory of the reduced density matrix and minimizing the trace distance to multiple disorder realizations. The prediction $J_{\mathrm{}}^\infty(TS\varepsilon)\frac{\varepsilon}{S}$ agrees well with $\gamma(T)$ for both GOE and RHCs in the case of no system dynamics. However, in the presence of a system Hamiltonian (either the full kicked top or only kicks), the empirical rate $\gamma(T)$ flattens and plateaus near $J_{\mathrm{}}^\infty(T_*S\varepsilon)\frac{\varepsilon}{S}$, for some $T_*$. To understand this effect, consider a single kick at time $T_*=1$ with \(p=\pi/2\), without the nonlinear term \((\kappa=0)\) in the kicked-top Hamiltonian. We will explain that, after the kick, $\gamma(T)$ becomes independent of the observation time $T$. It takes exactly the same form as the case with no dynamics, but with $T$ replaced by $T_*$:
\begin{equation}
\label{eq:formgamma}
   \gamma=\gamma(T_*)=\frac{\varepsilon}{S}J(\varepsilon S T_*). 
\end{equation}
This can be proven rigorously upon averaging for GOE random matrices $\hat K_\mu$, using free probability. We only prove the statement considering the immediate effect of a single kick, but  a similar statement can be made for later times. 
\begin{prop}
\label{th:onlykicksmemory}
Let \(\hat K_x,\hat K_y,\hat K_z\) be independent GOE matrices in $\mathcal{P}$ with $\E[\langle \hat K_\mu^2\rangle ]=1$. For $\kappa=0$, \(p=\pi/2\), and $k\geq 4$, consider the initial state
$\rho(0)=\dyad{\Omega(0)}$ with $\Omega(0)= x$, and a single kick applied at time \(T_*=1\). At the observation time \(T=2T_*\), i.e.~right before the next kick would occur, the empirical ensemble-averaged rate satisfies $\E[\gamma(T)]=\E[\gamma(T_*)]$:
\begin{align}
\E[\text{\textup{Eq.}~\eqref{Eq:secondorderdiff}}]=
-\frac{\E[\gamma(T_*)]}{2}
\sum_{\mu=x,y,z}
\int_0^{2T_*}\dd s\,
\left[
\widetilde S_\mu(s),
\left[
\widetilde S_\mu(s),
\rho_{\mathcal S}(0)
\right]
\right]
+
O\!\left(\frac{1}{S^2}\right),
\label{eq:reset-theorem}
\end{align}
\end{prop}

\begin{proof} 

We split the double integral in
Eq.~\eqref{Eq:secondorderdiff} into cases depending on whether $t_1,t_2$ are larger or smaller than $T_*$:
\begin{equation}
    \eqref{Eq:secondorderdiff}=\int_0^{T} \dd t_1\int_0^{t_1} \dd t_2 \,(\cdots)=\underbrace{\int_0^{T_*} \dd t_1\int_0^{t_1} \dd t_2 \,(\cdots)}_{I_1:\quad t_1,t_2\leq T_*} +\underbrace{\int_{T_*}^{2T_*} \dd t_1\int_{T_*}^{t_1} \dd t_2 \,(\cdots)}_{I_2:\quad t_1,t_2\geq T_*}+\underbrace{\int_{T_*}^{2T_*} \dd t_1\int_0^{T_*} \dd t_2 \,(\cdots)}_{I_3:\quad t_2\leq T_*\leq t_1}.
\end{equation}
The first two integrals $I_1$ and $I_2$ encompass the cases when $t_1$ and $t_2$ are not
separated by the kick and take the same form as the
equation with no system dynamics. The last integral $I_3$, in which $t_1$ and $t_2$ are separated by a kick, gives a vanishing contribution proportional to $1/S^2$. To compute these integrals, recall that the only terms that contribute to the dephasing correspond to the directions perpendicular to $\Omega(t)$. These are:
\begin{align}
\Omega(t)&=x&\implies& &\Omega_{\perp,1}&=y,&
\Omega_{\perp,2}&=z,& (t< T_*)
\\
\Omega(t)&=y&\implies& &\Omega_{\perp,1}&=x,&
\Omega_{\perp,2}&=z,& (t> T_*).
\label{eq:transverse-directions}
\end{align}

Let us compute the three integrals:

\begin{enumerate}
    \item[$I_1:$] This integral is exactly the same as the one we already performed in the case without dynamics, where the transverse directions to $\Omega(t\leq T_*)=x$ are $y$ and $z$, and the bath evolution is simply 
$U_{\mathcal{P}}(\tau_1,\tau_2)
=
e^{-i(\tau_1-\tau_2)\hat K_x},$ giving
\begin{equation}
    I_1=-T_*
\frac{\bar{\gamma}_*}{2}
\left[
 \hat S_{y},
\left[
 \hat S_{y},
\rho_{\mathcal S}(0)
\right]
\right]-T_*\frac{\bar{\gamma}_*}{2}\left[
 \hat S_{z},
\left[
 \hat S_{z},
\rho_{\mathcal S}(0)
\right]
\right],
\end{equation}
with $\bar{\gamma}_*=\E[\gamma(T_*)]$. 
\item[$I_2:$] Here, the center of the coherent state has been rotated to
\(\Omega(t)=y\), so the transverse directions are \(x\) and \(z\).
Thus, in the interaction picture,
\begin{align}
I_2
&=
-T_*\frac{\bar{\gamma}_*}{2}
\left[
\widetilde S_x(s),
\left[
\widetilde S_x(s),
\rho_{\mathcal S}(0)
\right]
\right]
-
T_*\frac{\bar{\gamma}_*}{2}
\left[
\widetilde S_z(s),
\left[
\widetilde S_z(s),
\rho_{\mathcal S}(0)
\right]
\right]
\\
&=
-T_*\frac{\bar{\gamma}_*}{2}
\left[
\hat S_y,
\left[
\hat S_y,
\rho_{\mathcal S}(0)
\right]
\right]
-
T_*\frac{\bar{\gamma}_*}{2}
\left[
\hat S_z,
\left[
\hat S_z,
\rho_{\mathcal S}(0)
\right]
\right],
\qquad T_*<s<2T_*,
\end{align}
yielding the same result as $I_1$.

\item[$I_3:$]  The bath evolution from \(t_2\) to \(t_1\) is
$U_{\mathcal{P}}(t_1,t_2)
=
e^{-i\tau_1\hat K_y}
e^{-i\tau_2\hat K_x},$
where we have introduced shifted scaled times $\tau_1=\varepsilon S(t_1-T_*)$ and $\tau_2=\varepsilon S(T_*-t_2)$. Expanding $I_3$ and changing variables gives four terms, corresponding to the pairs of possible transverse directions before and after the kick:
\begin{align}
        I_3=-\frac{1}{S^2} \int_{0}^{\varepsilon ST_*} \dd \tau_1 \int_{0}^{\varepsilon ST_*}\dd \tau_2 \Big(&\bar{B}_{xz}(\tau_1,\tau_2)[\widetilde S_{x}(t_1) ,[\widetilde S_z(t_2),\rho_{\mathcal{S}}(0)]] + \bar{B}_{zz}(\tau_1,\tau_2)[\widetilde S_{z}(t_1) ,[\widetilde S_z(t_2),\rho_{\mathcal{S}}(0)]] \\
&+ \bar{B}_{xy}(\tau_1,\tau_2)[\widetilde S_{x}(t_1) ,[\widetilde S_y(t_2),\rho_{\mathcal{S}}(0)]]+ \bar{B}_{zy}(\tau_1,\tau_2)[\widetilde S_{z}(t_1) ,[\widetilde S_y(t_2),\rho_{\mathcal{S}}(0)]] \Big).
\end{align}
The $xz$ and $zy$ correlators vanish because they contain a single $\hat K_z$, which averages to zero (since $\hat K_z$ does not appear in the evolution operator). The remaining terms can be computed using free probability, up to $O(d_{\mathcal{P}}^{-1})=O(1/S^2)$ corrections:
\begin{align}
 \bar{B}_{zz}(\tau_1,\tau_2)&=\E\left[\langle
\hat K_z
e^{-i\tau_1 \hat K_y}
e^{-i\tau_2\hat K_x}
\hat K_z
e^{i\tau_2\hat K_x}
e^{i\tau_1 \hat K_y}
\rangle
\right]
\nonumber \\&= \E\!\left[
\langle
e^{-i\tau_1 \hat K_y}\rangle\right] \E\!\left[\langle
e^{-i\tau_2\hat K_x}\rangle\right] \E\!\left[\langle
e^{i\tau_2\hat K_x}\rangle\right]
\E\!\left[\langle
e^{i\tau_1 \hat K_y}\rangle\right]
=
\bar{B}_{\mathrm{GOE}}(\tau_1 )
\bar{B}_{\mathrm{GOE}}(\tau_2).
\end{align}
The resulting integral is bounded independently of $S$, 
\begin{align}
& \int_{0}^{\varepsilon ST_*} \hspace{-1.5em}\dd \tau_1 \int_{0}^{\varepsilon ST_*}\hspace{-1.5em}\dd \tau_2  \bar{B}_{zz}(\tau_1,\tau_2) =\left(\int_{0}^{\varepsilon S T_*} \hspace{-1.5em}\dd \tau
\bar{B}_{\mathrm{GOE}}(\tau)\right)^2
=\left(\int_{0}^{\varepsilon S T_*} \hspace{-1.5em}\dd \tau \left(\frac{J_1(2\tau)}{\tau}\right)^2\right)^2
\leq \left(\int_{0}^{\infty}\hspace{-0.5em} \dd \tau \left(\frac{J_1(2\tau)}{\tau}\right)^2\right)^2
=\left(\frac{8}{3\pi}\right)^2.
\end{align}
 The same analysis applies to $B_{xy}$:
\begin{align}
\bar{B}_{xy}(\tau_1,\tau_2)&=\E\!\left[
\left\langle
\hat K_x e^{-i\tau_1\hat K_y}e^{-i\tau_2\hat K_x}
\hat K_y e^{i\tau_2\hat K_x}e^{i\tau_1\hat K_y}
\right\rangle
\right]
=
2\,
\E[\langle e^{-i\tau_1\hat K_y}\rangle]\,
\E[\langle e^{i\tau_2\hat K_x}\rangle]\,
\E[\langle \hat K_x e^{-i\tau_2\hat K_x}\rangle]\,
\E[\langle \hat K_y e^{i\tau_1\hat K_y}\rangle]
\nonumber\\
& =
2\,
{B_{1/2}(\tau_1)}{B_{1/2}(\tau_2)}
\left(i\frac{\dd}{\dd \tau_2}{B_{1/2}(\tau_2)}\right)
\left(-i\frac{\dd}{\dd \tau_1}{B_{1/2}(\tau_1)}\right),
\end{align}
where we used  \(\E[\langle X e^{-i\tau X}\rangle]
=
i\,\frac{\dd}{\dd\tau}\E[\langle e^{-i\tau X}\rangle]\) and \(B_{1/2}(\tau)\coloneqq\E[\langle e^{i\tau \hat K}\rangle]=J_1(2\tau)/\tau\). Integrating by parts gives the same $O(1)$ bound, using $B_{1/2}(\tau)^2=\bar{B}_{\mathrm{GOE}}(\tau)$. With the overall prefactor \(1/S^2\), this gives \(I_3=O(1/S^2)\), completing the proof of Eq.~\eqref{eq:reset-theorem}. \qedhere
\end{enumerate}
 \end{proof}

As confirmed by the numerics in Fig.~\ref{fig:S3}, in the kick-only dynamics, $\gamma(T)$ becomes time independent exactly at the first kick time, $T_*=1$. When the nonlinear term is also included,  $\gamma(T)$ instead deviates smoothly from the integrated correlator and plateaus near $\frac{\varepsilon}{S}J(\varepsilon S T_*)$, with a slightly smaller fitted value $T_*\approx 0.6$ for the parameters used here. The parameter $T_*$ is simply a characteristic timescale of the system dynamics. We remark that its precise value is not necessary to predict the effective rate $\gamma$ in the thermodynamic limit, as the relative error between $\bar{J}^{\infty}(\varepsilon T_*S)$ and $\bar{J}^{\infty}(\varepsilon T_*'S)$ for constant $T_*$ and $T_*'$ vanishes in the limit $N\to \infty$ for both the correlators of RHCs and GOE. In both Fig.~2 of the main text (RHCs) and Fig.~\ref{fig:S4} (GOE) we find excellent agreement between the predicted rate of Eq.~\eqref{eq:formgamma} and the empirical value from numerics, across a wide range of system sizes and disorder strengths $\varepsilon$.

\begin{figure}
    \centering
    \includegraphics[width=\textwidth]{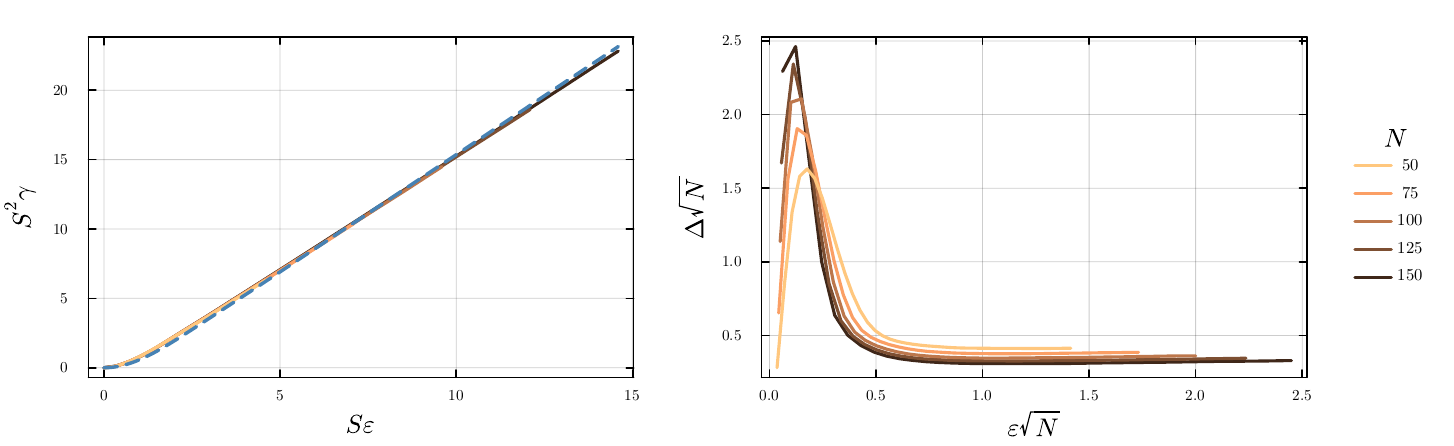}
    \caption{Effective Lindblad dynamics for couplings sampled from GOE. (a) Optimized rate $\gamma$ obtained by minimizing average trace distance $\Delta$ to exact quantum dynamics, with the permutation-sector couplings $\hat K_\mu$ sampled from the GOE. The dashed line is $\frac{\varepsilon}{S}
   \bar J^\infty_{\mathrm{GOE}}(\varepsilon S T),$ [see Eq.~\eqref{EQ:J:GOE}]. The vertical axis is scaled by $S^2$ and the horizontal axis by $S$, which collapses all curves. (b) Trace distance of the optimization. Axes are scaled by $\sqrt{N}$. All plots are averaged over 20 random-matrix realizations ($k=4$, $\kappa=4.2$, $p=\pi/2$, $\theta_0=0.5$, $\varphi_0=1.1$).}
    \label{fig:S4}
\end{figure}

 \subsection{Intuition}
 To understand how   stronger system dynamics can reduce the effective bath-induced decoherence, we can get intuition from a simplified model closely related to the Coleman-Hepp model~\cite{Hepp1972}. Consider a system spin moving through a one-dimensional bath of spins with spatial density $\chi$, where the two states of the system spin generate opposite local magnetic fields on each bath spin. If the system moves with speed $\nu$ and the field profile has integrated strength $G=\int g(u)\,du$, then each bath spin is only rotated by an angle $G/\nu$, so the overlap of the two conditional bath states is $\cos(G/\nu)$.  Since the system encounters $\chi \nu$ bath spins per unit time, the system's coherence decays at rate
\begin{equation}
\Gamma(\nu)=-\chi \nu\log|\cos(G/\nu)|\simeq \frac{\chi G^2}{2\nu}
\end{equation}
for $G/\nu\ll 1$.  Thus faster system motion increases the number of bath spins sampled per unit time but decreases the distinguishability imprinted on each one more strongly, making the net decoherence rate smaller. 

In the present model, the analogous role is played by the system Hamiltonian itself: rapid motion of the collective spin changes the interaction-picture system operators and cuts off the effective memory integral, preventing the Markov-limit rate from diverging even when the static bath correlator has a slowly decaying tail. This effect is reminiscent of dynamical decoupling \cite{Viola1999}, Zeno-bang-bang effects~\cite{Facchi2004}, and open-system motional narrowing \cite{Rnnburg2006}.
\section{Detailed proof of Theorem~1}
\label{sec:proofofth1}
Here we provide a more detailed proof of Theorem~1 in the main text. We recall the basic notation. We denote by $\hat\Omega=\dyad{\Omega}$ a spin coherent-state projector and use the measure $
    \dd\Omega
    =
    \frac{2S+1}{4\pi}\sin\theta\,\dd\theta\,\dd\varphi .$
We write $S_\mu(\Omega)=S\Omega_\mu$ 
and use the Poisson bracket $
    \{f,g\}
    =
    \bm S\cdot
    \left(
        \frac{\partial f}{\partial \bm S}
        \times
        \frac{\partial g}{\partial \bm S}
    \right),$ which satisfies $
    \{S_\mu,S_\nu\}
    =
\sum_{\lambda}\epsilon_{\mu\nu\lambda}S_\lambda$.

The following Lemma serves as a dictionary to transform commutators in the operator picture to Poisson brackets in the phase space picture.
\begin{lemma}[Dictionary from commutators to Poisson brackets]
\label{lem:coherent-state-dictionary}
Let $ \rho
    =
    \int \dd\Omega\,P(\Omega)\hat\Omega$.
Then
\begin{align}
\label{eq:dict-linear-comm}
    -i[\hat S_\mu,\rho]
    &=
    \int \dd\Omega\,\{S_\mu,P\}\hat\Omega ,
    \\
\label{eq:dict-double-comm}
    -[\hat S_\mu,[\hat S_\nu,\rho]]
    &=
    \int \dd\Omega\,\{S_\mu,\{S_\nu,P\}\}\hat\Omega ,
    \\
\label{eq:dict-square-comm}
    -i[\hat S_\mu^2,\rho]
    &=
    \int \dd\Omega\,
    \left[
        \{S_\mu^2,P\}
        +
        \frac{1}{S}
        \sum_{\nu,\lambda}
        \epsilon_{\mu\nu\lambda}
        \{S_\mu,\{S_\nu,S_\lambda P\}\}
    \right]\hat\Omega .
\end{align}
\end{lemma}

\begin{proof}
The commutator $-i[\hat S_\mu,\,\cdot\,]$ generates rotations of coherent-state
projectors about the $\mu$ axis, while the Poisson bracket $\{S_\mu,\,\cdot\,\}$
generates the corresponding classical rotation on the sphere. Then, integration by parts gives the equality
\begin{equation}
\label{eq:commutator-ibp-identity}
    -i\int \dd\Omega\,F(\Omega)[\hat S_\mu,\hat\Omega]
    =
    \int \dd\Omega\,\{S_\mu,F\}\hat\Omega,
\end{equation}
for any smooth function $F(\Omega)$.
 Equation~\eqref{eq:dict-linear-comm} follows by picking $F=P$. Now, for Eq.~\eqref{eq:dict-double-comm} note that
\begin{align}
\label{eq:second-dictionary-proof}
    -[\hat S_\mu,[\hat S_\nu,\rho]]
    =
    -i\left[\hat S_\mu,-i[\hat S_\nu,\rho]\right]
      =
    -i\left[
        \hat S_\mu,
        \int \dd\Omega\,\{S_\nu,P\}\hat\Omega
    \right]
    =
    -i \int \dd\Omega\,\{S_\nu,P\} [
        \hat S_\mu,
       \hat\Omega
    ]
    =
    \int \dd\Omega\,\{S_\mu,\{S_\nu,P\}\}\hat\Omega ,
\end{align}
where in the last equality we used $F=\{S_\nu,P\}$ in Eq.~\eqref{eq:commutator-ibp-identity}. 
We finally derive the identity for the quadratic term. For this, we show
\begin{equation}
\label{eq:anticommutator-projector-identity}
    \hat S_\mu\hat\Omega+\hat\Omega\hat S_\mu
    =
    2S_\mu\hat\Omega
    -
    \frac{i}{S}
    \sum_{\nu,\lambda}\epsilon_{\mu\nu\lambda}S_\lambda[\hat S_\nu,\hat\Omega].
\end{equation}
Indeed, by rotational invariance, it is enough to check Eq.~\eqref{eq:anticommutator-projector-identity} at the north pole $\hat z=\dyad{S,S}$. Using $\hat S_+\hat z=0$ and $\hat z\hat S_-=0$, one obtains
\begin{align}
\label{eq:anticommutator-projector-check}
    \hat S_x\hat z+\hat z\hat S_x
    &=
    -i[\hat S_y,\hat z],
    &
    \hat S_y\hat z+\hat z\hat S_y
    &=
    i[\hat S_x,\hat z],
    &
    \hat S_z\hat z+\hat z\hat S_z
    &=
    2S\hat z .
\end{align}
It follows that
\begin{align}
\label{eq:square-projector-commutator}
    [\hat S_\mu^2,\hat\Omega]
    &=
    [\hat S_\mu,\hat S_\mu\hat\Omega+\hat\Omega\hat S_\mu]
    =
    2S_\mu[\hat S_\mu,\hat\Omega]
    -
    \frac{i}{S}
    \sum_{\nu,\lambda}
    \epsilon_{\mu\nu\lambda}S_\lambda
    [\hat S_\mu,[\hat S_\nu,\hat\Omega]] .
\end{align}
Using Eq.~\eqref{eq:commutator-ibp-identity} with $F=2S_\mu P$ gives
\begin{align}
\label{eq:quadratic-leading-term}
    -i\int \dd\Omega\,2S_\mu P[\hat S_\mu,\hat\Omega]
    &=
    \int \dd\Omega\,\{S_\mu,2S_\mu P\}\hat\Omega
    =
    \int \dd\Omega\,\{S_\mu^2,P\}\hat\Omega .
\end{align}
For the remaining terms, note that in general
\begin{equation}
\label{eq:double-commutator-with-weight}
    \int \dd\Omega\,F[\hat S_\mu,[\hat S_\nu,\hat\Omega]]
    =
    -\int \dd\Omega\,\{S_\mu,\{S_\nu,F\}\}\hat\Omega ,
\end{equation}
so taking $F=S_\lambda P$ gives
\begin{equation}
\label{eq:quadratic-correction-term}
    -
    \frac{1}{S}
    \int \dd\Omega\,P
    \sum_{\nu,\lambda}
    \epsilon_{\mu\nu\lambda}S_\lambda
    [\hat S_\mu,[\hat S_\nu,\hat\Omega]]
    =
    \frac{1}{S}
    \int \dd\Omega
    \sum_{\nu,\lambda}
    \epsilon_{\mu\nu\lambda}
    \{S_\mu,\{S_\nu,S_\lambda P\}\}\hat\Omega .
\end{equation}
Combining Eqs.~\eqref{eq:quadratic-leading-term} and \eqref{eq:quadratic-correction-term} gives
Eq.~\eqref{eq:dict-square-comm}.
\end{proof}
To prove Theorem 1, we apply  Lemma~\ref{lem:coherent-state-dictionary} to the terms in the Lindblad equation \eqref{eq:Lindblad}, for $ \rho
    =
    \int \dd\Omega\,P(\Omega)\hat\Omega$, giving

\begin{equation}
    \int \dd\Omega\,\frac{\partial P}{\partial t}\hat\Omega
    =
    \int \dd\Omega\,\Bigg[
    \frac{\kappa}{2S+1}
    \Big(
        \{S_x^2,P\}
        +
        \frac{1}{S}
        \sum_{\nu,\lambda}
        \epsilon_{x\nu \lambda}
        \{S_x,\{S_\nu,S_\lambda P\}\}
    \Big) +
    p\sum_{n=1}^{\infty}\delta(t-n)\{S_z,P\}
    +
    \frac{\gamma}{2}
    \sum_{\mu}
    \{S_\mu,\{S_\mu,P\}\}
    \Bigg]\hat\Omega .
\end{equation}
This equation is satisfied whenever the coefficients in the integrands are equal, yielding an equation for $P$. Note that
\begin{align*}
\sum_{\mu,\nu}
\epsilon_{\sigma\mu\nu}
\{S_\sigma,\{S_\mu,S_\nu P\}\}
\nonumber
&=
\{S_\sigma^2,P\}
+
\sum_{\mu,\nu}
\epsilon_{\sigma\mu\nu}
\{S_\sigma,S_\nu\}\{S_\mu,P\}
+
\sum_{\mu,\nu}
\epsilon_{\sigma\mu\nu}
S_\nu\{S_\sigma,\{S_\mu,P\}\}
\nonumber\\
&=
\frac{3}{2}\{S_\sigma^2,P\}
+
\sum_{\mu,\nu}
\epsilon_{\sigma\mu\nu}
S_\nu\{S_\sigma,\{S_\mu,P\}\}
\nonumber\\
&=
\frac{3}{2}\{S_\sigma^2,P\}
+
\frac{1}{2}
\sum_{\mu,\nu}
\epsilon_{\sigma\mu\nu}S_\nu
\left(
\{S_\sigma,\{S_\mu,P\}\}
+
\{S_\mu,\{S_\sigma,P\}\}
+
\{\{S_\sigma,S_\mu\},P\}
\right)
\nonumber\\
&=
\frac{5}{4}\{S_\sigma^2,P\}
+
\frac{1}{2}
\sum_{\mu,\nu}
\epsilon_{\sigma\mu\nu}S_\nu
\left(
\{S_\sigma,\{S_\mu,P\}\}
+
\{S_\mu,\{S_\sigma,P\}\}
\right)
\nonumber\\
&=
\frac{5}{4}\{S_\sigma^2,P\}
+
\frac{1}{2}
\sum_{\mu,\nu,\lambda}
\left(
\delta_{\sigma\mu}\epsilon_{\sigma\nu\lambda}
+
\delta_{\sigma\nu}\epsilon_{\sigma\mu\lambda}
\right)
S_\lambda
\{S_\mu,\{S_\nu,P\}\},
\end{align*}
where the first equality is the product rule and \(\{S_\mu,S_\nu\}=\sum_\lambda\epsilon_{\mu\nu\lambda}S_\lambda\), the second uses the contraction of two Levi-Civita symbols, together with \(\{S_x^2+S_y^2+S_z^2,P\}=0\), and the third equality is a symmetrization using the Jacobi identity. Then,  the equation for $P$ becomes
\begin{align*}
    \frac{\partial P}{\partial t} &=\frac{\kappa}{2S+1}\left(1+\frac{5}{4S}\right)\{S_x^2,P\}
         +
    p\sum_{n=1}^{\infty}\delta(t-n)\{S_z,P\}\,\,+
    \\
    &\qquad+ 
    \frac{\kappa}{2S(2S+1)}
    \sum_{\mu,\nu,\lambda}
    \left(
    \delta_{x\mu}\epsilon_{x\nu\lambda}
    +\delta_{x\nu}\epsilon_{x\mu\lambda}
    \right)
    S_\lambda
    \{S_\mu,\{S_\nu,P\}\}+\frac{\gamma}{2}
    \sum_{\mu}
    \{S_\mu,\{S_\mu,P\}\}\\
    &=\{\widetilde H_\mathrm{KT}(t),P\}
    +
    \frac{1}{2}\sum_{\mu,\nu}\Lambda_{\mu\nu}
    \{S_\mu,\{S_\nu,P\}\},
\end{align*}
where $\widetilde{H}_{\mathrm{KT}(t)}=(\frac{\kappa}{2S}+\frac{3c}{4})S_x^2
         +
    p\sum_{n=1}^{\infty}\delta(t-n)S_z,$
with $c=\kappa/(S(2S+1))$ and $$\Lambda_{\mu\nu}
=
\gamma\delta_{\mu\nu}
+
c\sum_\lambda
\left(
\delta_{x\mu }\epsilon_{x\nu\lambda}
+
\delta_{x\nu }\epsilon_{x\mu\lambda}
\right)S_\lambda=
\begin{pmatrix}
\gamma & cS_z & -cS_y\\
cS_z & \gamma & 0\\
-cS_y & 0 & \gamma
\end{pmatrix}_{\mu\nu},$$
whose eigenvalues are $\gamma, \gamma\pm c {\sqrt{S_z^2+S_y^2}}$, which are positive if $\gamma\geq \abs{c}S$, as stated in Theorem 1.
\section{Classicality for \texorpdfstring{$\gamma \gg O(S^{-4/3})$}{gamma >> O(S^(-4/3))} with not-too-squeezed states}
\label{sec:NTS}

\providecommand{\Lspin}{\mathcal L_t}
\providecommand{\Lclass}{\mathcal L_t^{(\mathrm{cl})}}
\providecommand{\NTS}{\mathrm{NTS}}
In this section, we extend Theorem 1 of the main text to general Hamiltonians and weaker noise by allowing mixtures of squeezed spin states. Following the ideas of Refs.~\cite{hernandez2025ehrenfests,hernandez2024classical}, which study flat phase space, we prove an approximation by mixtures of {\it not-too-squeezed} (NTS) states.

Let $|\Omega\rangle$ be the spin coherent state centered at
$\Omega\in\mathbb S^2$. Choose a unit vector $u\perp\Omega$, set
$v=\Omega\times u$, and write $\hat S_a=\sum_\mu a_\mu\hat S_\mu$
for the spin along an axis $a$. We use the
\emph{two-axis-countertwisted squeezed states}~\cite{KitagawaUeda1993}
\begin{equation}
 |\Omega;\sigma\rangle
 =\exp\!\left[\frac{i\log\zeta}{4S}
 (\hat S_u\hat S_v+\hat S_v\hat S_u)\right]|\Omega\rangle,
 \qquad
 \sigma=\frac1{2S}
 \left(\zeta^{-1}uu^{\mathsf T}+\zeta vv^{\mathsf T}\right).
 \label{eq:quadratic-packet}
\end{equation}
Here $\Omega$ is the center about which squeezing is applied, and
$\zeta\geq1$ controls its strength, with $\zeta=1$ giving a coherent
state. At leading order, $u$ is the squeezed axis and $v$ the broadened
axis: $\sigma$ approximates the tangent covariance of the normalized
spin $\hat{\bm S}/S$, with curvature corrections
$O(S^{-2})$ at fixed $\zeta$. For $\zeta_{\mathrm{max}}\geq1$,
we consider the set $\NTS_{\zeta_{\mathrm{max}}}$ of labels  $(\Omega,\sigma)$ such that $1\leq\zeta\leq \zeta_{\mathrm{max}}$.

Consider a collective-spin Hamiltonian of the form
\begin{equation}
 \widehat H(t)=\sum_{a,b,c\geq0}
 \frac{\kappa_{a,b,c}(t)}{S^{a+b+c-1}}\,
 \mathrm{Sym}\bigl(\hat S_x^a\hat S_y^b\hat S_z^c\bigr),
 \label{eq:smooth-H}
\end{equation}
where the coefficients $\kappa_{a,b,c}(t)$ are real and
$\mathrm{Sym}$ averages over all orderings of the spin factors.\footnote{For example,
$\mathrm{Sym}(\hat S_x^2\hat S_y)
=(\hat S_x^2\hat S_y+\hat S_x\hat S_y\hat S_x+\hat S_y\hat S_x^2)/3$.}
Consider the Lindblad generator with isotropic noise with rate
$\gamma>0$,
\begin{equation}
 \Lspin[\rho]=-i[\widehat H(t),\rho]
 -\frac\gamma2\sum_{\mu=x,y,z}
 [\hat S_\mu,[\hat S_\mu,\rho]].
 \label{eq:quadratic-Lindblad}
\end{equation}

\begin{theorem}[Approximation by mixtures of NTS states]
\label{thm:quadratic-NTS}
Consider an initial coherent state $\rho(0)=|\Omega_0\rangle\langle\Omega_0|$, suppose $\gamma\geq KS^{-4/3}$, and set $\zeta_{\mathrm{max}}=\max\{1,K/(S\gamma)\}$. Then for all times $t\geq 0$ there is a normalized, nonnegative probability distribution
$p_t(\Omega,\sigma)$, supported on $\NTS_{\zeta_{\mathrm{max}}}$, such that
\begin{align}
 \widetilde\rho(t)&=\int_{\NTS_{\zeta_{\mathrm{max}}}}\dd\Omega\,\dd\sigma\,
 p_t(\Omega,\sigma)|\Omega;\sigma\rangle\langle\Omega;\sigma|,&&
  \label{eq:quadratic-main-bound}
 \norm{\rho(t)-\widetilde\rho(t)}_1
 \leq r\,t,\end{align}
 with an error rate $r\coloneqq 100(K+K_3)\max\{S^{-1/2},\,K^{3/2}S^{-2}\gamma^{-3/2}\}$, which is small when $\gamma\gg O(S^{-4/3})$.
\end{theorem}
Here the constants $K$ and $K_3$ are assumed to be finite, and given by
\begin{align}
 K&=\frac S2\sup_{t,\Omega}
 \bigl[\lambda_{\max}(A(t,\Omega))-\lambda_{\min}(A(t,\Omega))\bigr],&
  K_3&=\sup_t\sum_{n=3}^{\infty}n(n-1)(n-2)
  \sum_{a+b+c=n}|\kappa_{a,b,c}(t)|,
 \label{eq:smooth-A-K}\\
 A_{\mu\nu}(t,\Omega)&=\frac1S\left[
 \frac{\partial^2 h(\Omega,t)}{\partial\Omega_\mu\partial\Omega_\nu}
 -\frac{\delta_{\mu\nu}}3\sum_\lambda
 \frac{\partial^2 h(\Omega,t)}{\partial\Omega_\lambda^2}
 \right]&h(\Omega,t)&=\sum_{a,b,c\geq 0}{\kappa_{a,b,c}(t)} \Omega_x^a\Omega_y^b\Omega_z^c. \nonumber
\end{align}
Here the components $\Omega_\mu$ are treated as independent when
differentiating $h$. The quantities
$\lambda_{\max}$ and $\lambda_{\min}$ are the largest and smallest
eigenvalues of the real symmetric $3\times3$ matrix $A(t,\Omega)$. The integral over $\NTS_{\zeta_{\mathrm{max}}}$ is defined in terms of the measure $\dd\sigma=\dd\zeta\,\dd\varphi/\pi$,
where $\varphi\in[0,\pi)$ is the angle of $u$ in the plane tangent to the sphere at $\Omega$. For the kicked top of the main text, $K=S|\kappa|/(2S+1)$ and $K_3=0$.

As will be apparent below, the evolution of $p_t(\Omega,\sigma)$ constructed in the proof  has a direct classical
interpretation: the motion of the centers $\Omega$ and deformation of the covariances $\sigma$ of the squeezed states  follow the 
dynamics of the classical Hamiltonian $h(\Omega,t)$, with an added diffusion term, resulting in the emergence of classical Fokker–Planck dynamics. A precise statement capturing this can be phrased as is done in Refs.~\cite{hernandez2024classical,hernandez2025threshold} for flat phase space (we do not present it here).

\addcontentsline{toc}{subsection}{Proof of Theorem~\ref{thm:quadratic-NTS}}
\begin{proof}[Proof of Theorem~\ref{thm:quadratic-NTS}]
Our proof follows the same construction as Ref.~\cite{hernandez2025ehrenfests}. We perform a local quadratic approximation, which we use to construct a distribution over the space of labels $(\Omega,\sigma)$ so that the corresponding mixture approximately follows the quantum evolution.

At each center $\Omega$, we use the quadratic Taylor approximation
\begin{align}
 \widehat H^{(2)}_\Omega
 ={}&S h(\Omega,t)
 +\sum_\mu\frac{\partial h(\Omega,t)}{\partial\Omega_\mu}(\hat S_\mu-S\Omega_\mu) +\frac1{4S}\sum_{\mu,\nu}\frac{\partial^2 h(\Omega,t)}{\partial\Omega_\mu\partial\Omega_\nu}
 \bigl[(\hat S_\mu-S\Omega_\mu)(\hat S_\nu-S\Omega_\nu)
 +(\hat S_\nu-S\Omega_\nu)(\hat S_\mu-S\Omega_\mu)\bigr].
 \label{eq:smooth-Taylor}
\end{align}
Up to a scalar, this has the quadratic form
$\widehat H^{(2)}_\Omega=\sum_\mu b_\mu\hat S_\mu+
\frac14\sum_{\mu,\nu}A_{\mu\nu}
(\hat S_\mu\hat S_\nu+\hat S_\nu\hat S_\mu)$, with
$A$ from Eq.~\eqref{eq:smooth-A-K} and
$
 b_\mu(t,\Omega)=\frac{\partial h(\Omega,t)}{\partial\Omega_\mu}
 -\sum_\nu\Omega_\nu\frac{\partial^2 h(\Omega,t)}{\partial\Omega_\mu\partial\Omega_\nu}.$
The leading-order Hamiltonian motion changes the variables $\Omega$ and $\sigma$ as:
\begin{align}
 \dd\Omega=(\omega\times\Omega)\,\dd t,&&
 \dd\sigma=\bigl[\omega\times\sigma+(\omega\times\sigma)^{\mathsf T}
 +B\sigma+\sigma B\bigr]\,\dd t. && \textrm{(Only Hamiltonian)}
 \label{eq:quadratic-Hamiltonian-motion}
\end{align}
Here $\omega\times\sigma$ means taking the cross product with each
column of $\sigma$, and the angular velocity $\omega$ and the symmetric
stretching matrix $B$ are
\begin{align}
 \omega_\mu=b_\mu+S\sum_\nu A_{\mu\nu}\Omega_\nu
 +\frac S2A_{\Omega\Omega}\Omega_\mu,
&&
 B=SA_{uv}(uu^{\mathsf T}-vv^{\mathsf T})
 +\frac S2(A_{vv}-A_{uu})(uv^{\mathsf T}+vu^{\mathsf T}),
 \label{eq:quadratic-stretching}
\end{align}
where $A_{ab}=\sum_{\mu,\nu}a_\mu A_{\mu\nu}(t,\Omega)b_\nu$.
All coefficients are evaluated at the current center $\Omega$.
Along the squeezed axis, the variance changes at rate
$\dd\sigma_{uu}/\dd t=2B_{uu}\sigma_{uu}=A_{uv}/\zeta$.
Since $S|A_{uv}|\leq K$, this rate can be as negative as
$-K/(S\zeta)$ in the worst case. At $\zeta=\zeta_{\mathrm{max}}$, such a contraction
could take the squeezed state outside the NTS set.

We leverage the noise to counteract such a contraction. It spreads the centers of the squeezed states through random rotations, adding tangent
covariance at rate $\gamma\Omega_\perp$, with $\Omega_\perp=I-\Omega\Omega^{\mathsf T}$. Thus
$\gamma\geq K/(S\zeta_{\mathrm{max}})$ provides enough broadening
to counteract the squeezing. We divide the non-rotational covariance
change as follows:
\begin{equation}
 B\sigma+\sigma B+\gamma\Omega_\perp
 =\Sigma_{\mathrm{shape}}+\Sigma_{\mathrm{diff}},
 \qquad
 \Sigma_{\mathrm{diff}}
 =\gamma\Omega_\perp+\frac{K}{S\zeta_{\mathrm{max}}}
 \frac{4S^2\sigma^2-\Omega_\perp}{1-\zeta_{\mathrm{max}}^{-2}}.
 \label{eq:quadratic-split}
\end{equation}
The matrix $\Sigma_{\mathrm{diff}}$ is the full tangent covariance
rate from center diffusion; $\Sigma_{\mathrm{shape}}$, determined
by the first equality, is the rate of change of each state's
squeezing matrix at fixed $\Omega$. This is analogous to Eqs.~(D.9)--(D.11) of
Ref.~\cite{hernandez2025ehrenfests}.

To describe the resulting motion, we introduce some notation: $\partial/\partial\theta_\mu$ will
mean differentiation under a rotation of
$\Omega$ and $\sigma$ about the Cartesian $\mu$ axis, that is, for any function $F$ on $\NTS_{\zeta_{\mathrm{max}}}$,
\begin{equation}
 \frac{\partial F}{\partial\theta_\mu}
 =\left.\frac{\dd}{\dd\theta}
 F(R_\mu(\theta)\Omega,
 R_\mu(\theta)\sigma R_\mu(\theta)^{\mathsf T})\right|_{\theta=0}.
 \label{eq:quadratic-rotation-derivative}
\end{equation}
The expression $\sum_{\mu,\nu}(\Sigma_{\mathrm{shape}})_{\mu\nu}
\partial F/\partial\sigma_{\mu\nu}$ similarly means the derivative
at fixed $\Omega$ when the squeezing matrix changes at rate
$\Sigma_{\mathrm{shape}}$, along the family in
Eq.~\eqref{eq:quadratic-packet}.
We define the classical evolution on the set $\NTS_{\zeta_{\mathrm{max}}}$ by specifying the following generator, acting on a function $F(\Omega,\sigma)$.
\begin{align*}
 \Lclass[F]
 \coloneqq{}&\sum_\mu\omega_\mu\frac{\partial F}{\partial\theta_\mu}
 +\sum_{\mu,\nu}(\Sigma_{\mathrm{shape}})_{\mu\nu}
 \frac{\partial F}{\partial\sigma_{\mu\nu}}
 +\frac12\sum_{\mu,\nu}
 \left[\gamma\Omega_\mu\Omega_\nu
 +(\Sigma_{\mathrm{diff}})_{vv}u_\mu u_\nu
 +(\Sigma_{\mathrm{diff}})_{uu}v_\mu v_\nu\right]
 \frac{\partial^2F}{\partial\theta_\mu\,\partial\theta_\nu}.
\end{align*}
The three terms in this equation rotate the squeezed state, change its squeezing parameter, and add random
rotations, respectively. The three terms in the square brackets
describe random rotations about $\Omega$, $u$, and $v$, with rates
$\gamma$, $(\Sigma_{\mathrm{diff}})_{vv}$, and
$(\Sigma_{\mathrm{diff}})_{uu}$, respectively. Rotation about $\Omega$ turns the
squeezing axes without moving the center. The diffusion rates along
$u$ and $v$ are exchanged because rotation about $u$ displaces the
center along $v$, and conversely. The evolution of a distribution is then simply 
$\frac{\dd}{\dd t}\int p_tF=\int p_t\Lclass[F]$. For a squeezed state
$\rho_{\Omega,\sigma}=|\Omega;\sigma\rangle\langle\Omega;\sigma|$, the action $\Lclass[\rho_{\Omega,\sigma}]$ differentiates its dependence
on $\Omega,\sigma$, entry by entry in a fixed spin basis: it is not a
quantum superoperator acting on the density matrix. Our main technical result is that the evolution given by $\Lclass$ remains close to the actual Lindblad evolution by $\Lspin$.  
\addcontentsline{toc}{subsection}{Lemmata}
\begin{lemma}
\label{lem:quadratic-comparison}
If $\gamma\geq K/(S\zeta_{\mathrm{max}})$ and
$1<\zeta_{\mathrm{max}}\leq S^{1/3}$,
\begin{equation}
 \norm{\Lspin[\rho_{\Omega,\sigma}]
 -\Lclass[\rho_{\Omega,\sigma}]}_1
 \leq\frac{100(K+K_3)\zeta_{\mathrm{max}}^{3/2}}{\sqrt S}.
 \label{eq:quadratic-local-bound}
\end{equation}
\end{lemma}

We prove Lemma~\ref{lem:quadratic-comparison} below, but first use it to
finish the proof of Theorem~\ref{thm:quadratic-NTS}.
With $\zeta_{\mathrm{max}}=\max\{1,K/(S\gamma)\}$, the hypotheses of
Lemma~\ref{lem:quadratic-comparison} hold: $\gamma\geq K/(S\zeta_{\mathrm{max}})$
by construction, and $\zeta_{\mathrm{max}}\leq S^{1/3}$ is equivalent to
$\gamma\geq KS^{-4/3}$.
Evolve the initial delta distribution at
$(\Omega_0,(\Omega_0)_\perp/(2S))$ with the label dynamics above.
The corresponding mixture obeys
$\dot{\widetilde\rho}=\int\dd\Omega\,\dd\sigma\,
p_t(\Omega,\sigma)\Lclass[\rho_{\Omega,\sigma}]$. Let $\mathcal N_{t,s}$ be the quantum propagator from time $s$ to $t$.
To accumulate the local errors, differentiate the proposed mixture
after propagating it from $s$ to the final time: $\frac{\dd}{\dd s}\bigl(\mathcal N_{t,s}[\widetilde\rho(s)]\bigr)
 =\mathcal N_{t,s}\bigl[\dot{\widetilde\rho}(s)
 -\mathcal L_s[\widetilde\rho(s)]\bigr].$ 
Integrating from $s=0$ to $s=t$, and using
$\widetilde\rho(0)=\rho(0)$, gives Duhamel's formula,
\begin{equation}
 \rho(t)-\widetilde\rho(t)
 =\int_0^t\mathcal N_{t,s}
 \bigl[\mathcal L_s[\widetilde\rho(s)]
 -\dot{\widetilde\rho}(s)\bigr]\,\dd s.
 \label{eq:quadratic-Duhamel}
\end{equation}
Thus the final error is the accumulated local mismatch, propagated
by the quantum evolution. Trace-norm contractivity and the positivity
and normalization of $p_s$ give
\begin{align}
 \norm{\rho(t)-\widetilde\rho(t)}_1
 &\leq\int_0^t\dd s\int\dd\Omega\,\dd\sigma\,p_s(\Omega,\sigma)
 \norm{\mathcal L_s[\rho_{\Omega,\sigma}]
 -\mathcal L_s^{(\mathrm{cl})}[\rho_{\Omega,\sigma}]}_1
 \leq\frac{100(K+K_3)\zeta_{\mathrm{max}}^{3/2}}{\sqrt S}\,t=r\,t,
\end{align}
which proves Eq.~\eqref{eq:quadratic-main-bound}.
\end{proof}

To prove Lemma~\ref{lem:quadratic-comparison}, we require the following two intermediate lemmas.
\begin{lemma}[Bounds on spin operators acting on squeezed states]
\label{lem:spin-squeezed-estimates}
Define the operator $\hat a_\zeta\coloneqq(\hat S_u+i\zeta^{-1}\hat S_v)/\sqrt S$. For the squeezed states in Eq.~\eqref{eq:quadratic-packet} and $k$ an integer, the following bounds hold:
\begin{align}
 \norm{(S-\hat S_\Omega)^{k/2}|\Omega;\sigma\rangle}
 &\leq \sqrt{(2k-1)!!}\,(\zeta^{k/2}-1),
 && k\geq 1,
 \label{eq:quadratic-moments}\\
 \norm{(S+1-\hat S_\Omega)^{k/2}
       \hat a_\zeta|\Omega;\sigma\rangle}
 &\leq \sqrt{\frac{(2k+5)!!}{8}}\,
 \frac{(\zeta^2-1)\zeta^{(k-1)/2}}{S},
 && k\geq 0,
 \label{eq:quadratic-annihilator-bound}\\
 \norm{(\hat S_{\Theta}-S\Theta\cdot\Omega)^k|\Omega;\sigma\rangle}
 &\leq2^k\sqrt{(2k-1)!!}\,(S\zeta)^{k/2},
 && k\geq1,\quad \Theta\in\mathbb{S}^2.
 \label{eq:smooth-axis-moments}
\end{align}

\end{lemma}
\begin{proof}
In the basis $(S-\hat S_\Omega)|m\rangle=m|m\rangle$,
\begin{align}
 \frac{i(\hat S_u\hat S_v+\hat S_v\hat S_u)}{4S}
 &=\frac{(\hat S_u+i\hat S_v)^2-(\hat S_u-i\hat S_v)^2}{8S},\nonumber\\*
 \Bigl|\langle m+2|\frac{i(\hat S_u\hat S_v+\hat S_v\hat S_u)}{4S}|m\rangle\Bigr|
 &=\frac{\sqrt{(m+1)(m+2)(2S-m)(2S-m-1)}}{8S}
 \leq\frac{\sqrt{(m+1)(m+2)}}4\leq\frac{m+2}4.
 \label{eq:quadratic-squeezing-matrix-elements}
\end{align}
From Eq.~\eqref{eq:quadratic-squeezing-matrix-elements}, for $k\geq1$ and any vector $|\psi\rangle$,
\begin{align}
 &\Bigl|\langle \psi|\Bigl[\prod\nolimits_{j=1}^k(S+2j-1-\hat S_\Omega),
 \frac{i(\hat S_u\hat S_v+\hat S_v\hat S_u)}{4S}\Bigr]|\psi\rangle\Bigr|\nonumber\\*
 &\quad\leq\frac12\sum_{m=0}^{2S-2}(m+2)
 |\langle m|\psi\rangle\langle \psi|m+2\rangle|
 \Bigl[\prod\nolimits_{j=1}^k(m+2j+1)-\prod\nolimits_{j=1}^k(m+2j-1)\Bigr]\nonumber\\*
 &\quad\leq\frac k2\sum_{m=0}^{2S}|\langle m|\psi\rangle|^2
 \prod\nolimits_{j=1}^k(m+2j-1)
 \Bigl(\frac{m+2}{m+1}+\frac m{m+2k-1}\Bigr)
 \leq k\langle \psi|\prod\nolimits_{j=1}^k(S+2j-1-\hat S_\Omega)|\psi\rangle.
 \label{eq:quadratic-weighted-commutator}
\end{align}
The second inequality in Eq.~\eqref{eq:quadratic-weighted-commutator} uses $2|ab|\leq|a|^2+|b|^2$.
Integrating Eq.~\eqref{eq:quadratic-weighted-commutator} gives, for $1\leq\zeta'\leq\zeta$,
\begin{align}
 &\Bigl\|\Bigl[\prod\nolimits_{j=1}^k(S+2j-1-\hat S_\Omega)\Bigr]^{1/2}
 e^{\frac{i\log(\zeta/\zeta')}{4S}(\hat S_u\hat S_v+\hat S_v\hat S_u)}
 |\psi\rangle\Bigr\|
 \leq\Bigl(\frac\zeta\zeta'\Bigr)^{k/2}
 \Bigl\|\Bigl[\prod\nolimits_{j=1}^k(S+2j-1-\hat S_\Omega)\Bigr]^{1/2}|\psi\rangle\Bigr\|,
 \nonumber\\*
 &\Bigl\|\Bigl[\prod\nolimits_{j=1}^k(S+2j-1-\hat S_\Omega)\Bigr]^{1/2}
 |\Omega;\sigma\rangle\Bigr\|
 \leq\sqrt{(2k-1)!!}\,\zeta^{k/2}.
 \label{eq:quadratic-weighted-evolution}
\end{align}
For $k=0$, the first inequality in Eq.~\eqref{eq:quadratic-weighted-evolution} is an equality, with an empty product.

For $k\geq2$, Eq.~\eqref{eq:quadratic-moments} follows by integration from the coherent
state, writing $\sigma'$ for the covariance matrix with squeezing $\zeta'$:
\begin{align}
 \norm{(S-\hat S_\Omega)^{k/2}|\Omega;\sigma\rangle}
 &\leq\int_1^\zeta\frac{\dd\zeta'}{4S\zeta'}
 \Bigl\|\bigl[(S-\hat S_\Omega)^{k/2},
 \hat S_u\hat S_v+\hat S_v\hat S_u\bigr]|\Omega;\sigma'\rangle\Bigr\|\nonumber\\*
 &\leq\frac k2\int_1^\zeta\frac{\dd\zeta'}{\zeta'}
 \Bigl\|\Bigl[\prod\nolimits_{j=1}^k(S+2j-1-\hat S_\Omega)\Bigr]^{1/2}
 |\Omega;\sigma'\rangle\Bigr\|\nonumber\\*
 &\leq\frac k2\sqrt{(2k-1)!!}\int_1^\zeta\zeta'^{k/2-1}\,\dd\zeta'
 =\sqrt{(2k-1)!!}\,(\zeta^{k/2}-1).
 \label{eq:quadratic-moment-integration}
\end{align}
The second inequality in Eq.~\eqref{eq:quadratic-moment-integration} uses Eq.~\eqref{eq:quadratic-squeezing-matrix-elements} and
$(m+2)^{k/2}-m^{k/2}\leq k(m+2)^{k/2-1}$;
the last uses Eq.~\eqref{eq:quadratic-weighted-evolution}.
For $k=1$, Cauchy--Schwarz gives
\begin{align}
 \zeta\partial_\zeta\langle\Omega;\sigma|
 (S-\hat S_\Omega)|\Omega;\sigma\rangle
 &=-\frac{\operatorname{Re}\langle\Omega;\sigma|
 (\hat S_u+i\hat S_v)^2|\Omega;\sigma\rangle}{2S}\leq\Bigl[\langle\Omega;\sigma|(S-\hat S_\Omega)|\Omega;\sigma\rangle
 \bigl(1+\langle\Omega;\sigma|(S-\hat S_\Omega)|\Omega;\sigma\rangle\bigr)
 \Bigr]^{1/2},\nonumber\\*
 \norm{(S-\hat S_\Omega)^{1/2}|\Omega;\sigma\rangle}
 &\leq\frac{\sqrt\zeta-\zeta^{-1/2}}2\leq\sqrt\zeta-1,
 \label{eq:quadratic-half-moment}
\end{align}
where the last line of Eq.~\eqref{eq:quadratic-half-moment} integrates its first line from $\zeta=1$.
This proves Eq.~\eqref{eq:quadratic-moments}.

To prove Eq.~\eqref{eq:quadratic-annihilator-bound}, direct differentiation gives
\begin{align}
 2\zeta\partial_\zeta\bigl(\hat a_\zeta|\Omega;\sigma\rangle\bigr)
 ={}&\Bigl[\frac{i(\hat S_u\hat S_v+\hat S_v\hat S_u)}{2S}-1\Bigr]
 \hat a_\zeta|\Omega;\sigma\rangle+\frac{\hat a_\zeta^\dagger(S-\hat S_\Omega)
 +(S-\hat S_\Omega)\hat a_\zeta^\dagger}{2S}
 |\Omega;\sigma\rangle.
 \label{eq:quadratic-annihilator}
\end{align}
To bound the last term in Eq.~\eqref{eq:quadratic-annihilator}, we first compute
the nonzero matrix elements of $\hat S_{\Theta}$ in the basis
$(S-\hat S_\Omega)|m\rangle=m|m\rangle$,
\begin{align}
 \bigl|\langle m|(\hat S_{\Theta}-S\Theta\cdot\Omega)|m\rangle\bigr|
 &=|\Theta\cdot\Omega|\,m,
 &
 \bigl|\langle m+1|\hat S_{\Theta}|m\rangle\bigr|
 &=\frac{\sqrt{1-(\Theta\cdot\Omega)^2}}2
 \sqrt{(m+1)(2S-m)},
 \label{eq:smooth-axis-matrix-elements}
\end{align}
which are also used below. We then use
\begin{align}
 \hat a_{\zeta'}^\dagger
 &=\frac{1-\zeta'^{-1}}2\,\frac{\hat S_u+i\hat S_v}{\sqrt S}
 +\frac{1+\zeta'^{-1}}2\,\frac{\hat S_u-i\hat S_v}{\sqrt S},\label{eq:quadratic-annihilator-decomposition}\\*
 2m(2m-1)^2\prod\nolimits_{j=1}^k(m+2j-2)
 &\leq2(m+1)(2m+1)^2\prod\nolimits_{j=1}^k(m+2j)
 \leq8\prod\nolimits_{j=1}^{k+3}(m+2j-1).
 \label{eq:quadratic-source-coefficients}
\end{align}
Eqs.~\eqref{eq:smooth-axis-matrix-elements}, \eqref{eq:quadratic-annihilator-decomposition}, and~\eqref{eq:quadratic-source-coefficients}, together with convexity of the squared norm, give
\begin{align}
 &\Bigl\|\Bigl[\prod\nolimits_{j=1}^k(S+2j-1-\hat S_\Omega)\Bigr]^{1/2}
 \bigl[\hat a_{\zeta'}^\dagger(S-\hat S_\Omega)
 +(S-\hat S_\Omega)\hat a_{\zeta'}^\dagger\bigr]
 |\Omega;\sigma'\rangle\Bigr\|^2\nonumber\\*
 &\qquad \qquad \leq8\langle\Omega;\sigma'|
 \prod\nolimits_{j=1}^{k+3}(S+2j-1-\hat S_\Omega)
 |\Omega;\sigma'\rangle
 \leq8(2k+5)!!\,\zeta'^{k+3}.
 \label{eq:quadratic-annihilator-source}
\end{align}
The last inequality in Eq.~\eqref{eq:quadratic-annihilator-source} uses Eq.~\eqref{eq:quadratic-weighted-evolution}.
Finally, $\hat a_1|\Omega\rangle=0$ and Eqs.~\eqref{eq:quadratic-weighted-evolution}, \eqref{eq:quadratic-annihilator}, and \eqref{eq:quadratic-annihilator-source} give
\begin{align*}
 \norm{(S+1-\hat S_\Omega)^{k/2}\hat a_\zeta|\Omega;\sigma\rangle}
 &\leq\Bigl\|\Bigl[\prod\nolimits_{j=1}^k(S+2j-1-\hat S_\Omega)\Bigr]^{1/2}
 \hat a_\zeta|\Omega;\sigma\rangle\Bigr\|\\*
 &\leq\frac1S\sqrt{\frac{(2k+5)!!}{2}}\,
 \zeta^{(k-1)/2}\int_1^\zeta\zeta'\,\dd\zeta'
 =\sqrt{\frac{(2k+5)!!}{8}}\,
 \frac{(\zeta^2-1)\zeta^{(k-1)/2}}S,
\end{align*}
which proves Eq.~\eqref{eq:quadratic-annihilator-bound}.

Equations~\eqref{eq:smooth-axis-moments} and~\eqref{eq:smooth-axis-matrix-elements} give, for every integer $k\geq0$ and any vector $|\psi\rangle$,
\begin{align}
 &\Bigl\|\Bigl[\prod\nolimits_{\ell=1}^k
 (S+\ell-\hat S_\Omega)\Bigr]^{1/2}
 (\hat S_{\Theta}-S\Theta\cdot\Omega)|\psi\rangle\Bigr\|
 \nonumber\\*
 &\quad\leq\sqrt{2S}
 \bigl(|\Theta\cdot\Omega|+\sqrt{1-(\Theta\cdot\Omega)^2}\bigr)
 \Bigl\|\Bigl[\prod\nolimits_{\ell=1}^{k+1}
 (S+\ell-\hat S_\Omega)\Bigr]^{1/2}|\psi\rangle\Bigr\|
 \nonumber\\*
 &\quad\leq2\sqrt S
 \Bigl\|\Bigl[\prod\nolimits_{\ell=1}^{k+1}
 (S+\ell-\hat S_\Omega)\Bigr]^{1/2}|\psi\rangle\Bigr\|.
 \label{eq:smooth-weighted-axis}
\end{align}
The first inequality in Eq.~\eqref{eq:smooth-weighted-axis} applies Eq.~\eqref{eq:smooth-axis-matrix-elements} to the two off-diagonals and the diagonal
separately, using $m^2\leq2S(m+1)$; the second is Cauchy--Schwarz.
Iterating Eq.~\eqref{eq:smooth-weighted-axis} $k$ times and applying
Eq.~\eqref{eq:quadratic-weighted-evolution} gives
\begin{align*}
 \norm{(\hat S_{\Theta}-S\Theta\cdot\Omega)^k|\Omega;\sigma\rangle}
 &\leq2^kS^{k/2}
 \Bigl\|\Bigl[\prod\nolimits_{\ell=1}^{k}
 (S+\ell-\hat S_\Omega)\Bigr]^{1/2}|\Omega;\sigma\rangle\Bigr\|\\*
 &\leq2^kS^{k/2}
 \Bigl\|\Bigl[\prod\nolimits_{\ell=1}^{k}
 (S+2\ell-1-\hat S_\Omega)\Bigr]^{1/2}|\Omega;\sigma\rangle\Bigr\|
 \leq2^k\sqrt{(2k-1)!!}\,(S\zeta)^{k/2},
\end{align*}
which proves Eq.~\eqref{eq:smooth-axis-moments}.
\end{proof}

\begin{lemma}[Taylor remainder]
\label{lem:smooth-Taylor}
For the Hamiltonian in Eq.~\eqref{eq:smooth-H} and its quadratic
approximation in Eq.~\eqref{eq:smooth-Taylor},
\begin{equation}
 \norm{-i[\widehat H-\widehat H^{(2)}_\Omega,
 \rho_{\Omega,\sigma}]}_1
 \leq\frac{100K_3\zeta^{3/2}}{\sqrt S}.
 \label{eq:smooth-Taylor-bound}
\end{equation}
\end{lemma}
\begin{proof}
By Theorem~5.1 of Ref.~\cite{FriedlandLim2018}, any monomial of degree
$n=a+b+c$ admits a finite real decomposition of the following form:
\begin{equation}
 x^ay^bz^c=\sum_j w_j(\Theta_x^{(j)}x+\Theta_y^{(j)}y+\Theta_z^{(j)}z)^n,
 \qquad \Theta^{(j)}\in\mathbb{S}^2,
 \qquad \sum_j|w_j|\leq1.
 \label{eq:smooth-monomial-decomposition}
\end{equation}
For $n\geq3$, let $f_n(s,y)$ denote the second-order Taylor remainder of
$s^n$ about $y$, for which Taylor's theorem gives
\begin{align}
 f_n(s,y)&\coloneqq s^n-y^n-ny^{n-1}(s-y)
 -\tfrac12n(n-1)y^{n-2}(s-y)^2,
 \nonumber\\*
 |f_n(s,y)|&\leq\frac{n(n-1)(n-2)}6|s-y|^3,
 \qquad s,y\in[-1,1].
 \label{eq:smooth-power-Taylor}
\end{align}
Symmetric ordering preserves Eq.~\eqref{eq:smooth-monomial-decomposition}
as an operator identity, giving
$\mathrm{Sym}(\hat S_x^a\hat S_y^b\hat S_z^c)
=\sum_j w_j\hat S_{\Theta^{(j)}}^n$. 
Eqs.~\eqref{eq:smooth-H} and~\eqref{eq:smooth-Taylor} give
\begin{equation}
 \widehat H-\widehat H^{(2)}_\Omega
 =S\sum_{n=3}^{\infty}\sum_{a+b+c=n}\kappa_{a,b,c}(t)\sum_j w_j\,
 f_n\Bigl(\frac{\hat S_{\Theta^{(j)}}}S,\,\Theta^{(j)}\cdot\Omega\Bigr),
 \label{eq:smooth-remainder-identity}
\end{equation}
where $\Theta^{(j)}$ and $w_j$ depend on $(a,b,c)$. Therefore
\begin{align}
 \norm{-i[\widehat H-\widehat H^{(2)}_\Omega,\rho_{\Omega,\sigma}]}_1
 &\leq2\norm{(\widehat H-\widehat H^{(2)}_\Omega)|\Omega;\sigma\rangle}
 \nonumber\\*
 &\leq2S\sum_{n=3}^{\infty}\sum_{a+b+c=n}|\kappa_{a,b,c}(t)|\sum_j|w_j|\,
 \norm{f_n\Bigl(\frac{\hat S_{\Theta^{(j)}}}S,\,
 \Theta^{(j)}\cdot\Omega\Bigr)|\Omega;\sigma\rangle}
 \nonumber\\*
 &\leq\frac{2S}{6S^3}\sum_{n=3}^{\infty}n(n-1)(n-2)
 \sum_{a+b+c=n}|\kappa_{a,b,c}(t)|\sum_j|w_j|\,
 \norm{(\hat S_{\Theta^{(j)}}-S\Theta^{(j)}\cdot\Omega)^3
 |\Omega;\sigma\rangle}
 \nonumber\\*
 &\leq\frac1{3S^2}\sum_{n=3}^{\infty}n(n-1)(n-2)
 \sum_{a+b+c=n}|\kappa_{a,b,c}(t)|
 \sup_{\Theta\in\mathbb S^2}
 \norm{(\hat S_{\Theta}-S\Theta\cdot\Omega)^3|\Omega;\sigma\rangle}
 \nonumber\\*
 &\leq\frac{8\sqrt{15}}3\frac{K_3\zeta^{3/2}}{\sqrt S}
 \leq\frac{100K_3\zeta^{3/2}}{\sqrt S}.
 \label{eq:smooth-Taylor-chain}
\end{align}
The last inequality applies
Lemma~\ref{lem:spin-squeezed-estimates},
Eq.~\eqref{eq:smooth-axis-moments}, with $k=3$, together with the
definition of $K_3$ in Eq.~\eqref{eq:smooth-A-K}. 
\end{proof}
\begin{proof}[Proof of Lemma~\ref{lem:quadratic-comparison}]
Let $\mathcal L^{(2)}_{t,\Omega}$ be $\Lspin$ with $\widehat H$
replaced by $\widehat H^{(2)}_\Omega$. Eq.~\eqref{eq:smooth-Taylor-bound}
bounds the error in this replacement. We now compare
$\mathcal L^{(2)}_{t,\Omega}$ with $\Lclass$, keeping the Taylor
center fixed in the following calculation.
We start by subtracting the two generators. A rigid rotation gives
$\partial\rho_{\Omega,\sigma}/\partial\theta_\mu
=-i[\hat S_\mu,\rho_{\Omega,\sigma}]$, so the noise about $\Omega$
cancels exactly. The difference is
\begin{align}
 \mathcal L^{(2)}_{t,\Omega}[\rho_{\Omega,\sigma}]-\Lclass[\rho_{\Omega,\sigma}]
 ={}&-i\bigg[\widehat H^{(2)}_\Omega-\sum_\mu\omega_\mu\hat S_\mu,
 \rho_{\Omega,\sigma}\bigg]
 -\sum_{\mu,\nu}(\Sigma_{\mathrm{shape}})_{\mu\nu}
 \frac{\partial\rho_{\Omega,\sigma}}{\partial\sigma_{\mu\nu}}
 \nonumber\\*
 &+\frac12\sum_{\mu,\nu}
 \left[(\Tr(\Sigma_{\mathrm{diff}})-\gamma)(\Omega_\perp)_{\mu\nu}
 -(\Sigma_{\mathrm{diff}})_{\mu\nu}\right]
 [\hat S_\mu,[\hat S_\nu,\rho_{\Omega,\sigma}]].
 \label{eq:quadratic-generator-difference}
\end{align}

The quadratic Hamiltonian separates into rotation, tangent squeezing, and terms
containing longitudinal fluctuations:
\begin{align}
 \widehat H^{(2)}_\Omega={}&\sum_\mu\omega_\mu\hat S_\mu
 +\frac{A_{uu}-A_{vv}}4(\hat S_u^2-\hat S_v^2)
 +\frac{A_{uv}}2(\hat S_u\hat S_v+\hat S_v\hat S_u)
 +\widehat H_{\mathrm{rem}}\quad\text{(up to a scalar)},
 \label{eq:quadratic-H-expansion}\\*
 \widehat H_{\mathrm{rem}}={}&-\frac12\sum_{a=u,v}A_{a\Omega}
 \bigl[(S-\hat S_\Omega)\hat S_a+\hat S_a(S-\hat S_\Omega)\bigr]
 +\frac34A_{\Omega\Omega}(S-\hat S_\Omega)^2.
 \label{eq:quadratic-H-remainder-definition}
\end{align}

To evaluate the second term in Eq.~\eqref{eq:quadratic-generator-difference}, recall that $\varphi$ is the angle of the
squeezed axis $u$ in a fixed tangent frame at $\Omega$. The two rates
$\dot\zeta$ and $\dot\varphi$ describe only the shape change:
$\Sigma_{\mathrm{shape}}
=(\partial_\zeta\sigma)\dot\zeta+(\partial_\varphi\sigma)\dot\varphi$,
with $\Omega$ fixed. They let us replace the matrix derivative in
Eq.~\eqref{eq:quadratic-generator-difference} by
$\dot\zeta\,\partial_\zeta\rho_{\Omega,\sigma}
+\dot\varphi\,\partial_\varphi\rho_{\Omega,\sigma}$.
In the frame $(u,v)$, Eqs.~\eqref{eq:quadratic-packet} and \eqref{eq:quadratic-split} give

\begin{align}
 (\Sigma_{\mathrm{diff}})_{uu}
 &=\gamma-\frac{K}{S\zeta_{\mathrm{max}}}\frac{1-\zeta^{-2}}{1-\zeta_{\mathrm{max}}^{-2}},
 &
 (\Sigma_{\mathrm{diff}})_{vv}
 &=\gamma+\frac{K}{S\zeta_{\mathrm{max}}}\frac{\zeta^2-1}{1-\zeta_{\mathrm{max}}^{-2}}.
 \label{eq:quadratic-diffusion-frame}
\end{align}
Differentiating $\sigma$ in Eq.~\eqref{eq:quadratic-packet} and using
Eq.~\eqref{eq:quadratic-split} gives the shape rates
\begin{align}
 \frac{\dot\zeta}{2\zeta}
 &=-SA_{uv}-\frac{K}{\zeta_{\mathrm{max}}}\frac{\zeta-\zeta^{-1}}{1-\zeta_{\mathrm{max}}^{-2}},
 &
 \dot\varphi
 &=-\frac S2(A_{vv}-A_{uu})\frac{\zeta^2+1}{\zeta^2-1}.
 \label{eq:quadratic-shape-frame}
\end{align}
For $1\leq\zeta\leq\zeta_{\mathrm{max}}$,
$(\Sigma_{\mathrm{diff}})_{uu}\geq\gamma-K/(S\zeta_{\mathrm{max}})\geq0$
and $(\Sigma_{\mathrm{diff}})_{vv}\geq\gamma$.
At $\zeta=\zeta_{\mathrm{max}}$, $\dot\zeta/(2\zeta)=-SA_{uv}-K\leq0$, so shape
evolution cannot increase the squeezing beyond $\zeta_{\mathrm{max}}$.

Substituting Eqs.~\eqref{eq:quadratic-H-expansion}, \eqref{eq:quadratic-diffusion-frame}, and \eqref{eq:quadratic-shape-frame} into
Eq.~\eqref{eq:quadratic-generator-difference}, using
$\partial_\zeta\rho_{\Omega,\sigma}
=i[\hat S_u\hat S_v+\hat S_v\hat S_u,\rho_{\Omega,\sigma}]/(4S\zeta)$
from Eq.~\eqref{eq:quadratic-packet}, gives

\begin{align}
 \mathcal L^{(2)}_{t,\Omega}[\rho_{\Omega,\sigma}]-\Lclass[\rho_{\Omega,\sigma}]
 ={}&\underbrace{-i[\widehat H_{\mathrm{rem}},\rho_{\Omega,\sigma}]}_{(\star)}
 +\underbrace{\frac{iS(A_{vv}-A_{uu})}{2}
 \left[\frac{\hat S_u^2-\hat S_v^2}{2S}
 -\frac{\zeta^2+1}{\zeta^2-1}\hat S_\Omega,\rho_{\Omega,\sigma}\right]}_{(\star\star)}
 \nonumber\\*
 &+\underbrace{\frac{K(\zeta^2-1)}{2\zeta_{\mathrm{max}}(1-\zeta_{\mathrm{max}}^{-2})}
 \left(\frac{[\hat S_u,[\hat S_u,\rho_{\Omega,\sigma}]]}{S}
 -\frac{[\hat S_v,[\hat S_v,\rho_{\Omega,\sigma}]]}{S\zeta^2}
 +4\frac{\partial\rho_{\Omega,\sigma}}{\partial\zeta}\right)}_{(\star\star\star)}
 \label{eq:quadratic-residual}
\end{align}
We now use Eqs.~\eqref{eq:quadratic-moments} and \eqref{eq:quadratic-annihilator-bound} to bound the three terms in Eq.~\eqref{eq:quadratic-residual}.

\begin{align}
 \norm{(\star)}_1
 &\leq2\norm{\widehat H_{\mathrm{rem}}|\Omega;\sigma\rangle}
 \nonumber\\*
 &\leq\sum_{a=u,v}|A_{a\Omega}|
 \norm{[(S-\hat S_\Omega)\hat S_a+\hat S_a(S-\hat S_\Omega)]
 |\Omega;\sigma\rangle}
 +\frac32|A_{\Omega\Omega}|
 \norm{(S-\hat S_\Omega)^2|\Omega;\sigma\rangle}
 \nonumber\\*
 &\leq\frac{4\sqrt2K}{3\sqrt S}
 \left[\langle\Omega;\sigma|
 \bigl(8(S-\hat S_\Omega)^3+4(S-\hat S_\Omega)^2
 +6(S-\hat S_\Omega)+1\bigr)|\Omega;\sigma\rangle\right]^{1/2}
 +\frac{2K}{S}\norm{(S-\hat S_\Omega)^2|\Omega;\sigma\rangle}
 \nonumber\\*
 &\leq\frac{4\sqrt{278}}3\frac{K\zeta^{3/2}}{\sqrt S}
 +2\sqrt{105}\frac{K\zeta^2}{S}
 \leq\frac{24K\zeta^{3/2}}{\sqrt S}+\frac{21K\zeta^2}{S}.
 \label{eq:quadratic-H-remainder}
\end{align}
The first two inequalities in Eq.~\eqref{eq:quadratic-H-remainder} use the rank-one trace norm and
Eq.~\eqref{eq:quadratic-H-remainder-definition}.
The next uses Eq.~\eqref{eq:smooth-axis-matrix-elements} with $\Theta=u,v$,
$\sum_{a=u,v}|A_{a\Omega}|\leq\sqrt2\norm{A}_\infty$,
and $S\norm{A}_\infty\leq4K/3$, following from Eq.~\eqref{eq:smooth-A-K}.
The final bounds use Eq.~\eqref{eq:quadratic-moments} with $k=1,2,3,4$.

\begin{align}
 \norm{(\star\star)}_1
 &=\frac{S|A_{vv}-A_{uu}|}{2(1-\zeta^{-2})}
 \left\|\left[\hat a_\zeta^\dagger\hat a_\zeta
 +\frac{(1+\zeta^{-2})(S-\hat S_\Omega)^2
 -2\zeta^{-1}(S-\hat S_\Omega)}{2S},
 \rho_{\Omega,\sigma}\right]\right\|_1
 \nonumber\\*
 &\leq\frac{2K}{1-\zeta^{-2}}
 \left(\norm{\hat a_\zeta^\dagger\hat a_\zeta|\Omega;\sigma\rangle}
 +\frac{1+\zeta^{-2}}{2S}\norm{(S-\hat S_\Omega)^2|\Omega;\sigma\rangle}
 +\frac{\zeta^{-1}}S\norm{(S-\hat S_\Omega)|\Omega;\sigma\rangle}\right)
 \nonumber\\*
 &\leq\frac{2K}{1-\zeta^{-2}}
 \left(\sqrt2\norm{(S+1-\hat S_\Omega)^{1/2}
 \hat a_\zeta|\Omega;\sigma\rangle}
 +\frac{1+\zeta^{-2}}{2S}\norm{(S-\hat S_\Omega)^2|\Omega;\sigma\rangle}
 +\frac{\zeta^{-1}}S\norm{(S-\hat S_\Omega)|\Omega;\sigma\rangle}\right)
 \nonumber\\*
 &\leq\frac{2K}{S}
 \left[\frac{\sqrt{105}}2(2\zeta^2+1)+\frac{\sqrt3\,\zeta}{\zeta+1}\right]
 \leq(3\sqrt{105}+\sqrt3)\frac{K\zeta^2}{S}
 \leq\frac{33K\zeta^2}{S}.
 \label{eq:quadratic-angle-bound}
\end{align}
The equality in Eq.~\eqref{eq:quadratic-angle-bound} uses
$\hat S_u^2+\hat S_v^2+\hat S_\Omega^2=S(S+1)$;
scalar terms disappear inside the commutator.
The first inequality uses 
$S|A_{vv}-A_{uu}|\leq2K$, from Eq.~\eqref{eq:smooth-A-K}.
The second uses Eqs.~\eqref{eq:quadratic-annihilator-decomposition} and \eqref{eq:smooth-axis-matrix-elements}, which give
$\norm{\hat a_\zeta^\dagger|\psi\rangle}
\leq\sqrt2\norm{(S+1-\hat S_\Omega)^{1/2}|\psi\rangle}$.
The final bounds use Eq.~\eqref{eq:quadratic-moments} with $k=2,4$ and Eq.~\eqref{eq:quadratic-annihilator-bound} with $k=1$.

\begin{align}
 \norm{(\star\star\star)}_1
 &=\frac{K(\zeta^2-1)}
 {2\zeta_{\mathrm{max}}(1-\zeta_{\mathrm{max}}^{-2})}
 \left\|\hat a_\zeta^2\rho_{\Omega,\sigma}
 +\rho_{\Omega,\sigma}(\hat a_\zeta^\dagger)^2
 -\hat a_\zeta\rho_{\Omega,\sigma}\hat a_\zeta
 -\hat a_\zeta^\dagger\rho_{\Omega,\sigma}\hat a_\zeta^\dagger\right\|_1
 \nonumber\\*
 &\leq K\zeta
 \left(\norm{\hat a_\zeta^2|\Omega;\sigma\rangle}
 +\norm{\hat a_\zeta|\Omega;\sigma\rangle}
 \norm{\hat a_\zeta^\dagger|\Omega;\sigma\rangle}\right)
 \nonumber\\*
 &\leq\sqrt2K\zeta
 \left(\norm{(S+1-\hat S_\Omega)^{1/2}
 \hat a_\zeta|\Omega;\sigma\rangle}
 +\norm{\hat a_\zeta|\Omega;\sigma\rangle}
 \left[1+\norm{(S-\hat S_\Omega)^{1/2}|\Omega;\sigma\rangle}^2
 \right]^{1/2}\right)
 \nonumber\\*
 &\leq\frac{\sqrt{105}+\sqrt{15}}2
 \frac{K\zeta(\zeta^2-1)}S
 \leq\frac{8K\zeta^3}{S}.
 \label{eq:quadratic-squeezing-bound}
\end{align}
The equality in Eq.~\eqref{eq:quadratic-squeezing-bound} expands $\hat a_\zeta$ using Eq.~\eqref{eq:quadratic-annihilator-decomposition} and its adjoint, and uses
$\partial_\zeta\rho_{\Omega,\sigma}
=i[\hat S_u\hat S_v+\hat S_v\hat S_u,\rho_{\Omega,\sigma}]/(4S\zeta)$
from Eq.~\eqref{eq:quadratic-packet}.
The first inequality in Eq.~\eqref{eq:quadratic-squeezing-bound} uses the rank-one trace norm and
$(\zeta-\zeta^{-1})/(\zeta_{\mathrm{max}}-\zeta_{\mathrm{max}}^{-1})\leq1$.
The second uses Eq.~\eqref{eq:quadratic-annihilator-decomposition} and its adjoint, together with Eq.~\eqref{eq:smooth-axis-matrix-elements}, for $\hat a_\zeta$ and
$\hat a_\zeta^\dagger$.
The final bounds use Eq.~\eqref{eq:quadratic-moments} with $k=1$, Eq.~\eqref{eq:quadratic-annihilator-bound} with $k=0,1$, and
$1+(\sqrt\zeta-1)^2\leq\zeta$.

Combining Eqs.~\eqref{eq:quadratic-H-remainder}, \eqref{eq:quadratic-angle-bound}, \eqref{eq:quadratic-squeezing-bound}, and \eqref{eq:smooth-Taylor-bound} gives
\begin{align}
 \norm{\Lspin[\rho_{\Omega,\sigma}]-\Lclass[\rho_{\Omega,\sigma}]}_1
 &\leq
 \norm{-i[\widehat H-\widehat H^{(2)}_\Omega,\rho_{\Omega,\sigma}]}_1
 +\norm{\mathcal L^{(2)}_{t,\Omega}[\rho_{\Omega,\sigma}]
 -\Lclass[\rho_{\Omega,\sigma}]}_1
 \nonumber\\*
 &\leq K\left(\frac{24\zeta^{3/2}}{\sqrt S}
 +\frac{54\zeta^2+8\zeta^3}{S}\right)
 +\frac{100K_3\zeta^{3/2}}{\sqrt S}
 \leq\frac{100(K+K_3)\zeta_{\mathrm{max}}^{3/2}}{\sqrt S}.
 \label{eq:smooth-local-combination}
\end{align}
The last inequality in Eq.~\eqref{eq:smooth-local-combination} uses $\zeta\leq\zeta_{\mathrm{max}}$ and
$\zeta_{\mathrm{max}}^3\leq S$, proving
Eq.~\eqref{eq:quadratic-local-bound}. The calculation is
written for $\zeta>1$; the bound extends to coherent labels by
smoothness in $\sigma$.
\end{proof}

\section{Fidelity of records in the Lindblad regime}
\label{sec:fidelity}
In this section, we estimate the overlap between two branches specified by
sequences of spin-coherent-state projections. 
Using the effective Lindblad description, we obtain a
simple estimate in which the branch overlap is exponentially suppressed by the
time-integrated separation between the two histories. We then test this prediction numerically.

\subsection{Factorization in time from effective Lindblad}
The branch overlap can be expressed using the effective Lindblad description as
\begin{align}
    \abs{\braket{\phi_{\bm \Theta}}{\phi_{\bm \Omega}}}^2&=\frac{1}{\abs{c_{\bm \Omega}}^2\abs{c_{\bm \Theta}}^2}\abs{\tr(\hat \Omega_M U_M\cdots\hat{\Omega}_{2} U_{2} \hat \Omega_1 U_1 \dyad{\psi(0)} U_1^\dagger \hat{\Theta}_1 U_{2}^\dagger \hat{\Theta}_{2}\cdots U_M^\dagger\hat{\Theta}_M)}^2\\
    &\approx \frac{1}{\abs{c_{\bm \Omega}}^2\abs{c_{\bm \Theta}}^2}\abs{
    \tr(\hat \Omega_M \mathcal{N}_M\Big[\cdots\hat{\Omega}_{2} \mathcal{N}_2\Big[ \hat \Omega_1 \mathcal{N}_1\big[ \dyad{\Omega_0}{\Theta_0}\big]  \hat{\Theta}_1 \Big] \hat{\Theta}_{2}\cdots \Big]\hat{\Theta}_M)}^2
    \\&=  \frac{1}{\abs{c_{\bm \Omega}}^2\abs{c_{\bm \Theta}}^2} \abs{\braket{\Theta_M}{\Omega_M}}^2\prod_{m=1}^M \abs{\bra{\Omega_m}\mathcal{N}_m\big[\dyad{\Omega_{m-1}}{\Theta_{m-1}}\big]\ket{\Theta_m}}^2,
\end{align}
where $\mathcal{N}_m=\mathcal{T} \exp\int_{t_{m-1}}^{t_m} \dd t \mathcal{L}_t$. In this approximation, the effective Lindblad evolution applies not only to the initial state, but also after each intermediate projection into the coherent states $\hat{\Omega}_m$ and $\hat{\Theta}_m$.
Under the same Lindbladian approximation, the normalization coefficients can be expanded as
\begin{align}
    \abs{c_{\bm \Omega}}^2\approx \prod_{m=1}^M \bra{\Omega_m}\mathcal{N}_m\big[\hat{\Omega}_{m-1}\big]\ket{\Omega_m}, &&     \abs{c_{\bm \Theta}}^2\approx \prod_{m=1}^M \bra{\Theta_m}\mathcal{N}_m\big[\hat{\Theta}_{m-1}\big]\ket{\Theta_m}.
\end{align}

Thus, the fidelity of records factorizes in time:
\begin{equation}
\label{eq:fidelityofrecords}
      \abs{\braket{R_{\bm \Theta}}{R_{\bm \Omega}}}^2 = \frac{ \abs{\braket{\phi_{\bm \Theta}}{\phi_{\bm \Omega}}}^2}{ \abs{\braket{\Theta_M}{\Omega_M}}^2} \approx \prod_{m=1}^M \frac{\abs{\bra{\Omega_m}\mathcal{N}_m\big[\dyad{\Omega_{m-1}}{\Theta_{m-1}}\big]\ket{\Theta_m}}^2}{\bra{\Theta_m}\mathcal{N}_m\big[{\hat\Theta_{m-1}}\big]\ket{\Theta_m}\bra{\Omega_m}\mathcal{N}_m\big[\hat{\Omega}_{m-1}\big]\ket{\Omega_m}}.
\end{equation}

\subsection{Off-diagonal decay under Lindblad dynamics}
To estimate each factor in Eq.~\eqref{eq:fidelityofrecords}, we first present an elementary result for coherent states in flat phase space undergoing pure isotropic diffusion.

\begin{prop}
\label{prop:cohstatesL}
Let $\alpha,\beta,\alpha',\beta'\in\mathbb{C}$ be the centers of coherent states of a
harmonic oscillator, i.e., eigenstates of the annihilation operator $ \hat a=(\hat x+i\hat p)/{\sqrt{2\hbar}}$, $\hat a\ket{\alpha}=\alpha \ket{\alpha}$. Consider evolution under pure isotropic diffusion $\mathcal{N}=\exp(t\mathcal{L})$, with
$\mathcal{L}[\rho]
    =
    -
    \frac{\gamma}{2\hbar^2}
    \left(
        [\hat p,[\hat p,\rho]]
        +
        [\hat x,[\hat x,\rho]]
    \right)$. Then
\begin{equation}
    \frac{
    \abs{
    \bra{\alpha'}
    \mathcal N[\dyad{\alpha}{\beta}]
    \ket{\beta'}
    }^2
    }{
    \bra{\alpha'}
    \mathcal N[\dyad{\alpha}{\alpha}]
    \ket{\alpha'}
    \bra{\beta'}
    \mathcal N[\dyad{\beta}{\beta}]
    \ket{\beta'}
    }
    =
    \exp[
        -
        \frac{\gamma t/\hbar}{1+\gamma t/\hbar}
        \left(
            \abs{\alpha-\beta}^2
            +
            \abs{\alpha'-\beta'}^2
        \right)
    ].
\end{equation}
\end{prop}

\begin{proof}
The channel \(\mathcal N\) is simply a Gaussian average over displacements,
\begin{align}
\label{eq:gaussianint}
    \mathcal N[\rho]
    =
    \int_{\mathbb{C}}
    \frac{\dd\xi}{\pi\nu}
    \exp[-\frac{\abs{\xi}^2}{\nu}]
    D(\xi)\rho D(\xi)^\dagger,\qquad \text{where} \qquad
    D(\xi)=\exp(\xi\hat a^\dagger-\xi^*\hat a)
\end{align}
is the displacement operator and $\nu=\gamma t/\hbar$. To see this, evaluate the action on the displacement-operator basis $D(\eta)$,
\[
\begin{aligned}
\mathcal L[D(\eta)]
=
-\frac{\gamma}{2\hbar^2}
\left(
[\hat x,[\hat x,D(\eta)]]
+
[\hat p,[\hat p,D(\eta)]]
\right)
=
-\frac{\gamma}{2\hbar^2}
\left(
2\hbar(\mathrm{Re}\,\eta)^2
+
2\hbar(\mathrm{Im}\,\eta)^2
\right)D(\eta)
=
-\frac{\gamma}{\hbar}|\eta|^2D(\eta).
\end{aligned}
\]
Thus \(D(\eta)\) is an eigenoperator of \(\mathcal L\) with eigenvalue $-\nu\abs{\eta}^2/t$. Therefore
\[
\begin{aligned}
e^{t\mathcal L}[D(\eta)]
=
e^{-\nu|\eta|^2}D(\eta)
=
\int_{\mathbb C}
\frac{\dd^2\xi}{\pi\nu}
\exp[-\frac{|\xi|^2}{\nu}]
\exp(\eta^*\xi-\eta\xi^*) D(\eta)
&=
\int_{\mathbb C}
\frac{\dd^2\xi}{\pi\nu}
\exp[-\frac{|\xi|^2}{\nu}]
D(\xi)D(\eta)D(\xi)^\dagger .
\end{aligned}
\]
Inserting the off-diagonal operator \(\dyad{\alpha}{\beta}\) into Eq.~\eqref{eq:gaussianint} yields
\begin{align*}
    \bra{\alpha'}
    \mathcal N[\dyad{\alpha}{\beta}]
    \ket{\beta'}
    &=
    \int_{\mathbb{C}}
    \frac{\dd\xi}{\pi\nu}
    e^{-\abs{\xi}^2/\nu}
    \bra{\alpha'}D(\xi)\ket{\alpha}
    \bra{\beta}D(\xi)^\dagger\ket{\beta'}
\\
    &=
    \int_{\mathbb{C}}
    \frac{\dd^2\xi}{\pi\nu}
    \exp[
        -\frac{\abs{\xi}^2}{\nu}
        -\abs{\xi}^2
        -\frac{
            \abs{\alpha'}^2+\abs{\alpha}^2
            +\abs{\beta}^2+\abs{\beta'}^2
        }{2}
        +\alpha'^*\alpha+\beta^*\beta'
        +\xi(\alpha'^*-\beta^*)
        +\xi^*(\beta'-\alpha)
    ]
    \\
    &=
    \frac{1}{1+\nu}
    \exp[
        -\frac{
            \abs{\alpha'}^2+\abs{\alpha}^2
            +\abs{\beta}^2+\abs{\beta'}^2
        }{2}
        +\alpha'^*\alpha+\beta^*\beta'
        +\frac{\nu}{1+\nu}
        (\alpha'^*-\beta^*)(\beta'-\alpha)
    ],
\end{align*}
where we used the coherent-state identity 

\[
    \bra{z}D(\xi)\ket{w}
    =
    \exp[
        -\frac{\abs{z}^2+\abs{w}^2+\abs{\xi}^2}{2}
        +z^*w
        +\xi z^*
        -\xi^*w
    ],
\]
and performed the Gaussian integral.
Thus
\begin{equation}
    \abs{
    \bra{\alpha'}
    \mathcal N[\dyad{\alpha}{\beta}]
    \ket{\beta'}
    }^2
     =
    \frac{1}{(1+\nu)^2}
    \exp[
        -
        \frac{
            \abs{\alpha'-\alpha}^2
            +
            \abs{\beta'-\beta}^2
        }{1+\nu}
        -
        \frac{\nu}{1+\nu}
        \left(
            \abs{\alpha-\beta}^2
            +
            \abs{\alpha'-\beta'}^2
        \right)
    ].
\end{equation}
The diagonal terms follow by setting \(\beta=\alpha\) and
\(\beta'=\alpha'\):
\[
    \bra{\alpha'}
    \mathcal N[\dyad{\alpha}{\alpha}]
    \ket{\alpha'}
    =
    \frac{1}{1+\nu}
    \exp[
        -\frac{\abs{\alpha'-\alpha}^2}{1+\nu}
    ] \qquad\qquad   \bra{\beta'}
    \mathcal N[\dyad{\beta}{\beta}]
    \ket{\beta'}
    =
    \frac{1}{1+\nu}
    \exp[
        -\frac{\abs{\beta'-\beta}^2}{1+\nu}
    ],
\]
which combine to give the stated expression.
\end{proof}

 An analogous result holds for spin coherent states under pure isotropic diffusion $\mathcal{N}=\exp( t\sum_\mu \frac{\gamma}{2}[\hat{S}_\mu,[\hat{S}_\mu,\,\cdot\,]])$ in the spherical phase space. The isotropic diffusion on the sphere can be approximated by an average over Gaussian rotations, for short times $\sqrt{\gamma t}\ll 1$, and the overlap between coherent states is approximately Gaussian for large $S$, yielding the same result. The resulting formula follows by making the replacements
\[
    \hbar\rightarrow \frac{1}{S}\qquad\qquad\text{and}
    \qquad\qquad
    \abs{\alpha-\beta}^2
    \rightarrow
    \frac{S}{2}d^{(2)}(\Omega,\Theta)
\]
in Proposition~\ref{prop:cohstatesL},
where $d^{(2)}(\Omega,\Theta)=2(1-\cos d(\Omega,\Theta))\approx d(\Omega,\Theta)^2$. We obtain

\begin{align}
\label{eq:decayoffdiags}
    &\frac{
    \abs{
    \bra{\Omega'}
    \mathcal{N}[\dyad{\Omega}{\Theta}]
    \ket{\Theta'}
    }^2
    }{
    \bra{\Omega'}
    \mathcal{N}[\dyad{\Omega}{\Omega}]
    \ket{\Omega'}
    \bra{\Theta'}
    \mathcal{N}[\dyad{\Theta}{\Theta}]
    \ket{\Theta'}
    }
 \approx
    \exp[
        -
        \frac{\gamma S^2 t}{1+\gamma S t}
        \frac{1}{2}
        \left(
            d^{(2)}(\Omega,\Theta)
            +
            d^{(2)}(\Omega',\Theta')
        \right)
    ].
\end{align}

\subsection{Full fidelity of records}
Assuming equal temporal spacings $\Delta t=t_m-t_{m-1}$, Eq.~\eqref{eq:decayoffdiags} reduces Eq.~\eqref{eq:fidelityofrecords} to a simple expression for the fidelity of records in the case where there is no Hamiltonian term in the Lindbladian:
\begin{equation}
\label{eq:fidofreccords}
    \abs{\braket{R_{\bm \Theta}}{R_{\bm \Omega}}}^2  \approx \exp(-\Gamma D^{(2)}_{\bm \Omega,\bm \Theta}),
    \end{equation}where 
    \begin{align}
     \Gamma\coloneq \frac{\gamma S^2 \Delta t}{1+\gamma S\Delta t}&& \text{and}&& D^{(2)}_{\bm \Omega,\bm \Theta}\coloneqq\frac{1}{2}\sum_{m=1}^M \left[d^{(2)}(\Omega_{m-1},\Theta_{m-1})+d^{(2)}(\Omega_{m},\Theta_{m})\right].
\end{align}

Although the derivation above assumed that there was no Hamiltonian term in the effective Lindbladian, and hence no term acting solely on $\mathcal{S}$ in the global Hamiltonian, Eq.~\eqref{eq:fidofreccords} still describes the fidelity of branches in the presence of such a term. We show this numerically by considering the unitary evolution under the disordered kicked-top Hamiltonian, starting in the pure state $\ket{\psi(0)}= \ket{\theta_0, \varphi_0}^{\otimes (N-k)}\otimes [\tfrac{1}{\sqrt{2}}(\ket{01}-\ket{10})]^{\otimes k/2}$. After evolving for one period, $\Delta t=1$, we compute the Husimi distribution
$\
Q_1(\Omega)
=
\bra{\Omega}\tr_{\mathcal P}(\dyad{\psi(\Delta t)})\ket{\Omega},
$
and sample $\Omega_1\sim Q_1$. Note that the Husimi function is normalized with respect to the measure $\dd\Omega=\frac{2S+1}{4\pi}\sin\theta\,\dd\theta\,\dd\varphi$, i.e., $\int \dd\Omega\, Q_1(\Omega)=1$.

We then define the normalized branch conditioned on this outcome as
\[
c_{\Omega_1}\ket{\phi_{\Omega_1}}
=
(\dyad{\Omega_1}\otimes\mathds{1}_{\mathcal P})\ket{\psi(\Delta t)},
\qquad
|c_{\Omega_1}|^2=Q_1(\Omega_1).
\]
We evolve this branch for another period, define
\[
Q_2(\Omega)
=
\bra{\Omega}\tr_{\mathcal P}(\dyad{\phi_{\Omega_1}(\Delta t)})\ket{\Omega},
\]
sample $\Omega_2\sim Q_2$, and project again to obtain $\ket{\phi_{(\Omega_1,\Omega_2)}}$. Iterating this procedure for $M$ steps gives a sampled branch $\ket{\phi_{\bm\Omega}}$ with probability density
\[
|c_{\bm\Omega}|^2
=
Q_1(\Omega_1)Q_2(\Omega_2)\cdots Q_M(\Omega_M).
\]

\begin{figure}
    \centering
    \includegraphics[width=1\linewidth]{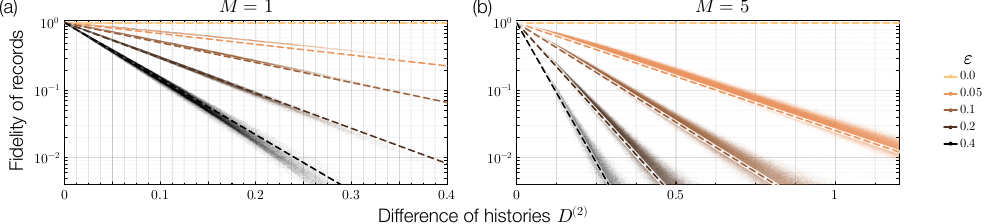}
    \caption{Fidelity of records for couplings sampled from GOE vs cumulative squared distance between histories $D^{(2)}$. Each dot is a pair of independently sampled branches. Dashed lines mark the decay $e^{-\Gamma D^{(2)}}$. (a) $M=1$, (b) $M=5$ ($N=100$,  $k=4$, $\kappa=4.2$, $p=\pi/2$, $\theta_0=0.5$, $\varphi_0=1.1$, $\Delta t=1$) }
    \label{fig:S5}
\end{figure}

Fig.~\ref{fig:S5} shows the fidelity of records for independently sampled pairs of branches with $\hat K_\mu$ drawn from the GOE. The numerical results agree well with Eq.~\eqref{eq:fidofreccords}. The same procedure is shown in Fig.~3 of the main text with $\hat K_\mu$ taken to be random Heisenberg couplings. In that case the scatter is larger, but Eq.~\eqref{eq:fidofreccords} still tracks the lower end of the distribution.

\bibliography{references}